\documentclass[11pt]{article}
\usepackage{graphicx}    %
\usepackage{amsmath}
\usepackage{amsfonts}
\usepackage{hyperref}
\usepackage{wrapfig}
\usepackage{subcaption}
\usepackage{amssymb}
\usepackage{amsthm}
\usepackage{thm-restate}
\usepackage{latexsym}
\usepackage[space]{grffile}
\usepackage{soul}
\usepackage{xcolor}
\definecolor{hl}{rgb}{1,1,.71}
\sethlcolor{hl}
\usepackage{colortbl}
\usepackage[margin=1in]{geometry}
\newtheorem{theorem}{Theorem}[section]
\newtheorem{lemma}[theorem]{Lemma}
\newtheorem{corollary}[theorem]{Corollary}
\newtheorem{proposition}[theorem]{Proposition}

\theoremstyle{definition}
\newtheorem{definition}{Definition}

\let\paragraph=\subparagraph

\def\player{a}
\def\Player{A}
\def\opponent{b}
\def\Opponent{B}
\def\escaper{h}
\def\Escaper{H}
\def\pursuer{z}
\def\Pursuer{Z}
\def\exit{x}
\def\Exit{X}

\def\domain{\operatorname{dom}}

{\makeatletter
 \gdef\xxxmark{%
   \expandafter\ifx\csname @mpargs\endcsname\relax %
     \expandafter\ifx\csname @captype\endcsname\relax %
       \marginpar{xxx}%
     \else
       xxx %
     \fi
   \else
     xxx %
   \fi}
 \gdef\xxx{\@ifnextchar[\xxx@lab\xxx@nolab}
 \long\gdef\xxx@lab[#1]#2{\textbf{[\xxxmark #2 ---{\sc #1}]}}
 \long\gdef\xxx@nolab#1{\textbf{[\xxxmark #1]}}
 \long\gdef\xxx@lab[#1]#2{}\long\gdef\xxx@nolab#1{}%
}

\let\realbfseries=\bfseries
\def\bfseries{\realbfseries\boldmath}

\newif\ifabstract
\abstracttrue
\abstractfalse
\newif\iffull
\ifabstract \fullfalse \else \fulltrue \fi

\ifabstract
\def\Section{Appendix}
\else
\def\Section{Section}
\fi

\newcounter{section-preserve}
\newcounter{theorem-preserve}
\newcommand{\blank}[1]{}
\newtoks\magicAppendix
\magicAppendix={}
\newtoks\magictoks
\newif\iflater
\laterfalse
\long\def\later#1{\iflater#1\else\magictoks={#1}%
  \edef\magictodo{\noexpand\magicAppendix={\the\magicAppendix \par
    \the\magictoks%
  }}
  \magictodo\fi}
\long\def\both#1{\iflater#1\else\magictoks={#1}%
  \edef\magictodo{\noexpand\magicAppendix={\the\magicAppendix \par
    \noexpand\setcounter{theorem-preserve}{\noexpand\arabic{theorem}}%
    \noexpand\setcounter{theorem}{\arabic{theorem}}%
    \noexpand\setcounter{section-preserve}{\noexpand\arabic{section}}%
    \noexpand\setcounter{section}{\arabic{section}}%
    \noexpand\let\noexpand\oldsection=\noexpand\thesection
    \noexpand\def\noexpand\thesection{\thesection}
    \noexpand\let\noexpand\oldlabel=\noexpand\label
    \noexpand\let\noexpand\label=\noexpand\blank
    \the\magictoks%
    \noexpand\setcounter{theorem}{\noexpand\arabic{theorem-preserve}}%
    \noexpand\setcounter{section}{\noexpand\arabic{section-preserve}}%
    \noexpand\let\noexpand\thesection=\noexpand\oldsection
    \noexpand\let\noexpand\label=\noexpand\oldlabel
  }}
  \magictodo
  \the\magictoks\fi}
\def\magicappendix{\latertrue \the\magicAppendix}

\let\epsilon\varepsilon
\let\eps\varepsilon
\DeclareMathOperator\interior{int}
\DeclareMathOperator\round{round}
\def\NCL{{\mathrm{NCL}}}

\iffull
  \long\def\both#1{#1}
  \let\later=\both
  \def\magicappendix{}
\fi

\def\defn#1{\textit{\textbf{#1}}}
\let\nonsouldefn=\defn

\def\defn#1{\textit{\textbf{\mbox{#1\/}}}}
\soulregister{\defn}{7}

\title{Escaping a Polygon}

\author{%
  Zachary Abel%
    \thanks{Department of Electrical Engineering and Computer Science,
      Massachusetts Institute of Technology, Cambridge, MA, USA,
      \protect\url{zabel@mit.edu}}
\and
  Hugo Akitaya%
    \thanks{Miner School of Computer and Information Sciences,
      University of Massachusetts, Lowell, MA, USA,
      \protect\url{hugoakitaya@gmail.com}}
\and
  Erik D. Demaine%
    \thanks{Computer Science and Artificial Intelligence Laboratory,
      Massachusetts Institute of Technology, Cambridge, MA, USA,
      \protect\url{{edemaine,mdemaine,jaysonl}@mit.edu}}
\and
  Martin L. Demaine\footnotemark[3]
\and
  Adam Hesterberg%
    \thanks{John A. Paulson School of Engineering and Applied Sciences,
      Harvard University, Cambridge, MA, USA,
      \protect\url{achesterberg@gmail.com}}
\and
  Jason S. Ku%
    \thanks{Department of Mechanical Engineering,
      National University of Singapore, Singapore,
      \protect\url{jasonku@mit.edu}}
\and
  Jayson Lynch%
    \footnotemark[3]
}

\date{}

\begin{document}
\maketitle

\begin{abstract}
 
Suppose an \emph{escaping} player (``human'') moves continuously
at maximum speed $1$ in the interior of a region,
while a \emph{pursuing} player (``zombie'') moves continuously
at maximum speed $r$ outside the region.
For what $r$ can the first player escape the region, that is, reach the
boundary a positive distance away from the pursuing player, assuming optimal
play by both players?
We formalize a model for this infinitesimally alternating 2-player game and
prove that it has a unique winner in any locally rectifiable region.
Our model thus avoids pathological behaviors (where both players can have
``winning strategies'') previously identified for pursuit--evasion
games such as the Lion and Man problem in certain metric spaces.
For some specific regions, including both equilateral triangle and square,
we give exact results for the \emph{critical speed ratio}, above which
the pursuing player can win and below which the escaping player can win
(and at which the pursuing player can win).
For simple polygons, we give a simple formula and polynomial-time algorithm
that is guaranteed to give a $10.89898$-approximation to the critical speed
ratio, and we give a pseudopolynomial-time approximation scheme for
approximating the critical speed ratio arbitrarily closely.
On the negative side, we prove NP-hardness of the problem for polyhedral
domains in 3D, and prove stronger results (PSPACE-hardness and NP-hardness
even to approximate) for generalizations to multiple escaping and pursuing
players.

\end{abstract}

\setcounter{page}0
\thispagestyle{empty}
\newpage

\section{Introduction}
\label{IntroductionSection}

What would \textit{you} do in a zombie apocalypse?
Humans are fascinated by this question:
zombies are the subject of
over 1,300 films,%
\footnote{\url{https://www.imdb.com/search/keyword?keywords=zombie&title_type=movie}}
over 150 TV shows,%
\footnote{\url{https://www.imdb.com/search/keyword?keywords=zombie&title_type=tvSeries}}
over 1,000 books,%
\footnote{\url{https://www.goodreads.com/shelf/show/zombie-apocalypse}}
and over 900 video games.%
\footnote{\url{https://store.steampowered.com/tag/browse/\#global_1659}}
A 2009 epidemiology study \cite{Smith2009} launched an entire academic
discipline of zombie mathematics, culminating in a collected works of
fifteen papers on the topic~\cite{SmithBook}.
In this paper, we provide a computational geometric study of how and when humans
can successfully escape zombies
in a new type of game called ``pursuit--escape''.

\paragraph{Related work: Pursuit--evasion.}
One well-studied family of geometric problems relevant to the zombie apocalypse
are \defn{pursuit--evasion games} \cite{Nahin-2007},
which arise in many military applications \cite{Isaacs-1965}.
In the most famous ``Lion and Man'' problem \cite{Littlewood-1986},
one evader (human/man) aims to eternally flee one
pursuer (zombie/lion) while moving at unit speed in a shared domain.
\label{sec:capture}
If the pursuer and evader are ever at the same point, then
the pursuer \defn{captures} the evader and the pursuer thereby wins the game.
For example, in a Euclidean disk domain, an evader can evade capture from an
equal-speed pursuer, but the pursuer can get arbitrarily close to the evader
\cite{Littlewood-1986,Croft-1964}.
If the evader is a factor $r > 1$ faster, then there is a closed form for the
minimum distance they can maintain from the pursuer \cite{Lewin-1986}.
Two pursuers can capture one equal-speed evader in the disk, and similarly $d$
pursuers can win in a $d$-dimensional ball \cite{Croft-1964};
but there is a (rectifiable) 2D polygonal region
with holes where the evader can evade two equal-speed pursuers
\cite{Abrahamsen-Holm-Rotenberg-Wulff-Nilsen-2017}.
In the infinite plane, an evader can evade equal-speed pursuers if and only if
the evader is outside the convex hull of pursuers
\cite{Rado-Rado-1974,Jankovic-1978},
but a $(1+\epsilon)$-faster evader can always evade countably many 
pursuers \cite{Abrahamsen-Holm-Rotenberg-Wulff-Nilsen-2018-arXiv}.
In 3D with polyhedral evader, pursuer, and obstacles, it is (weakly)
EXPTIME-hard to decide whether the evader can reach a goal point
without being captured \cite{Reif-Tate-1993}.

A discrete-time analog of the game, where the players take discrete steps of
up to unit distance, has been analyzed in many domains, including
polygons with holes \cite{pursuit-games-polygonal},
genus-$g$ polyhedral surfaces \cite{pursuit-games-polyhedral},
unbounded convex Euclidean domains \cite{pursuit-games-unbounded}, and
compact \textsc{cat}$(0)$ (nonpositive-curvature) spaces \cite{pursuit-games-2D}.
A discrete-space discrete-time analog of the game is the \emph{cops and robber
game} \cite{Bonato-Nowakowski-2011}, where $k$ cops/pursuers and one
robber/evader alternate turns moving along edges on a graph;
the smallest $k$ for which some cop can land on the robber
is EXPTIME-complete \cite{Kinnersley-2015} and W[2]-hard
\cite{Fomin-Golovach-Kratochvil-2008} to compute,
but e.g.\ at most $3$ in planar graphs \cite{Aigner-Fromme-1984}.
Other discrete pursuit-evasion games include treewidth
\cite{Seymour-Thomas-1993} and fire fighting \cite{Finbow-MacGillivray-2009}
on graphs, and Conway's Angel Problem \cite{Kloster-2007,Mathe-2007}
on grids.

\paragraph{Our problem: Pursuit--\emph{escape}.}
In this paper, we introduce and explore a variation called the
\defn{pursuit--escape game},
where the two players are the \defn{escaper} (human/man) and \defn{pursuer}
(zombie/lion), and they move in \emph{complementary domains} --- for example,
the interior and exterior of a simple polygon --- and the escaper's goal is to
reach a common point on the boundaries of these domains
where the pursuer is not.
As ``practical'' motivation, the escaper/human/man may be inside a building or on
its roof, while the pursuer/zombie/lion is restricted to remain outside;
the escaper would like to reach an exit
when the pursuer is a positive distance away.
(Assume, for example, that the building is surrounded by a parking lot full of
cars, enabling escape if the escaper has a brief head start.)
The escaper and pursuer move continuously, at speeds bounded by respective
maximum speeds, and each move optimally.
When can the escaper escape, and when can the pursuer always prevent escape?
Unlike pursuit--evasion, the escaper can easily evade \emph{capture},
because of the complementary domains: just stand still.
The challenge in pursuit--escape is to \emph{escape} at a point
where the escaper could not be captured.

One specific instance of this problem, where the pursuer and escaper regions
are the interior and exterior of a unit disk,
has been studied many times before in different guises.
In 1961, Richard Guy \cite{FirstRecord} posed this problem in the form of
the following puzzle, reproduced in \cite{FirstCircleSolution}:

\begin{quotation} 
\noindent
Some robbers have stolen the green eye of a little yellow god
from a temple on a small island in the middle of a circular lake. As they embark
in their boat, they are observed by a solitary guard on the shore, who can run
four times as fast as they can row the boat. Can they be sure of reaching the
shore and escaping with their loot? If so, how? And what if the guard could move
four \emph{and a half} times as fast as the robbers? \looseness=-1
\end{quotation}

\noindent
The same problem was rethemed by Martin Gardner \cite{Gardner1965}
\nocite{Gardner1990}
to be about a maiden on a rowboat, and more recently, featured
on Numberphile \cite{Numberphile}.
The first explicit positive solution we know of is \cite{FirstCircleSolution};
see also e.g.\ \cite[Section~4.1]{Nahin-2007}.
We prove (for the first time) that this strategy is in fact optimal.

In this paper, we study this problem for more general domains than the
unit disk.
Specifically, suppose an escaper $\escaper$ and a pursuer $\pursuer$ move simultaneously and
continuously within respective geometric domains $D_\escaper$ and $D_\pursuer$,
while each player has full knowledge of the movements of the other player.%
\footnote{Notationally, we use $\escaper$ to denote the escaper and $\pursuer$
  to denote the pursuer, as $e$ and $p$ are used for other concepts (notably,
  edge and point); for a mnemonic, think ``human'' and ``zombie''.}
The pursuer moves at a maximum speed that is $r$ times faster than the escaper,
who we can assume has maximum speed~$1$.
To get started, the escaper chooses a starting position in~$D_\escaper$, and
then the pursuer chooses a starting position in~$D_\pursuer$.
The escaper wins if they can reach an exit point
among a specified set $\Exit$ of exits, say $D_\escaper \cap D_\pursuer$,
that is a positive distance away from the pursuer; and the pursuer wins if
they can prevent the escaper from winning for arbitrarily long.
The goal of the \defn{pursuit--escape game} is to determine who wins for
given domains $D_\escaper,D_\pursuer$ for the escaper and the pursuer,
an exit set $\Exit$, and a speed ratio~$r$.

\paragraph{Capture vs.\ no capture.}
There are two possible models for what happens when the escaper and pursuer
meet at the same geometric point.
The Lion-and-Man game follows the \defn{capture} model where the pursuer
wins if they are ever at the same location as the escaper.
For simplicity in both model and strategy descriptions, we assume the
\defn{no-capture model}: if the escaper and the pursuer are at a common point,
then (instead of the pursuer immediately winning)
the escaper is merely unable to escape at such an exit point,
because they are not a positive distance from the pursuer.
Intuitively, the pursuer \emph{blocks} the escaper from exiting
instead of \emph{capturing}.
Equivalently, we can think of there being two copies of the exit set~$\Exit$ ---
one for the escaper and one for the pursuer, where the distance between
corresponding points is zero --- and the escaper wins if they can reach
a pursuer's exit point without capture,
while the pursuer must remain in their domain;
by this perspective,
the no-capture model is a special case of the capture model.

Our no-capture model makes it easier to specify strategies.
For example, an escaper strategy can start at an exit point,
which forces the pursuer to start at the same point;
this exact forced placement then makes it easier to specify the rest of the
escaper strategy.
Figure~\ref{no capture} gives some simple examples of such strategies.
For convex escaper domains, such behavior can be simulated in a
capture model: the pursuer can instead start extremely close to an exit,
forcing the escaper to be very close to that exit.
\label{sec:capture-nonconvex}
For nonconvex domains like Figure~\ref{no capture nonconvex},
we need to modify strategies to avoid prematurely touching
the boundary where the escaper might accidentally be captured by the pursuer,
instead moving arbitrarily close to such reflex vertices.
This is easy to do for the interior of a polygon or polyhedron, or
more generally any escaper domain that has an $\epsilon > 0$ offset
that metrically approximates the original:
apply the strategy to the offset domain (which avoids touching the boundary)
until it is time to exit, then walk $\epsilon$ to the boundary.

In most cases, we extend our results to the capture model.
(In fact, it makes some of our hardness proofs easier.)
But we focus on the no-capture model in particular because it
makes it easier to relate a discrete game (as defined below)
to the continuous game, which enables us to derive
pseudopolynomial-time approximation schemes;
we leave it open whether these can extend to the capture model.

\begin{figure}
  \centering
  \subcaptionbox{Disk}{\includegraphics{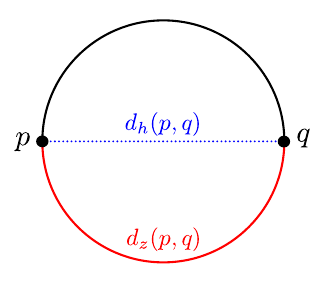}}\hfil
  \subcaptionbox{\label{no capture nonconvex} Nonconvex polygon}
    {\includegraphics{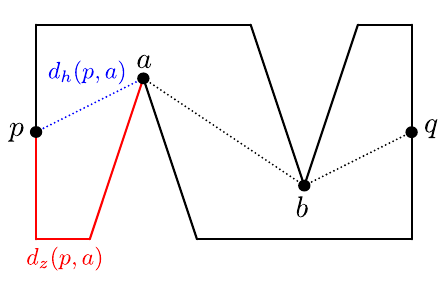}}

  \caption{Simple (suboptimal) strategies for the escaper in two domains:
    start at $p$, and run at full speed along the dotted shortest path to~$q$.
    The speed ratio $r$ must be at least $d_z(p,q) \over d_h(p,q)$
    for the pursuer to thwart this strategy,
    and thus the critical speed ratio is at least this large;
    see Theorem~\ref{lower bound}.}
  \label{no capture}
\end{figure}

\paragraph{Our results: Well-behaved model.}
It is not obvious that this game is well defined: how can two players decide
their motion continuously and instantaneously on the past motion of each other?
In contrast to most two-player games where the players take discrete turns,
so each move can easily depend on all past moves, this game involves
effectively infinitesimal alternation between the players' moves.
This difficulty was partially addressed by
Bollob\'as et al.~\cite{Bollobas-Leader-Walters-2012}
in the context of the Lion and Man problem, by giving a natural definition
of ``winning strategy'' which can fully depend on the past
(and in some sense the present)
behavior of the opponent.
Unfortunately, they also showed that this definition (without further
restrictions, at least in some scenarios) actually allows
\emph{both} players to have a winning strategy,
essentially because two strategies do not have a well-defined outcome of
playing against each other.

We prove an analogous result for the no-capture pursuit--escape game:
under definitions of winning strategy
analogous to \cite{Bollobas-Leader-Walters-2012}, the escaper always wins
(assuming the exit set is at least one-dimensional).
But notably, under this strategy, when the escaper exits,
their distance to the pursuer can be arbitrarily small,
depending on how quickly the pursuer responds.

Thus we turn to an alternate definition of ``winning''
the pursuit--escape game: the escaper must exit at a distance of at least
$\epsilon > 0$ from the pursuer, for a uniform constant $\epsilon$ that
does not depend on the pursuer's strategy
(in particular, how quickly they respond).
We prove that this definition guarantees that exactly one player wins,
in very general scenarios.
Indeed, we show that there is a \defn{critical speed ratio} $r^* \ge 0$
(possibly $\infty$) such that the escaper wins if and only if $r < r^*$
and the pursuer wins if and only if $r \geq r^*$ (for finite~$r^*$).
The \defn{pursuit--escape problem} thus asks to determine $r^*$,
given domains $D_\escaper,D_\pursuer$ and exit set $\Exit$.

In Section~\ref{sec:model}, we give a precise and general model for the
pursuit--escape problem, presented concisely to enable reading of
the algorithms in Sections~\ref{sec:O(1)}--\ref{sec:exact}.
In \Section~\ref{appendix:model}, we further detail the model and prove that it
satisfies the natural property that exactly one player wins the game,
for arbitrary domains (in any dimension)
that are finitely rectifiable in any bounded ball.
(Because the full model details are complicated, we delay them until we need
the techniques for developing additional algorithms in
\Section~\ref{appendix:pseudoPTAS}.)
\iftrue In particular, our model captures several natural settings for pursuit--escape:

\begin{itemize}
\item Escaper domains:
\begin{itemize} 
\item \defn{Polygon model}:
      the escaper domain $D_\escaper$ consists of the interior and boundary of
      a simple polygon.
\item \defn{Jordan model}:
      the escaper domain $D_\escaper$ consists of the interior and boundary of a
      Jordan curve of finite length, such as a circle in the original problem.
\item \defn{Polyhedron model}:
      the escaper domain $D_\escaper$ consists of the interior and boundary of
      a polyhedron homeomorphic to a sphere.
\end{itemize}
\item Pursuer domains:
\begin{itemize}
\item \defn{Exterior model}:
      the pursuer domain $D_\pursuer$ consists of the exterior and boundary of~$D_\escaper$.
\item \defn{Moat model}:
      the pursuer domain $D_\pursuer$ consists of the boundary of~$D_\escaper$
      (as with a shark trapped in a moat surrounding a building).
\item \defn{Graph model}:
      instead of Euclidean space, we have a graph with edge lengths
      (defining distance along the edges), and
      $D_\escaper$ and $D_\pursuer$ consist of some vertices and/or edges
      (including their endpoints).
\end{itemize}
\end{itemize}
\fi

For the Lion-and-Man game,
Bollob\'as et al.~\cite{Bollobas-Leader-Walters-2012}
gave an alternate approach for guaranteeing a unique winner to the game,
by restricting strategies to be ``locally finite''.
Our approach differs in that it redefines ``winning'' instead of directly
restricting strategies, though we also show that our definition implies
the existence of strategies satisfying a stronger (uniform) property
than local finiteness which we call ``obliviousness''
(see Section~\ref{sec:oblivious strategies}).
This stronger notion of obliviousness allows us to discretize the game
in a new way that enables efficient approximation algorithms.
Our results also apply more generally:
we allow strategies to run for unbounded time
(which is useful when the domains are unbounded),
and we guarantee unique winners without needing the Axiom of Choice.
(With the Axiom of Choice, we do obtain a simpler definition of
the pursuer winning phrased in terms of a single pursuer strategy,
but the rest of our results do not depend on this simpler definition.)

\paragraph{Our results: Algorithms.}
We develop several algorithms and prove several complexity results for
computing both exact and approximately optimal strategies for pursuit--escape.
For the benefit of the reader, we present the most algorithmically interesting
result first.

\definecolor{header}{rgb}{0.29,0,0.51}
\definecolor{rule}{rgb}{0.50,0.23,0.71}
\arrayrulecolor{rule}
\doublerulesepcolor{rule}
\setlength\arrayrulewidth{1pt}
\definecolor{gray}{rgb}{0.85,0.85,0.85}
\def\header#1{\textcolor{white}{\textbf{#1}}}

\begin{itemize}

\item In Section~\ref{sec:O(1)}, we give a
  polynomial-time $2(3 + \sqrt 6) < 10.89898$-approximation algorithm
  for the critical speed ratio $r^*$
  when the escaper domain is a simple polygon $P$
  and the pursuer domain is defined by either the exterior or moat model.
  The algorithm is based on a simple and natural formula
  $\max_{p,q \in \partial P} \frac{d_\pursuer(p,q)}{d_\escaper(p,q)}$,
  which we show is within a constant factor of~$r^*$
  (in particular, a lower bound on~$r^*$)
  in both the polygon and polyhedron model.
  These results extend to the capture model.

\item
  In \Section~\ref{appendix:exact},
  we solve the pursuit--escape problem
  exactly for several specific Jordan shapes
  in both the exterior and moat models:
  when $D_\escaper$ is 
  an unbounded wedge, a halfplane, 
  a disk (Guy's problem), 
  an equilateral triangle, and a square.
  We use the simple cases of wedge and halfplane to motivate
  a generalized escaper strategy called ``APLO''
  (axially progressing laterally opposing), which moves the
  escaper forward in an axial direction, with a lateral component that
  linearly opposes the pursuer's movement. We use APLO to 
  define optimal escaper strategies for the disk, equilateral
  triangle, and square.
  The last two results are especially complicated, requiring intricate
  strategies for both escaper and pursuer. 
  Table~\ref{exact table} summarizes the critical speed ratios we prove.
  These results extend to the capture model, as our optimal escaper strategies
  do not visit the boundary until the moment of escape.

  \begin{table}[b]
    \centering
    \def\arraystretch{1.75}
    \begin{tabular}{|c|c|c|c|c|} \cline{2-5}
    \rowcolor{header}
    \multicolumn{1}{c|}{\cellcolor{white}}
    & \header{$\theta$-Wedge} & \header{Disk} & \header{Equilateral Triangle} & \header{Square}
    \\ \hline
    \cellcolor{header} \header{$r^*$} & $1 / \sin \theta$ & $1/\cos \theta^* \approx 4.603$ & $(3+\sqrt 5) \sqrt 2 \approx 7.405$ & $\sqrt{\frac{5}{2} (7+\sqrt{41})} \approx 5.789$
    \\ \hline
    \end{tabular}
    \caption{Exact critical speed ratios for specific Jordan shapes,
      as proved in \Section~\ref{appendix:exact}.}
    \label{exact table}
  \end{table}

\item
  In \Section~\ref{appendix:pseudoPTAS},
  we give a
  pseudopolynomial-time approximation scheme for the critical speed ratio
  $r^*$ when the escaper domain is a simple polygon $P$
  and the pursuer domain is defined by either the exterior or moat model.
  This algorithm builds on the discrete model introduced in
  \Section~\ref{appendix:model} as an approximation to the continuous game
  to prove the game has a unique winner.
  The main extra step for an algorithm is proving a \emph{margin-of-victory}
  lemma (Lemma~\ref{MarginOfVictoryLemma}):
  if the escaper can win the continuous game at all, and the pursuer
  becomes slightly slower, then the escaper can win with a bit of time to spare.
  This seemingly innocuous claim is surprisingly involved to prove.
  It enables us to quantitatively decouple the interdependency of the escaper
  and pursuer strategies, and thereby bound the incurred discretization error.

\item
  In \Section~\ref{appendix:NPhard},
  we prove that the pursuit--escape problem
  in 3D is weakly NP-hard, even for polyhedral domains.
  This result motivates our focus on approximation algorithms.
  Our proof builds on the famous result by Canny and Reif
  \cite{Canny-Reif-1987} that it is weakly NP-hard to find shortest paths in 3D
  amidst polyhedral obstacles.

\item In \Section~\ref{MultipleZombiesSection}, we generalize the problem to
  multiple escapers and multiple pursuers, where the escapers win if at least one of them can escape.  
  On the positive side,
  our polynomial-time $O(1)$-approximation and pseudopolynomial-time
  approximation scheme generalize to this scenario.
  We also give a partial analysis of the case where the escapers and pursuers
  move at the same speed.
  On the negative side, we prove computational complexity --- both NP-hardness
  and PSPACE-hardness --- of even approximating the critical speed ratio
  in several scenarios, as summarized in Table~\ref{ComplexityTable}.
  Our reductions are from Nondeterministic Constraint Logic \cite{PSPACE-book},
  Planar Vertex Cover \cite{PlanarVertexCoverHard}, and Vertex Cover \cite{Karp72}.

\end{itemize}

\begin{table}[h]
\centering
  \tabcolsep=4pt
  \def\arraystretch{1.25}

\begin{tabular}{|c|c|c||c|}
\rowcolor{header}
\hline
\header{Escapers} & 
\header{Pursuers} & 
\header{Domain} & 
\header{Result} \\
\hline
Multiple & Multiple & 
Planar & $\vcenter{\vspace{4pt}%
  \hbox{PSPACE-hard [Theorem~\ref{thm:PSPACE}]; and}%
  \vspace{2pt}%
  \hbox{NP-hard, even to approximate at all [Theorem~\ref{PlanarHardnessTheorem}]}
  \vspace{2pt}%
}$
\\
\hline
1 & Multiple & 
Connected & \multicolumn{1}{l|}{NP-hard, even to 2-approximate [Theorem~\ref{OneHumanHardnessTheorem}]}\\
\hline
\end{tabular}
\caption{Multi-pursuer hardness results, as proved in \Section~\ref{MultipleZombiesSection}.}
\label{ComplexityTable}
\end{table}

\section{Brief Model (Abbreviated Version of \Section~\ref{appendix:model})}
\label{sec:model}

As mentioned above, it takes some care to define a precise model of
simultaneous play of two (or more) continuously moving players that can
continuously adapt to each other's motion.  We generally follow the definitions
from pursuit--evasion games in \cite{Bollobas-Leader-Walters-2012}, generalized
to where the players have different speeds and different domains they traverse.
Crucially, however, our game's definition of ``winning'' is different,
and we show that under it
exactly one player wins in any game.

In this abbreviated version of \Section~\ref{appendix:model},
we define the key notions of our model and summarize the
main results that are necessary for understanding
the algorithms in Sections~\ref{sec:O(1)}--\ref{sec:exact}.
For a more detailed description of why we use these
particular definitions, how they differ from past work, and proofs of why
exactly one player has a winning strategy, read instead the long form
of the model in \Section~\ref{appendix:model}.

\paragraph{Domains.}
A \defn{player domain} is a closed subset $D$ of Euclidean space
$\mathbb R^k$
that is \nonsouldefn{locally finitely rectifiable},
meaning that its intersection $D \cap B$
with any bounded closed Euclidean ball $B$
is ``finitely rectifiable''.  Formally,
$R \subseteq \mathbb R^k$ is \defn{finitely rectifiable} if it is
the union of the images of finitely many functions of the form
$S : [0,1]^k \to R$ satisfying the Lipschitz condition
$d(S(u),S(v)) \leq d(u,v)$ for all $u, v \in [0,1]^k$.%

The input to the pursuit--escape problem consists of both an \defn{escaper
  domain} $D_\escaper$ and a \defn{pursuer domain} $D_\pursuer$, and an \defn{exit
  set} $\Exit$.  The escaper and pursuer domains must be \emph{player domains} as
described above.
The exit set $X$ must also be a player domain,
and a subset of the player domains: $\Exit \subseteq D_\escaper \cap D_\pursuer$.
The goal of the escaper will be to reach an \defn{exit} --- any point of the
exit set $\Exit$ --- while being sufficiently away from the pursuer.

\paragraph{Motion paths.}
A \defn{motion path} with maximum speed $s \geq 0$ in metric domain $D$
is a function $\player : [0,\infty) \to D$
satisfying the \defn{speed-limit constraint} (Lipschitz condition)
$$ d_D(\player(t_1), \player(t_2)) \leq s \cdot |t_1 - t_2| \text{ for all }t_1,t_2 \geq 0. $$

We consider a model where the pursuer maximum speed is a factor of $r$
larger than the escaper maximum speed, which we assume is $1$ for simplicity.
Thus an \defn{escaper motion path} is a motion path of maximum speed $1$ in
the escaper domain $D_\escaper$,
while a \defn{pursuer motion path} is a motion path of maximum speed $r$ in
the pursuer domain $D_\pursuer$.

\paragraph{Strategies.}

For symmetry, the following definitions refer to a \defn{player}
(either escaper and pursuer) and their \defn{opponent} (pursuer or escaper,
respectively).
A \defn{player strategy} is a function $\Player$ mapping an opponent motion path
$\opponent$ to a player motion path $\Player(\opponent)$ satisfying the following \defn{nonbranching-lookahead
constraint}:
\begin{quote}
for any two opponent motion paths $\opponent_1,\opponent_2$ agreeing on $[0,t]$,
the strategy's player motion paths $\Player(\opponent_1),\Player(\opponent_2)$ also agree on $[0,t]$.
\end{quote}
An \defn{escaper strategy}~$\Escaper$ must satisfy one additional constraint,
the \defn{escaper-start constraint}:
\begin{quote}
all paths $\Escaper(\pursuer)$ (over all pursuer motion paths~$\pursuer$)
must start at a common point $\Escaper(\pursuer)(0)$.
\end{quote}

\paragraph{Win condition.}

First we define an infinite family of games $G_\epsilon$
for all $\epsilon > 0$.
An escaper strategy $\Escaper$ \defn{wins $G_\epsilon$} or
\defn{wins $G$ by $\epsilon$} if,
for every pursuer motion path $\pursuer$,
there is a time $t$ at which $\Escaper(\pursuer)(t)$ is on an exit and at distance
$\geq \epsilon$ from $\pursuer(t)$ in the pursuer metric.
A pursuer strategy $\Pursuer$ \defn{wins $G_\epsilon$} if,
for every escaper motion path $\escaper$,
and every time $t$ at which $\escaper(t)$ is on an exit,
$\escaper(t)$ is at distance $< \epsilon$ from $\Pursuer(\escaper)(t)$ in the pursuer metric:
$d_\pursuer(\escaper(t), \Pursuer(\escaper)(t)) < \epsilon$.

Now we can define the win condition for the pursuit--escape game~$G$.
The \defn{escaper wins $G$} if, for some $\epsilon > 0$,
there is an escaper strategy that wins $G$ by~$\epsilon$, i.e.,
wins~$G_\epsilon$.
The \defn{pursuer wins $G$} if, for all $\epsilon > 0$,
there is a pursuer strategy that wins $G_\epsilon$.

The main result of \Section~\ref{appendix:model} is the following:

\global\let\hlsecond=\relax
\begin{restatable*}{corollary}{finalmodelresult} \label{cor:main model}
  \hlsecond{%
  Any (continuous) pursuit--escape instance $(D_\escaper, D_\pursuer, \Exit)$
  has a critical speed ratio $r^* \geq 0$ (possibility $\infty$)
  such that the escaper wins $G(r)$
  for all speed ratios $r < r^*$ and the pursuer wins $G(r)$
  for all speed ratios $r \geq r^*$.
  }
  \global\let\hlsecond=\hl
\end{restatable*}

\section{$O(1)$-Approximation Algorithm}
\label{sec:O(1)}

In this section, we show that the critical speed ratio
for any simple polygon $P$
is lower bounded by 
$\max_{p,q \in \partial P}\frac{d_\pursuer(p,q)}{d_\escaper(p,q)}$ and
upper bounded by $10.89898 \max_{p,q \in \partial P}\frac{d_\pursuer(p,q)}{d_\escaper(p,q)}$, where
the escaper domain $D_\escaper$ is the interior and boundary of~$P$,
the pursuer domain $D_\pursuer$ is the boundary and optional exterior of~$P$
(thus allowing either the exterior or moat models),
and $d_\escaper$ and $d_\pursuer$ are the intrinsic (shortest-path) metrics
in the escaper and pursuer domains respectively
(as defined in Section~\ref{sec:model}).
Our results are constructive: in Section~\ref{O(1) strategies}
we give a winning escaper strategy for speed ratio
$\max_{p,q \in \partial P}\frac{d_\pursuer(p,q)}{d_\escaper(p,q)}$
and a winning pursuer strategy for speed ratio
$10.89898 \max_{p,q \in \partial P}\frac{d_\pursuer(p,q)}{d_\escaper(p,q)}$.
Furthermore, we give a polynomial-time algorithm in Section~\ref{O(1) algorithm}
to compute a maximizing point pair $(p,q)$,
resulting in a polynomial-time constant-factor approximation algorithm.
As described in Section~\ref{sec:capture-nonconvex},
the strategies can also be modified to work in the capture model
by a small inset.

\subsection{Strategies}
\label{O(1) strategies}

The escaper strategy is simple: run from $p$ to $q$
for the pair $p,q$ achieving the maximum ratio.
The main idea for our pursuer strategy is to decompose the polygon into its
medial axis, and within each region corresponding to a polygon edge, try to
follow a natural strategy for a halfplane, namely, following the projection
of the escaper onto the edge (proved optimal for a halfplane
in Section~\ref{sec:halfplane}).
The challenge is when the escaper crosses the medial axis from one
region to the other, and possibly jumps back and forth between two regions.
We only view the escaper as having changed regions once they have left a
larger region called the ``fringe'', meaning they are deeply in another region;
see Figure~\ref{MedialAxisFigure}.
Then we argue that the pursuer has enough time to transition to the new region's
strategy before the escaper can escape or transition again.

\begin{theorem}[lower bound] \label{lower bound}
  For any escaper domain $D_\escaper$, pursuer domain $D_\pursuer$, and exit set $\Exit$,
  the critical speed ratio $r^*$ is at least
  $$\max_{p,q \in \Exit}\frac{d_\pursuer(p,q)}{d_\escaper(p,q)}.$$
\end{theorem}

\begin{proof}
Let $p$ and $q$ be points maximizing the expression above,
and let $r_\eps=\frac{d_\pursuer(p,q)-\eps}{d_\escaper(p,q)}$.
The escaper can start at $p$ (escaper-start constraint); we can assume that the pursuer is also at $p$, or else the escaper escapes at $p$. 
Then, the escaper can run toward~$q$ at full speed (speed-limit constraint). 
This strategy does not depend on the pursuer's position at all (nonbranching-lookahead constraint).
The escaper's distance to $q$ is $d_\escaper(p,q)$ and the pursuer's is $d_\pursuer(p,q)$,
so when the escaper reaches $q$, the pursuer is at least $\eps$ away in pursuer metric, and the escaper escapes.
Therefore $r^* \geq r_\eps$ for all $\eps > 0$, and thus $r^* \geq r_0$.
\end{proof}

For a polygonal escaper domain $D_h$,
this escaper strategy can be extended to the capture model as described in
Section~\ref{sec:capture-nonconvex} and Figure~\ref{no capture nonconvex}.
For $\delta > 0$, consider the modified strategy where we inset $P$ by a disk
of radius $\delta$ to produce a region $P'$,
which for sufficiently small $\delta$ is connected
and has approximately the same shortest-path metric;
round the start point $p$ and end point $q$ to nearest points $p'$ and $q'$
respectively on $\partial P'$;
start at~$p'$;
run along a shortest path from $p'$ to $q'$ within~$P'$;
and then run along a shortest path from $q'$ to~$q$.
This strategy only touches the boundary of $P$ at the final time
when it reaches~$q$, but it starts at approximately the same point $p'$
and runs approximately the same distance.
Now take the limit as $\delta \to 0$.

\begin{theorem}[upper bound] \label{ConstantApproximationTheorem}
  For any simple polygon $P$ or polyhedron $P$ homeomorphic to a sphere,
  define escaper domain $D_\escaper$ as $P$'s interior and boundary,
  pursuer domain $D_\pursuer$ as $P$'s boundary and any subset of
  $P$'s exterior,
  and exit set $\Exit = \partial P$ as $P$'s boundary.
  Then the critical speed ratio $r^*$ is at most
  $$2 (3 + \sqrt 6) \max_{p,q \in \Exit}\frac{d_\pursuer(p,q)}{d_\escaper(p,q)}
  < 10.89898 \max_{p,q \in \Exit}\frac{d_\pursuer(p,q)}{d_\escaper(p,q)}.$$
\end{theorem}

\begin{proof}
Divide $P$ into (open) \defn{medial-axis regions}, as shown in
Figure~\ref{MedialAxisSimpleFigure}: each region is associated
with a facet (edge or face) $f$ of $P$ and is the set of points inside $P$
closer to $f$ than to any other facet of~$P$.
For each medial-axis region, also define its \defn{fringe} to be the union,
over points $p$ inside the region, of the ball of points within distance
$x \cdot d(p,\partial P)$ of~$p$, where
$d(p,\partial P )$ is the distance from $p$ to the nearest point on the boundary
of $P$ and $x = \sqrt 6 - 2 \approx 0.45$ is a fringe size parameter.
In particular, each fringe contains its medial-axis region.
Because there is a bijection between medial-axis regions and facets of $P$,
we also refer to the fringe of a facet of~$P$.

\begin{figure}[h]
\centering
\includegraphics[scale=0.6]{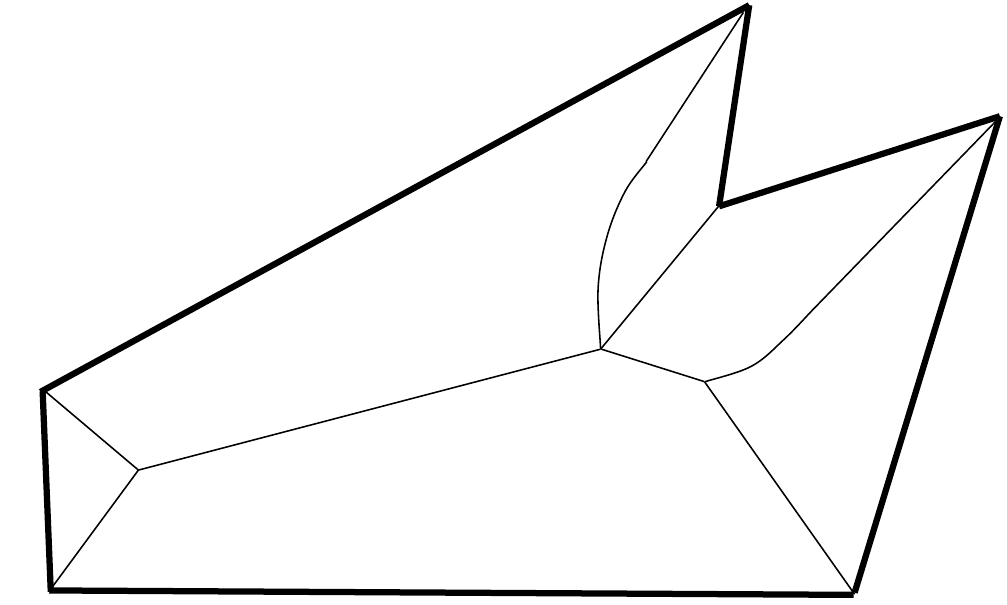}
\caption{A polygon and its medial axis.}
\label{MedialAxisSimpleFigure}
\end{figure}

Define the following pursuer strategy:
\begin{enumerate}
\item At all times, the pursuer has a \defn{target facet} $f$ of $P$ such that it attempts to be at the closest point on $f$ to the escaper. Initially, $f$ is a facet of $P$ that is closest to the escaper.
\item When the escaper exits the fringe of $f$, the pursuer runs to the closest point on the boundary $\partial P$ to the escaper. If that point is on facet $f'$ of $P$, then the pursuer switches its target facet to~$f'$.
\end{enumerate}

This strategy depends only on the current escaper position (nonbranching-lookahead constraint). 
We have to show that the strategy also satisfy the speed-limit constraint and that the pursuer is at the escaper's position whenever the escaper is in $\partial P$.
We show that, when the escaper leaves the fringe of facet $f$ in the medial-axis region of a facet $f'$, the pursuer can run into position (reaching the closest point in $\partial P$ to the escaper) before the escaper either reaches the boundary $\partial P$ (and escapes) or leaves the fringe of $f'$ (which would trigger another strategy change).

Next we define some points, as in Figure~\ref{MedialAxisFigure}.
Let $\escaper$ be the point at which the escaper leaves the fringe (drawn in blue) of a medial-axis region $R$ (drawn in red) with corresponding facet $f_p$. 
Because $\escaper$ is on the boundary of the fringe of $R$, it is also on a sphere centered at a point $o$ on the boundary of $R$ (i.e., on the medial axis) of radius $d(o,\escaper) = x \cdot d(o,\partial P) = x\cdot d(o,p)$ where $p$ is the closest point to $o$ on~$f_p$.
Let $\pursuer$ be the closest point to $\escaper$ on $f_p$, which is where the pursuer stands when the escaper is at $\escaper$.
Note that $\pursuer$ is an endpoint of $f_p$ if such endpoint is a reflex vertex of $P$, i.e., it is not necessarily the projection of $\escaper$ on the supporting line of $f_p$.
Let $q$ be the closest point to $\escaper$ on $\partial P$, and let $f_q$ be a facet containing $q$.

\begin{figure}[h]
\centering
\includegraphics[scale=0.5]{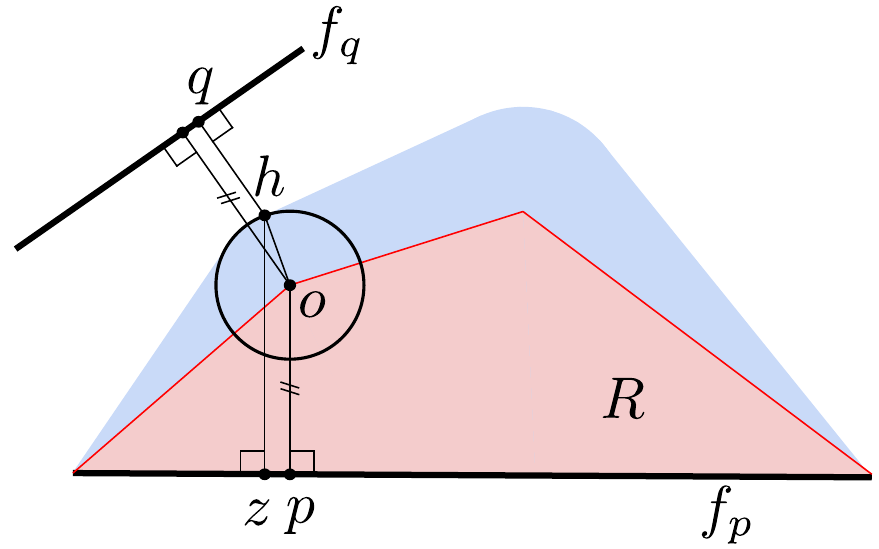}
\caption{The scenario when the escaper leaves the fridge (blue) of a medial-axis region $R$ (red), at a point $h$ now closest to facet $f_q$.}
\label{MedialAxisFigure}
\end{figure}

At $\escaper$, the escaper's distance to the boundary is
\begin{gather}
\begin{align}
d(\escaper,q) &\geq d(o,q) - d(o,\escaper) \quad \text{by triangle inequality}
\nonumber\\
&= d(o,q) - x \cdot d(o,p)
\nonumber\\
&\geq (1-x) \, d(o,p) \quad \text{because $d(o,q) \geq d(o,p)$}.
\end{align}
\label{d(h,q)}
\end{gather}
To leave the fringe of their new medial-axis region for facet $f_q$,
the escaper must run a distance of at least $x \cdot d(\escaper,q)$. 
We arrange for the pursuer to be in position for the new region's strategy
before either event (reaching the boundary or leaving the new fringe),
by bounding the motion of the pursuer during the next motion
of the escaper by at most $x \, d(\escaper,q) \leq d(\escaper,q)$
(assuming $x \leq 1$).
To reach the new strategy, the pursuer must move at most the sum of
three distances:
\begin{enumerate}
\item $d(\pursuer,p)$ to return to $p$. Because $\pursuer$ is the closest point on $f_p$ to $\escaper$, it is at least as close to $p$ as the projection of $\escaper$ onto the supporting line of $f_p$ (possibly closer, if $f_p$ is incident to a reflex vertex). 
The length of that projection is at most
$d(o,\escaper) = x \cdot d(o,p) \leq \frac{x}{1-x} \, d(\escaper,q)$
by (\ref{d(h,q)}),
so that is an upper bound on the pursuer's distance to return to $p$.
\item $d_\pursuer(p,q)$ to reach $q$. 
\item $\leq x \cdot d(\escaper,q)$ to match the escaper's move (projected onto $f_q$).
\end{enumerate}

So, if the pursuer's speed is enough to travel these three distances in the
time the escaper travels a distance of $x \cdot d(\escaper,q)$, then the
pursuer can be in position in time for the escaper's next region change
or escape.  That is, the critical speed ratio $r^*$ is at most
$$
\frac{\frac{x}{1-x} \, d(\escaper,q) + d_\pursuer(p,q) + x \cdot d(\escaper,q)}{x \cdot d(\escaper,q)}
= 1 + \frac{1}{1-x} + \frac{d_\pursuer(p,q)}{x \cdot d(\escaper,q)}. %
$$
Because a closest point to $o$ on $\partial P$ is $p$, the circle centered at $o$ with radius $d(o,p)$ is contained in $P$, so the line segment from $p$ to $q$ is also contained in~$P$.
Thus $d_\escaper(p,q) = d(p,q)$, which by triangle inequality is
at most $2 \, d(o,p) \leq \frac{2}{1-x} \, d(\escaper,q)$.
Thus our upper bound on $r^*$ is at most
$$
1 + \frac{1}{1-x} + \frac{d_\pursuer(p,q)}{x \, \frac{1-x}{2} \, d_\escaper(p,q)}.
$$
Because $d_\escaper(p,q)$ follows the straight line segment between $p$ and $q$,
$\frac{d_\pursuer(p,q)}{d_\escaper(p,q)} \geq 1$.
Therefore we can upper bound $r^*$ by
$$
\left(1 + \frac{1}{1-x} + \frac{2}{x(1-x)} \right) \frac{d_\pursuer(p,q)}{d_\escaper(p,q)}.
$$
This upper bound is minimized when $x = \sqrt 6 - 2$,
so picking $x = \sqrt 6 - 2$,
we obtain an upper bound of $r^* \leq 2 \, (3 + \sqrt 6) \frac{d_\pursuer(p,q)}{d_\escaper(p,q)}$.
\end{proof}

\subsection{Algorithm}
\label{O(1) algorithm}

The upper bound of Theorem~\ref{ConstantApproximationTheorem} combined with
the lower bound of Theorem~\ref{lower bound} suggest a polynomial-time
constant-factor approximation algorithm for simple polygons and
polyhedra homeomorphic to a sphere.
However, it requires some work to actually find a pair of points $p,q \in \Exit$
maximizing $\frac{d_\pursuer(p,q)}{d_\escaper(p,q)}$.  
Here we show how to solve the polygon case, and leave the polyhedron case
as an open problem.

\begin{theorem} \label{thm:approx}
  Given a simple polygon $P$ with $\leq n$ vertices,
  and given exit set $\Exit \subseteq \partial P$ as a set of $\leq n$ segments,
  we can compute the pair of points
  $(p^*,q^*) = \arg\max_{p,q \in \Exit}\frac{d_\pursuer(p,q)}{d_\escaper(p,q)}$,
  up to a $1+\epsilon$ factor error,
  in $O(n^4 \log {1 \over \epsilon})$ time.
\end{theorem}

\begin{proof}
  Two shortest paths $(p_1,p_2,\ldots,p_k)$ and $(p'_1,p'_2,\ldots,p'_l)$
  between point pairs $(p_1,p_k)$ and $(p'_1,p'_l)$ in $\partial P$ are
  \defn{combinatorially equivalent} if $p_1$ and $p_1'$ are on the same edge,
  $p_k$ and $p'_l$ are on the same edge, $k=l$, and $p_i=p'_i$ for
  $i\in\{2,\ldots,k-1\}$.

  Consider a point $p \in \partial P$ and its (geodesic) shortest path
  within $P$ to every other point in $\partial P$. 
  Let $\mathcal{S}(p)$ be the set of combinatorial equivalence classes
  of these shortest paths from $p$.
  By the shortest path map \cite{Mitchell-2017}, $|\mathcal{S}(p)|=O(n)$
  and $\mathcal{S}(p)$ can be computed in $O(n)$ time.

  We will partition the boundary of $\partial P$ into segments $S$ with the
  property that, for every $p,p'\in S$, $\mathcal{S}(p)=\mathcal{S}(p')$.
  Compute the arrangement of the lines going through every pair of vertices
  of~$P$.  There are $O(n^2)$ such lines, so we can compute the arrangement
  in $O((n^2)^2) = O(n^4)$ time \cite{Halperin-Sharir-2017}.
  Partition each edge of $P$ into $O(n^2)$ segments according to this
  arrangement, for a total of $O(n^3)$ segments.
  We can then clip and/or remove the segments to lie within~$\Exit$.

  Let $S$ be such a segment of~$\partial P$.
  For $k\ge 4$, every shortest path $(p_1,p_2,\ldots,p_k)$ where $p_1=p$
  satisfies that $S$ is on the same side of the line through $p_2$ and~$p_3$.
  Hence, every shortest path from a point $p'_1 \in S$ to $p_k$ is
  $(p_1',p_2,\ldots,p_k)$, and thus combinatorially equivalent to
  $(p_1,p_2,\ldots,p_k)$.
  For $k=3$, let $p_1$ be the leftmost point of $S$ and $p_3$ be the point that
  minimizes the convex angle at $p_2$ in the equivalence class of
  $(p_1, p_2, p_3)$.
  Then consider moving a point $p'_1$ starting at $p_1$ toward the other
  endpoint of~$S$.
  If $(p_1, p_2, p_3)$ ever becomes straight before reaching the endpoint,
  then $S$ would have been subdivided further, contradicting its definition.
  Thus $(p_1', p_2, p_3)$ remains a shortest path.
  We can use a similar argument to show that, for $k=2$,
  given two visible points $(p_1,p_2)$ where $p_1\in S$,
  every point in $S$ sees a point on the same edge as $p_2$
  (not necessarily $p_2$ itself).

  For each segment $S$, we can compute a member in each equivalence class of
  shortest paths from $S$ in $O(n)$ time.
  We map $S \times (\partial P \setminus S)$ to the square subset of the plane
  $[0,1]\times[0,1]$.
  It is easy to partition the boundary $\partial P$ into shortest-path
  equivalence classes when $k \ge 4$ based on the last endpoint of the
  shortest path; for example, the set of points $p_4$ on the same edge
  for which $(p_1, p_2,p_3,p_4)$ is a shortest path for all $p_1\in S$ and
  fixed $p_2, p_3$ can be computed from the line arrangement.
  Each equivalence class corresponds to a horizontal slab in the square.
  The intervals $I$ of the boundary $\partial P$ for which there are
  one-or-two-edge ($k \in \{2,3\}$) shortest paths from $S$ to $I$,
  the distance function is more complicated.
  The set $S \times I$ corresponds to a horizontal slab
  of the square $[0,1] \times [0,1]$. 
  The boundary between points on this square corresponding to one-edge
  shortest paths and points corresponding to two-edge shortest paths
  are straight lines connecting the left and right edges of the square,
  because such points correspond to shortest paths $(p_1,p_2,p_3)$
  where the points are collinear for fixed $p_2$, and $p_1\in S$.
  Moreover, the projection of such boundary line segments to the $y$ axis
  are interior-disjoint.
  Using these boundary lines, we can compute a partition of the square into
  regions and, for each region, compute $d_\escaper(p,q)$ efficiently
  because either we know $p$ and $q$ are visible from each other or we know
  the points $p_2,\ldots,p_{k-1}$ through which the shortest path passes.

  \begin{figure}[h]
  	\centering
  	\includegraphics[width=0.7\linewidth]{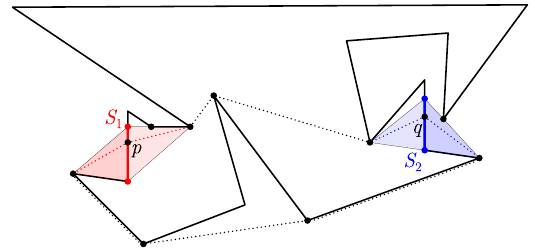}
  	\caption{Two ``hourglasses'', one inside and the other outside $P$, representing a region $S_1 \times S_2$ where shortest paths (inside or outside $P$) between $S_1$ and $S_2$ are in the same equivalence class.}
  	\label{fig:double-funnel}
  \end{figure}

  The computation of $d_\pursuer$ can be done in a similar manner,
  but using (geodesic) shortest paths on the exterior of~$P$.
  The partition of $\partial P$ into regions $S$ with combinatorially
  equivalent shortest paths is exactly the same.
  For each $S$, we obtain a new partition of $O(n)$ regions in the square
  $[0,1]\times[0,1]$ corresponding to $S\times (\partial P\setminus S)$.
  Overlaying both escaper and pursuer partitions of the square,
  we obtain $O(n)$ regions because of the horizontal separation
  between nonhorizontal boundaries.
  Figure~\ref{fig:double-funnel} illustrates one such a region $S_1 \times S_2$.
  For each region, computing
  $\max_{p \in S_1, q \in S_2}\frac{d_\pursuer(p,q)}{d_\escaper(p,q)}$
  becomes a constant-size optimization problem of the form
  $\max_{x,y \in [0,1]} \frac{f(x,y)}{g(x,y)}$
  where
  $f,g$ are functions on the segment parameters $x,y$
  of the form $\sqrt{x^2+a}+\sqrt{x^2+b}+\sqrt{y^2+c}+\sqrt{y^2+d}+e$.
  (The constant distances $a,b,c,d,e$ in each function $f,g$ can be computed
  exactly on a real RAM, or approximated using standard methods for computing
  square roots, such as Newton's Method.)
  This optimization can be solved by computing the gradient of
  $\frac{f(x,y)}{g(x,y)}$ and setting it to zero.
  We obtain two equations with two variables ($x$ and~$y$).
  We argue that each equation is a polynomial of degree at most~$48$.
  The numerator of a partial derivative of $\frac{f(x,y)}{g(x,y)}$ will contain $8$ types of square roots $\sqrt{w(x,y)}$ and we can eliminate each by multiplying by $1-\sqrt{w(x,y)}$.
  Each such multiplication blows up the degree of our polynomial
  by a factor of $2$, for a total of degree~$48$.
  The system has a constant number of variables and polynomials,
  and the polynomials have constant degree,
  so it can be solved using the Existential Theory of Reals
  in time linear in the bit complexity of the input and output
  \cite{Grigorev-Vorobjov-1988}, i.e., $O(\log {1 \over \epsilon})$ time.
  Then we take the maximum over all $O(n)$ regions for the segment,
  and take the maximum over all $O(n^3)$ segments $S$ of the boundary,
  for a total of $O(n^4 \log {1 \over \epsilon})$ time.
\end{proof}

In the case where the exit set $X$ is the entire boundary $\partial P$, the following lemma allows us to simplify the analysis in Theorem~\ref{thm:approx} by limiting our attention to regions where the escaper shortest path (inside $P$) has a single edge.

\begin{lemma}
	\label{lem:single-hour-glass}
	If $D_\escaper$ is a polygon, then there is a pair $(p,q)$ of points on its boundary maximizing 
	$\frac{d_z(p,q)}{d_h(p,q)}$ for which the shortest path inside $D_\escaper$ between $p$ and $q$ intersects $D_\escaper$ only at $p$ and $q$.
\end{lemma}

\begin{proof}
  Suppose that $(p,q)$ is a pair of points for which $\frac{d_z(p,q)}{d_h(p,q)}$ is maximized and, of such pairs, $(p,q)$ minimizes the number of segments (possibly single vertices) of $D_\escaper$'s boundary that intersects with the shortest path inside $D_\escaper$ between $p$ and $q$;
  see Figure~\ref{no capture nonconvex} for an example.
  (Because $D_\escaper$ is a polygon, that number of segments is always finite---in particular, at most the number of sides of the polygon---so we can choose to minimize it. This is the only place we use the assumption that $D_\escaper$ is a polygon.) Suppose for contradiction that there is a segment on the boundary of $D_\escaper$, that does not contain $p$ or $q$, through which the shortest path from $p$ to $q$ passes, and let $a$ be an endpoint of it.  Then $\frac{d_z(p,a)}{d_h(p,a)} \le \frac{d_z(p,q)}{d_h(p,q)}$ and $\frac{d_z(a,q)}{d_h(a,q)} \le \frac{d_z(p,q)}{d_h(p,q)}$. Note that by algebra, 
	$$\frac{d_z(p,a) + d_z(a,q)}{d_h(p,a) + d_h(a,q)} \le \max\left(\frac{d_z(p,a)}{d_h(p,a)}, \frac{d_z(a,q)}{d_h(a,q)}\right),$$ with equality only if one of the distances is 0 (impossible by assumption) or $\frac{d_z(p,a)}{d_h(p,a)} = \frac{d_z(a,q)}{d_h(a,q)}$. 
	Also, by the triangle inequality, $d_z(p,q) \le d_z(p,a) + d_z(a,q)$, and by the assumption that $a$ is on the shortest interior path between $p$ and $q$, $d_h(p,q) \ge d_h(p,a) + d_h(a,q)$, so
	$$\frac{d_z(p,q)}{d_h(p,q)} \le \frac{d_z(p,a) + d_z(a,q)}{d_h(p,a) + d_h(a,q)} \le \max\left(\frac{d_z(p,a)}{d_h(p,a)}, \frac{d_z(a,q)}{d_h(a,q)}\right) \le \frac{d_z(p,q)}{d_h(p,q)},$$
	so we must have equality at every step. In particular, $\frac{d_z(p,a)}{d_h(p,a)} = \frac{d_z(p,q)}{d_h(p,q)}$, so $(p,a)$ 
	is a pair of points for which $\frac{d_z(p,a)}{d_h(p,a)}$ is maximized and the number of segments of $D_\escaper$'s boundary that the shortest path inside $D_\escaper$ between $p$ and $a$ intersects is less than the corresponding number for $p$ and $q$, contradicting the choice of $p$ and $q$. Hence the shortest path inside $D_\escaper$ between $p$ and $q$ intersects the polygon only at $p$ and $q$, as claimed.
\end{proof}

\section{Exact Solutions}
\label{sec:exact}
\label{appendix:exact}

\iftrue
In this section, we compute the precise critical speed ratio for a few specific
escaper domains: a wedge (Section~\ref{sec:wedge}), a halfplane with specified
starting positions (Section~\ref{sec:halfplane}), the unit disk
(Section~\ref{sec:disk}), and two challenging cases --- the equilateral triangle
(Section~\ref{sec:triangle}) and the square (Section~\ref{sec:square}).
Motivated by the winning escaper strategy for the wedge and halfplane,
we also develop a generalized escaper strategy called APLO
(Section~\ref{sec:aplo}),
which we use to compute critical speed ratios in the later sections.
\else
In this section, we give the precise critical speed ratio for two challenging
escaper domains: the equilateral triangle
(Section~\ref{sec:triangle}) and the square (Section~\ref{sec:square}).
Appendix~\ref{appendix:exact} proves these results
by building a powerful escaper tool called APLO;
it also solves the additional cases of
the wedge, the halfplane, and the disk (proving optimality of the classic
escaper strategy for the first time).
\fi
Because the optimal pursuer strategies we prove never
leave the convex boundary of the escaper domain,
our results apply in both the moat and exterior models. 
The optimal escaper strategies we prove do not touch the boundary
of the escaper domain until the moment of escape,
so they trivially extend to the capture model
described in Section~\ref{sec:capture}.

\subsection{Wedge}
\label{sec:wedge}

While the case of an infinite wedge is not particularly interesting by itself,
a wedge models the local behavior around a vertex of a polygon, which will
be useful later.

\begin{theorem}
\label{WedgeTheorem}
If the escaper domain is a wedge, i.e., an unbounded intersection of two
halfplanes, having positive angle $2\theta \le \pi$, the
critical speed ratio is $r^* = 1/\sin\theta$.
\end{theorem}

\begin{proof} 

Let $o = (0, 0)$ be the apex of the wedge (or any point on the boundary if
$2\theta = \pi$); refer to Figure~\ref{WedgeFigure}.
Define right-handed coordinate frame $(\hat{x}, \hat{y})$ such
that $\hat{x}$ is the unit vector parallel to the angle bisector of the wedge,
where every point $p = (x, y)$ in the wedge satisfies $x \ge 0$, and $\hat{y}$
is the counterclockwise rotation of $\hat{x}$ by $90^\circ$.

We first provide a winning pursuer strategy when $r = r^*$: if the escaper is at point
$\escaper = (x,y)$, the pursuer will be at boundary point $\pursuer = (|y|/\tan\theta, y)$. This
pursuer strategy satisfies the escaper-start constraint and the nonbranching-lookahead
constraint (it only depends on the current position of the escaper) with paths
that satisfy the speed-limit constraint: given points $\Escaper(t) = (x_1, y_1), \Escaper(t +
\tau) = (x_2, y_2)$ on the
escaper path, noting that $(|y_2| - |y_1|)^2 \le (y_2 - y_1)^2$,
$$
\frac{\|\pursuer(t + \tau) - \pursuer(t)\|}{\|\Escaper(t + \tau) - \Escaper(t)\|} \leq
\frac{\sqrt{(|y_2| - |y_1|)^2/\tan^2\theta + (y_2 - y_1)^2}}{|y_2-y_1|} \le
r^*\sqrt{(1/\tan^2\theta + 1)\sin^2\theta} = r^*,$$
as desired.
This strategy is winning for the pursuer, as whenever the escaper is at a boundary
point $p$ the pursuer is also at $p$.

Next, we provide a winning escaper strategy when $r = r^* - \epsilon$ for any
$\epsilon > 0$. The escaper begins at point $s_\escaper = (\cos\theta, 0)$ on the angle
bisector, and the pursuer chooses a starting point $s_\pursuer = (|d|\cos\theta,
d\sin\theta)$ on the boundary. Without loss of generality, assume the pursuer
starts below the angle bisector with $d \le 0$. If $2\theta < \pi$ the escaper
runs at full speed to point $p = (\cos\theta, \sin\theta)$; otherwise if
$2\theta = \pi$, the escaper runs to point $(0, 1)$. This escaper strategy satisfies
the escaper-start constraint and the nonbranching-lookahead constraint (it only depends on
the starting position of the pursuer) with paths that satisfy the speed-limit
condition (escaper speed is always one). We claim this escaper strategy wins
$G_\delta$ for $0 < \delta < \epsilon\sin\theta$ when $2\theta < \pi$,
and wins for $0 < \delta < \epsilon$ when $2\theta = \pi$. In both cases, the
distance between $s_\pursuer$ and $p$ in the pursuer metric is $1 + |d|$. When $2\theta
< \pi$, the escaper reaches $p$ in time $t_\escaper = \sin\theta$, whereas the pursuer
travels at most distance $r t_\escaper = (r^* - \epsilon)\sin\theta$; so when the escaper
reaches $p$, the pursuer is at least distance $(1 + |d|) - (r^* -
\epsilon)\sin\theta \geq \epsilon\sin\theta$ from $p$. Alternatively, when
$2\theta = \pi$, the escaper reaches $p$ in time $1$, whereas the pursuer travels
at most distance $r = 1-\epsilon$; so when the escaper reaches $p$, the pursuer is
at least distance $(1 + |d|) - (1 - \epsilon) = \epsilon$ from $p$.
\end{proof}

\begin{figure}%
\centering
\includegraphics[scale=0.3]{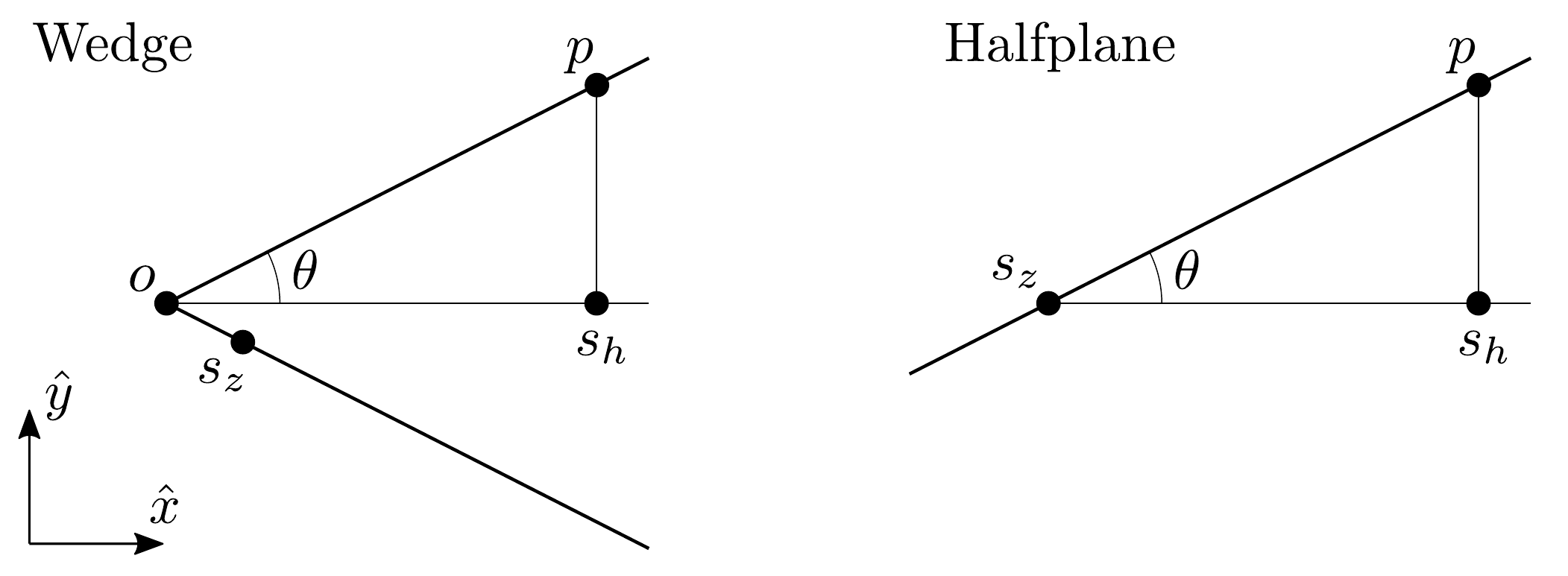}
\caption{Geometry of winning strategies in a wedge [Left] and halfplane [Right].}
\label{WedgeFigure}
\end{figure}

\subsection{$G(s_\escaper, s_\pursuer)$ in a Halfplane}
\label{sec:halfplane}

A halfplane is a special case of a wedge, so Theorem~\ref{WedgeTheorem} implies
that the critical speed ratio of a halfplane is $1$. We generalize this
strategy to find the critical speed ratio for the game $G(s_\escaper,
s_\pursuer)$ with prescribed escaper and pursuer
starting positions, $s_\escaper$ and $s_\pursuer$ respectively
(like the Lion and Man problem).
The halfplane case models the local behavior around an edge of a polygon after 
running another partial strategy, which again will be useful later.

\begin{theorem}
\label{HalfplaneTheorem}
If the escaper domain is the halfplane, the critical speed ratio for the game
$G(s_\escaper, s_\pursuer)$ is $r^* = 1/\sin\theta$ where angle $\theta = \angle s_\escaper s_\pursuer \escaper' \le
\pi/2$ and $\escaper'$ the closest boundary point to $s_\escaper$ (or any other boundary point
if the closest boundary point is $s_\pursuer$).
\end{theorem}

\begin{proof}
If $s_\escaper$ is on the boundary, either $s_\escaper = s_\pursuer$ and $r^* = 1$, or $s_\escaper \neq s_\pursuer$
and $r^* = \infty$. Otherwise, without loss of generality, let $s_\pursuer = (0,0)$ and
$s_\escaper = (1, 0)$.

We first provide a winning pursuer strategy when $r = r^*$: if the escaper is at
point $\escaper = (x, y)$, the pursuer will be at boundary point $\pursuer = (y/\tan\theta,y)$.
This pursuer strategy satisfies the nonbranching-lookahead constraint (it only depends on
the current position of the escaper) with paths that satisfy the speed-limit
constraint: given points $\Escaper(t) = (x_1, y_1), \Escaper(t + \tau) = (x_2, y_2)$ on the
escaper path, 
$$\frac{\|\pursuer(t+\tau) - \pursuer(t)\|}{\|\Escaper(t+\tau) - \Escaper(t)\|} \le \frac{\sqrt{(y_2 -
y_1)^2/\tan^2\theta + (y_2 - y_1)^2}}{|y_2 - y_1|}\le
r^*\sqrt{(1/\tan^2\theta + 1)\sin^2\theta} = r^*,$$
as desired.
This strategy is winning for the pursuer, as whenever the escaper is at a boundary
point $p$ the pursuer is also at $p$.

Next, we provide a winning escaper strategy when $r = r^* - \epsilon$ for any
$\epsilon > 0$: if $\theta < \pi/2$, the escaper runs straight to $p = (1,
\tan\theta)$ at full speed; otherwise if $\theta = \pi/2$, the escaper runs to
$s_\pursuer$ at full speed, and then to $p = (0, 1)$. This escaper strategy satisfies
the nonbranching-lookahead constraint (it only depends on the starting pursuer position)
with paths that satisfy the speed-limit constraint (the escaper speed is always
$1$). We claim that this strategy wins $G_\delta$, for $0 < \delta <
\epsilon\tan\theta$ when $\theta < \pi/2$, and for $0 < \delta < \epsilon/2$ when
$\theta =\pi/2$. When $\theta < \pi/2$, the escaper reaches $p$ in time $t_\escaper =
\tan\theta$, and the distance between $s_\pursuer$ and $p$ is $\sqrt{1 + \tan^2\theta}
= 1/\cos\theta$. However, the pursuer can travel at most distance $r t_\escaper = (r^* -
\epsilon)t_\escaper = 1/\cos\theta - \epsilon\tan\theta$ in that time, at least
distance $\epsilon\tan\theta$ from $p$. Alternatively, $2\theta = \pi$; when the
escaper first reaches $s_\pursuer$ the pursuer is within $\delta$ of $s_\pursuer$ or else the
escaper has already won. Then escaper reaches $p$ in time $2$, whereas the pursuer
travels at most distance $r = 1 - \epsilon < 1 - 2\delta$; so when the escaper
reaches $p$, the pursuer is at least distance $(1 - \delta) - (1 - \epsilon) >
\delta$ from $p$ as desired.
\end{proof}

\subsection{APLO Strategy}
\label{sec:aplo}

The strategy employed by the escaper in the previous section is quite simple: pick
a point on the boundary and run to it at full speed. Motivated by this escaper
strategy, we define a useful generalization which interpolates between two
extreme straight-line strategies depending on the position of pursuer, which we
will use to prove the critical speed ratio for the disk, equilateral triangle,
and square.

\begin{definition} For games where the pursuer domain is a topological circle,
let $D(\pursuer, t)$ denote the net signed counterclockwise distance\footnote{For
example, if the pursuer domain has length $\ell$ and the pursuer starts at $\pursuer(0)$
and in time $t$ circles the boundary clockwise exactly three times back to
$\pursuer(0)$ and then runs counterclockwise a distance $\ell/3$, then $D(\pursuer,t) =
-8\ell/3$. Note that the net signed distance $D(\pursuer, t)$ only depends on $\pursuer(t)$
and the homotopy type of the pursuer's path up to time $t$.} from $\pursuer(0)$ to
$\pursuer(t)$ counterclockwise along the pursuer domain, for any pursuer path $\pursuer(t)$.
Given:
\begin{itemize}
\item an escaper starting position $\escaper_0$, 
\item a preferred forward ``axial'' unit vector $\hat{u}$ (referencing also the ``lateral'' unit
vector $\hat{v}$ which is $\hat{u}$ rotated by a quarter-turn counterclockwise in the
plane),
\item speed ratio $r'$ (which must be an upper bound on pursuer speed), and
\item positive axial and lateral speeds $d_u$ and $d_v$ (which must satisfy $\sqrt{d_u^2
+ d_v^2} \leq 1$), 
\end{itemize}
we define the
\defn{axially progressing laterally opposing (APLO) escaper strategy} as follows
(see Figure~\ref{APLOFigure}): for a pursuer at position $\pursuer(t)$ at time $t$,
the escaper is at position:
$$H_{APLO}(\pursuer, t; \escaper_0, \hat{u}, r', d_u, d_v) = \escaper_0 + (td_u)\cdot\hat{u} +
\left(\frac{D(\pursuer, t)}{r'}d_v\right)\cdot\hat{v}.$$
\end{definition}

\begin{figure}[t]
\centering
\includegraphics[scale=0.35]{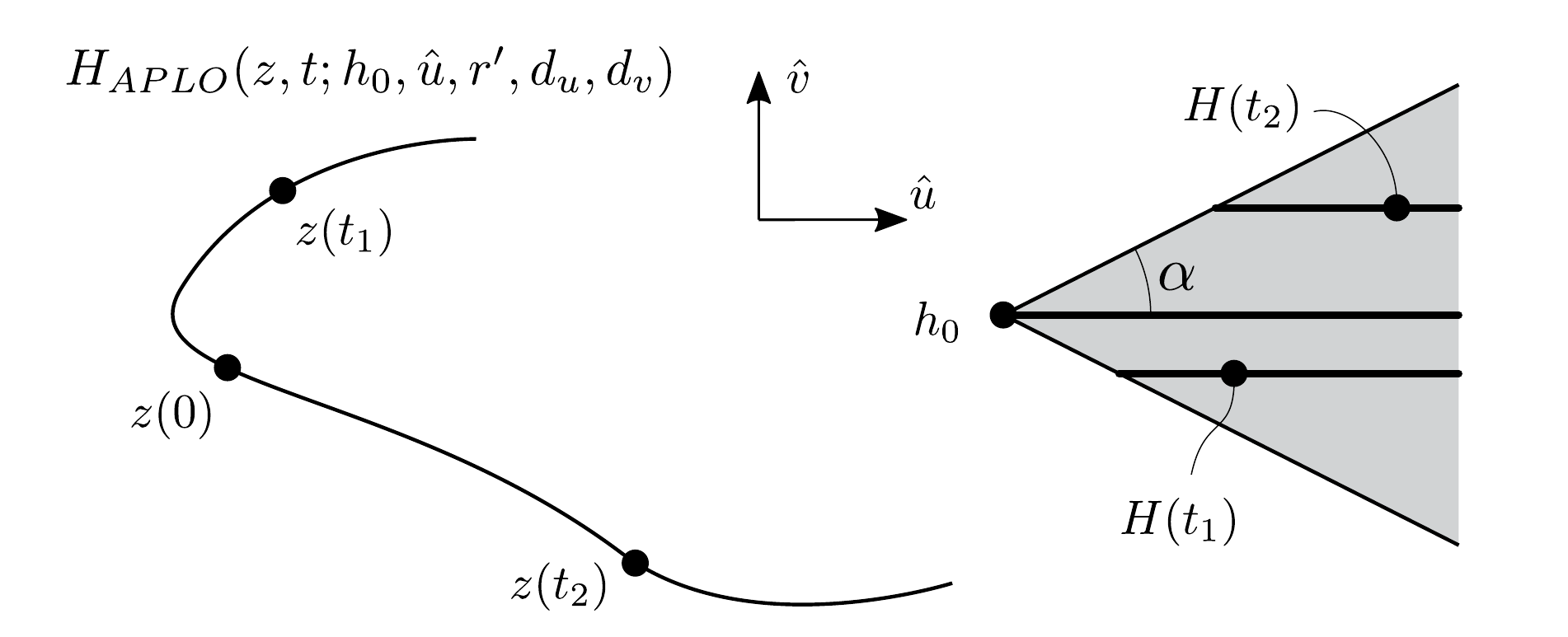}
\caption{Geometry of APLO strategy $H_{APLO}(\pursuer, t; \escaper_0, \hat{u}, r', d_u, d_v)$, where $d_v=\sin(\alpha)$ (and hence $d_u=\cos(\alpha)$). The shaded wedge represent the possible escaper positions.}
\label{APLOFigure}
\end{figure}

For example, if the pursuer runs clockwise at full speed $r$ then the escaper's
APLO response is to run in a straight line with velocity $d_u\hat{u} +
\frac{r}{r'}d_v\hat{v}$, which by the assumptions placed on our inputs has
magnitude at most $1$. If the pursuer stays at $\pursuer(0)$ then the escaper runs forward
along $\hat{u}$ at speed $d_u$. In general, the escaper always progresses forward
(in the $\hat{u}$ direction) with constant speed $d_u$, while the pursuer's
position at time~$t$ dictates the escaper's lateral offset (in the $\hat{v}$
direction) at time~$t$. Observe that this is done in a ``memory-less'' way: the
escaper's position at time~$t$ depends only on~$t$ and the pursuer's position at
time~$t$, not on the pursuer's position at any earlier (or later!) time.

\begin{lemma}
Any APLO escaper strategy
$H_{APLO}(\pursuer,t;\escaper_0,\hat{u},r', d_u, d_v)$ satisfies the escaper-start and
nonbranching-lookahead conditions with paths that satisfy the speed-limit condition. In
other words, $H_{APLO}$ is a valid strategy.
\end{lemma}

\begin{proof}
$H_{APLO}$ satisfies the escaper-start condition as at time $t=0$, $D(\pursuer, t)
= 0$, so $H_{APLO}$ places the escaper at position $\escaper_0 + 0\cdot
\hat{u} + 0\cdot \hat{v} = \escaper_0$, as required.

$H_{APLO}$  satisfies the nonbranching-lookahead condition as it does not depend on the
pursuer's position at any time except at time~$t$. 

To show that $H_{APLO}$ paths satisfy the speed-limit condition, we must show
that after any positive time $\tau$ from any time $t \geq 0$, the escaper travels
at most distance $\tau$. The distance traveled by escaper between times $t$ and
$t+\tau$ is:
$$|H_{APLO}(\pursuer, t) - H_{APLO}(\pursuer, t + \tau)| = \sqrt{\tau^2d_v^2 +
    \left(\frac{D(\pursuer, t) - D(\pursuer, t+\tau)}{r'}\right)^2d_u^2}.$$
This distance is maximized when $D(\pursuer, t) - D(\pursuer, t+\tau)$ is maximized. Since the
pursuer moves at rate at most $r$, this distance is at most $r\tau$. And since
$r' \geq r$ and $\sqrt{d_u^2+d_v^2} \leq 1$ by assumption on the inputs,
the distance the escaper travels is at most $\tau$, proving the claim.
\end{proof}

\subsection{Disk}
\label{sec:disk}

In this section, we solve for the first time the well-studied case of the disk.
While an escaper strategy with this speed ratio was known before,
we give an alternative escaper strategy based on our APLO technique.
Furthermore, we are not aware of any previous presentation of a matching
pursuer strategy.

\begin{theorem} \label{disktheorem}
Let $\varphi^*$ be the angle such that $\tan \varphi^* = \pi + \varphi^*$,
i.e.,~$\varphi^* \approx 0.430\pi$. If the escaper domain is a unit disk, the
critical speed ratio is $r^* = 1/\cos \varphi^* \approx 4.603$.
\end{theorem}

\begin{figure}
\centering
\includegraphics[scale=0.35]{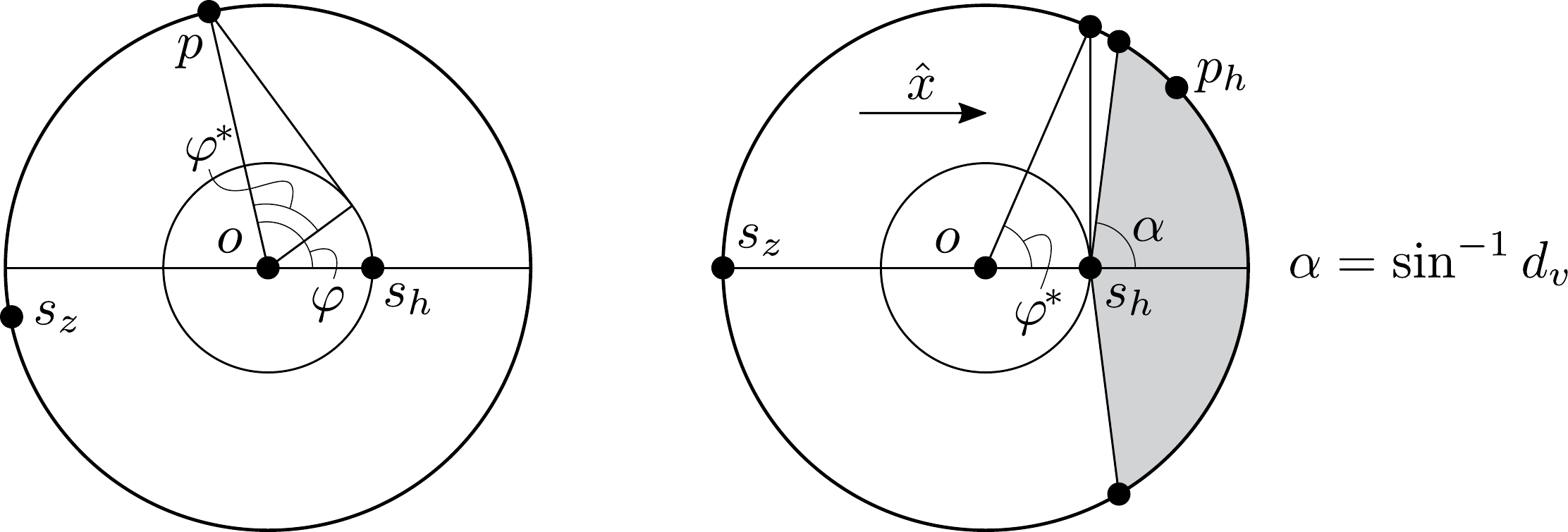}
\caption{Winning strategy geometries on a unit disk for both pursuer [Left] and
escaper [Right].}
\label{CircleFigure}
\end{figure}

\begin{proof}

Let $o$ be the center of the unit disk. We first provide a winning pursuer
strategy when $r \geq r^*$. The pursuer starts at the boundary point closest to the
escaper start point. When the escaper is greater than distance $1/r^*$ from $o$ and
the pursuer is not at the boundary point $\escaper'$ closest to the escaper, the pursuer
moves at full speed along the shorter arc toward $\escaper'$, breaking ties
arbitrarily, and otherwise stands still. This pursuer strategy satisfies the
escaper-start constraint and the nonbranching-lookahead constraint (it only depends on the
current position of the escaper) with paths that satisfy the speed-limit
constraint (pursuer runs at speed at most $r^*$). We claim this pursuer strategy
is winning. 

Suppose for contradiction there exists a winning escaper path $\Escaper$ ending
at some boundary point $p$. $\Escaper$ must contain at least one point at
distance $1/r^* = \cos\varphi^*$ from $o$; otherwise, if $\Escaper$ is always
outside the circle of radius $1/r^*$, the pursuer can at all times match the
escaper's angular velocity without exceeding speed $r^*$, so will always exist
at the closest boundary point to the escaper (in particular at $p$ at the end of
$\Escaper$). Then let $s_\escaper$ be the last point of $\Escaper$ at distance
$1/r^*$ from $o$, and without loss of generality, assume $s_\escaper =
(1/r^*,0)$ and $p= (\cos\varphi, \sin\varphi)$ for some $0\le\varphi<\pi$ (see
Figure~\ref{CircleFigure} [Left]). Then the escaper cannot reach $p$ faster than time
$t_\escaper$, where:
\begin{itemize}
\item  $t_\escaper = \sqrt{(\cos \varphi - \cos\varphi^*)^2 + \sin^2\varphi}$
when $0 \le \varphi \le \varphi^*$ (by straight line from $s_\escaper$ to $p$),
and 
\item $t_\escaper > \sin\varphi^* + (\varphi - \varphi^*)/r^*$ when $\varphi^* <
\varphi < \pi$ (by first running around the circle of radius $1/r^*$, then in a
straight line to $p$).
\end{itemize}
Since the subset of $\Escaper$ after $s_\escaper$ to $p$ lies strictly outside
the circle of radius $1/r^*$, the pursuer's angular velocity around $o$ is
always greater than the escaper's, meaning the arclength between the pursuer and
the closest boundary point to the escaper only decreases, so the pursuer runs in
a consistent direction. If this arclength reaches zero, the pursuer can track
the closest boundary point to the escaper and the escaper will not win, so if
$\Escaper$ wins, the pursuer always runs at full speed toward $p$. Let
$s_\pursuer = (\cos\theta, \sin\theta)$ be the pursuer position when the escaper
is at $s_\escaper$, and let $t_\pursuer$ be the time the pursuer takes to reach
$p$. If $0 \le \theta < \pi$, then $t_\pursuer = |\theta - \varphi|/r^*$;
otherwise if $\pi\le \theta < 2\pi$, the pursuer reaches $p$ in time $t_\pursuer
= (2\pi + \varphi - \theta)/r^*$. $t_\pursuer$ is maximized when $\theta = \pi$,
so without loss of generality we can assume that $s_\pursuer = (-1, 0)$ and
$t_\pursuer = (\pi + \varphi)/r^*$. The pursuer is at $p$ when the escaper
reaches $p$ if $t_\escaper - t_\pursuer\ge 0$. When $\varphi > \varphi^*$,
observe that
$$t_\escaper - t_\pursuer > (\sin\varphi^* + (\varphi - \varphi^*)/r^*) - (\pi +
\varphi^* + (\varphi - \varphi^*))/r^* = \sin\varphi^* - \tan\varphi^*/r^* = 0.
$$ 
Alternatively, when $\varphi \le \varphi^*$, observe that $t_\escaper-t_\pursuer
\ge (t_\escaper-t_\pursuer)|_{\varphi = \varphi^*} =  0$, as the derivative of
$t_\escaper - t_\pursuer$ is never positive over the domain:
$$\frac{d}{d\varphi}(t_\escaper - t_\pursuer) = -\cos\varphi^*\left(1 -
\frac{\sin\varphi}{\sqrt{\sin^2\varphi + (\cos\varphi -
\cos\varphi^*)^2}}\right) \le 0.$$
Thus the pursuer is at $p$ when the escaper reaches $p$, a contradiction. 

Next, we provide a winning escaper strategy when $r = r^* - \epsilon$ for any
positive $\epsilon$. The escaper begins on the circle $C$ of radius $1/r^*$
concentric with the unit disk, and then runs at full speed around $C$ (with
angular speed $r^*$ about $o$) until the escaper and pursuer reach respective
positions $s_\escaper$ and $s_\pursuer$ where $\angle s_\escaper o s_\pursuer =
\pi$. Without loss of generality, $s_\escaper = (\cos\varphi^*, 0)$ and
$s_\pursuer = (-1, 0)$. The escaper reaches such a state in finite time because
the pursuer can run around the unit disk with angular speed at most $r < r^*$.
Then, the escaper executes APLO strategy $H_{APLO}(\pursuer, t; s_\escaper,
\hat{x}, r, d_u, d_v)$ where $\pursuer(0) = s_\pursuer$, $\hat{x}$ is the unit
direction from $s_\pursuer$ to $s_\escaper$, and $d_v = r/r^* < 1$ and $d_u
=\sqrt{1 - d_u^2}$ (see Figure~\ref{CircleFigure} [Right]). At some finite time $t_f$ while executing this strategy, the
escaper reaches some boundary point $p_\escaper = (\cos\varphi, \sin\varphi)$;
without loss of generality assume $0 < \varphi$. Then at the same time, the
pursuer is at point $p_\pursuer = (\cos(\theta-\pi), \sin(\theta-\pi))$ where
$\theta = D(\pursuer, t_f) = r\sin\varphi/d_v = r^*\sin\varphi$ by definition of
APLO. 

We claim this strategy wins $G_\delta$ for some $\delta > 0$, i.e., $p_\pursuer
\neq p_\escaper$. It suffices to show that $\varphi > \theta - \pi$. Since
$\varphi < \varphi^*$ and function
$f(x) = (\sin x)/(\pi+ x)$ strictly increases over the
domain $0 \leq \varphi < \varphi^*$,
$$\varphi - (\theta - \pi) = (\pi + \varphi) - r^*\sin\varphi 
= (\pi + \varphi)\left(1 - 
\frac{\pi + \varphi^*}{\sin\varphi^*}
\frac{\sin\varphi}{\pi + \varphi}
\right) > 0,$$
proving the claim.
\end{proof}

\subsection{Equilateral Triangle}
\label{sec:triangle}

The equilateral triangle is perhaps the simplest polygon,
so serves as a natural starting point for exact bounds:

\begin{theorem} 
If the escaper domain is an equilateral triangle, the critical
speed ratio is $r^* = (3 + \sqrt{5})\sqrt{2} \approx 7.405$. 
\label{thm:triangle}
\end{theorem}

\begin{figure}[htbp]
\centering 
\includegraphics[scale=0.5]{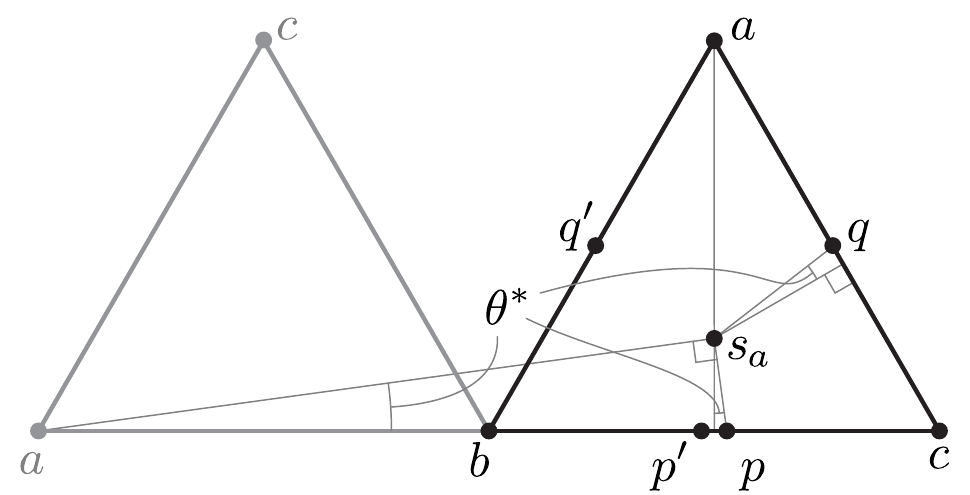}
\caption{Geometry for computing the critical speed ratio $r^* = 1/\sin\theta^*$
for a triangle.}
\label{TriangleRatioFigure} 
\end{figure}

Let $\theta^* <\pi/2$ be the positive angle such that $r^* = 1/\sin\theta^*$;
see Figure~\ref{TriangleRatioFigure}. The speed ratio $r^*$ is chosen such that
if the pursuer is at corner $a$ and the escaper is at point $s_a$ at distance
$(\sqrt{3} - 3\tan\theta^*)/2 = \sqrt{3(7 - 3 \sqrt{5})/2}\approx 0.6616$ along
the angle bisector of $a$, then the escaper has four simultaneous threats to
exit at $p$, $p'$, $q$, and $q'$. Specifically, the escaper distance from $s_a$
to $p$ is exactly factor $r^*$ smaller than the pursuer distance
counterclockwise from $a$ to $b$ to $p$, i.e., $r^*\|s_a-p\| = 1 + \|b-p\|$, and
the escaper distance from $s_a$ to $q$ is exactly a factor $r^*$ smaller than
the pursuer distance from $a$ to $b$ to $c$ to $q$, i.e., $r^*\|s_a-q\| = 2 +
\|c-q\|$; and similarly for $p'$ and $q'$ in the clockwise direction.

\begin{proof} We first provide a winning pursuer strategy when $r \geq r^*$. Our
pursuer strategy transitions between six different strategies as the escaper
move within the triangle. These six strategies $z(h; i, j)$ are shown in
Figure~\ref{TriangleTransitionsFigure}, where each strategy is associated with a
corner $i\in\{a,b,c\}$ and a sign $j\in\{-1, 1\}$. Each of these strategies is
identical up to rotations and reflections, so let us first focus on one of the
strategies, $z(h; a, 1)$. 

\begin{figure}
\centering 
\includegraphics[scale=0.2]{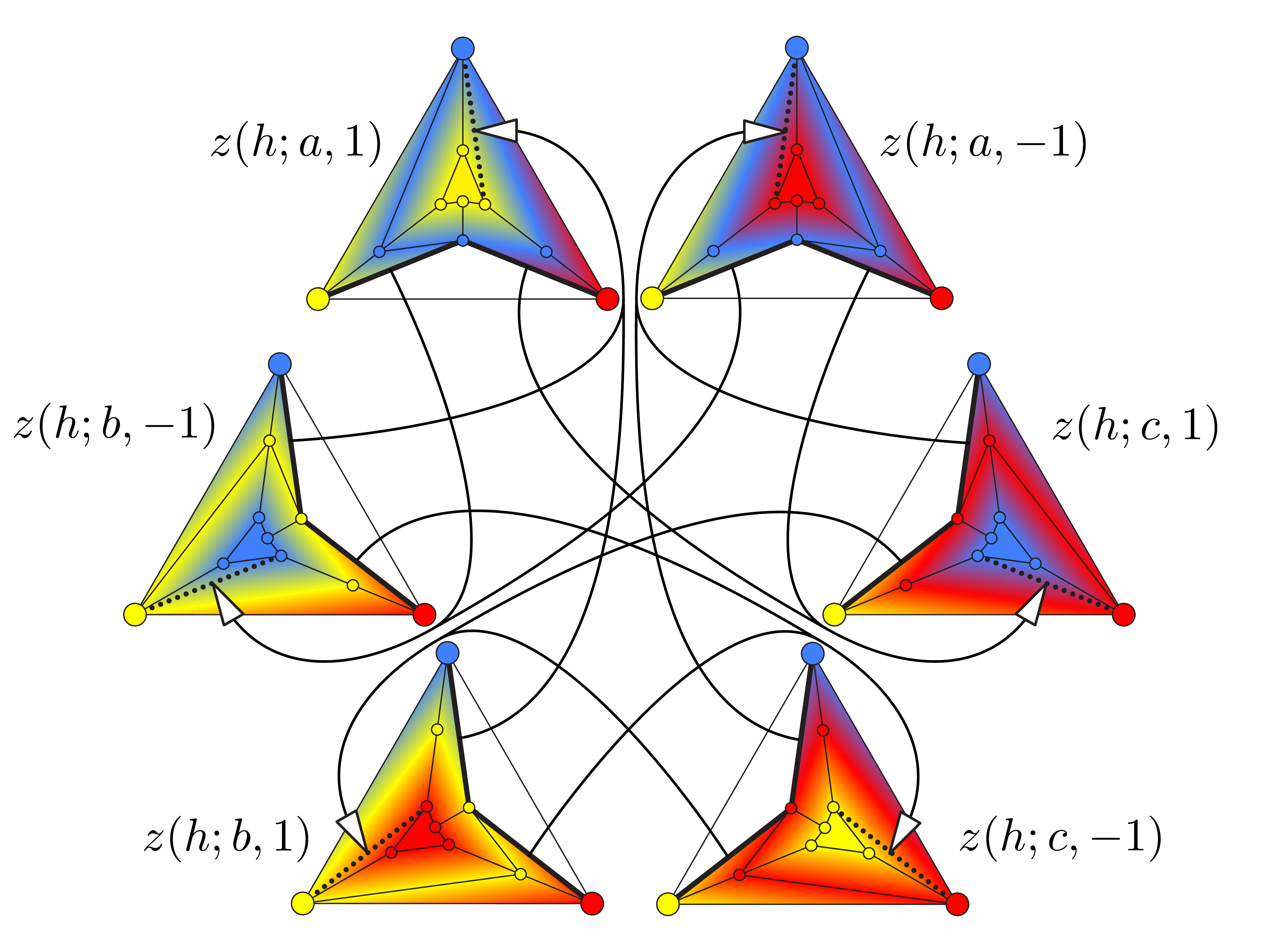}
\caption{Transitions between pursuer strategies.} 
\label{TriangleTransitionsFigure} 
\end{figure}

The $z(h; a, 1)$ strategy, depicted in
Figure~\ref{TriangleZombieFigure}, maps each point of the colored subset of the
triangle to a point on the boundary via a piecewise-linear map. Wherever the
escaper is in the colored region of a strategy, the strategy will place the
pursuer at the boundary point designated by the map. To make it easier to
reference points on the boundary, we map each boundary point on edge $ab$ and
edge $ac$ to a number, varying linearly from $-1$ at vertex $b$ (yellow), to $0$
at vertex $a$ (blue), to $1$ at vertex $c$ (red). The left drawing of
Figure~\ref{TriangleZombieFigure} depicts the geometry of the linear patches of
this map: 

\begin{figure}
\centering 
\includegraphics[scale=0.5]{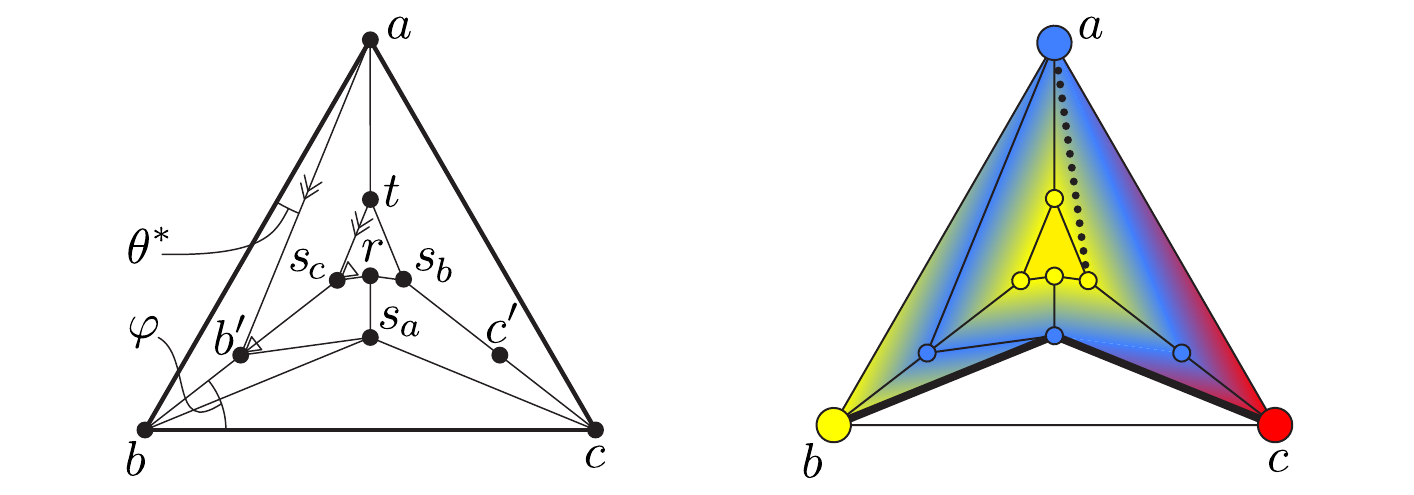}
\caption{Geometry of $\pursuer(\escaper; a, 1)$.
This function is linear in each region $a b b'$, $a b' s_c t$,
$b' s_a r s_c$, $r s_b t s_c$, $a t s_b c$, and $c s_b r s_a$,
where points $\{a, b', s_a, c'\}$ have value $0$ (blue), points $\{b,
s_c, r, s_b\}$ have value $-1$ (yellow), and point $c$ has value $1$ (red).}
\label{TriangleZombieFigure} 
\end{figure}

\begin{itemize}
\item point $s_i$ for $i\in\{a, b, c\}$ is distance
$\sqrt{3(7-3\sqrt{5})/2} \approx 0.6616$
along the angle bisector of corner $i$;
\item point $b'$ is the midpoint of segment $bs_c$;
\item point $c'$ is the midpoint of segment $cs_b$;
\item point $t$ is the intersection of the angle bisector of $a$ and the line
though $s_c$ parallel to segment $ab'$; and
\item point $r$ is the intersection of the angle bisector of $a$ and the line
through $b'$ parallel to segment $b's_a$.
\end{itemize}

We specify each linear patch by specifying the value at each vertex:

\begin{itemize}
\item points $\{a, b', s_a, c'\}$ have value $0$ (blue),
\item points $\{b, s_c, r, s_b\}$ have value $-1$ (yellow), and
\item point $c$ has value $1$ (red).
\end{itemize}

This map has the property that the gradient at every point within each linear
patch has the same value, namely $r^*$. Thus, as the escaper moves within the
colored region, the pursuer's speed will always stay below $r^* \leq r$, so the
strategy will be valid. This map also has the property that the pursuer and the
escaper will be collocated whenever the escaper is on edges $ab$ or $ac$, so the
escaper cannot win along those edges. If the escaper reaches edge $bs_a$ or edge
$cs_a$, the pursuer will switch strategies, respectively to either $z(h; b, -1)$
or $z(h; c, 1)$. These strategies exactly match strategy $z(h; a, 1)$ along
their respective transition edges. By transitioning between these strategies via
the transition graph shown in Figure~\ref{TriangleTransitionsFigure}, the
pursuer will always be collocated with the escaper whenever the escaper is at
the boundary, as desired.

Next, we provide a winning escaper strategy when $r = r^* - \epsilon$ for any
positive $\epsilon$. Our escaper strategy follows a similar strategy as the
circle escaper strategy: reach a state where the escaper can win via a single
APLO strategy. In particular, when the escaper is on the boundary of triangle $T
= s_as_bs_c$ (e.g., at some point $p_h$ on $s_bs_c$), and the pursuer is
antipodal along the opposite edge boundary with the same ratio (e.g., at point
$p_z$ along segment $bc$ where $\|b-p_z\|/1 = \|s_b-p_h\|/\|s_b-s_c\|$), then the
escaper will be able to win via an APLO strategy to the boundary. We will reach
such a configuration in two phases.

In the first phase, the escaper starts anywhere on $T' = t_at_bt_c$, the
triangle formed by connecting the midpoints of triangle $T$. Let $m_a$, $m_b$,
and $m_c$ be the midpoints of $bc$, $ca$, and $ab$ respectively; see
Figure~\ref{fig:TriangleHumanFigure}. The perimeter of $T'$ has length
$3(7-3\sqrt{5})/4 \approx 0.2188$ which is less than $3/r^* \approx 0.4051$, so
the escaper can run around $T'$ faster than the pursuer can run around the
boundary. The escaper runs around $T'$ until the escaper reaches a position
$p_{h1}$ on $T'$ such that the pursuer's position $p_{z1}$ is antipodal. Without
loss of generality, assume $p_{h1}$ is on segment $t_at_b$ and $p_{z1}$ is
antipodal on segment $cm_a$ such that $\|m_a-p_{z1}\|/1 = \|t_a - p_{h1}\|/\|t_a -
t_b\|$.

\begin{figure}
\centering 
\includegraphics[scale=0.45]{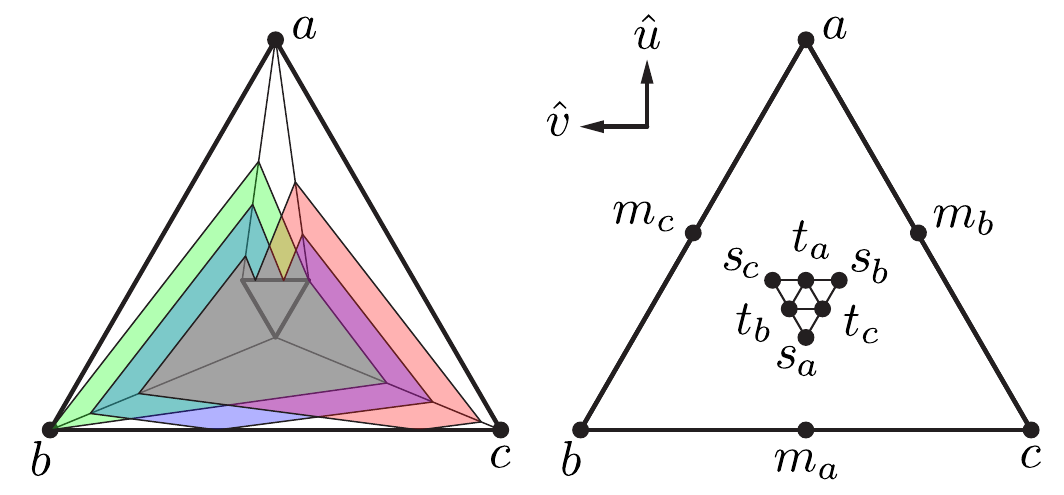}
\caption{Geometry of the escaper strategy for a triangle.} 
\label{fig:TriangleHumanFigure} 
\end{figure}

Now that the escaper is antipodal to the pursuer on $T'$, the escaper enters the
second phase, executing an APLO strategy $H_{APLO}(z, t; p_{h1}, \hat{c}, r,
d_u, d_v)$ where $\hat{c}$ is the unit direction from $c$ to $m_c$, $d_v = \|t_a
- t_b\|r < 1$ , and $d_u = \sqrt{1 - d_v^2} < 1$, until the escaper reaches
triangle $T$ at some point $p_{h2}$. By definition of APLO, during this process
the escaper's projection onto segment $t_at_b$ remains antipodal to the pursuer,
so when the escaper reaches $p_{h2}$, the pursuer is at a point $p_{z2}$
antipodal to $p_{h2}$ on $T$. 
Without loss of generality, assume $p_{h2}$ is on
segment $t_as_c$ and $p_{z2}$ is antipodal on segment $cm_a$ such that $\|m_a -
p_{z1}\|/1 = \|t_a - p_{h2}\|/\|s_b-s_c\|$. Let $x_{z2} = \|m_a - p_{z1}\|$, let
$d_s = \|s_b-s_c\| = (7-3\sqrt{5})/2$, and let $x_{h2} = \|t_a - p_{h2}\| =
x_{z2}d_s$.

Now that the escaper is antipodal to the pursuer on $T$, the escaper enters the
third and final phase, executing an APLO strategy $H_{APLO}(z, t; p_{h2},
\hat{a}, r, d_u, d_v)$, where $\hat{a}$ is the unit direction from $m_a$ to $a$,
$d_u= \cos(\pi/3 + \theta^*)$, and $d_v = \sin(\pi/3 + \theta^*)$, until the
escaper reaches the boundary at some point $p_{h3}$, with the pursuer at some
point $p_{z3}$. It remains to show that $\|p_{h3}-p_{z3}\|$ is bounded away from
zero.

If the pursuer remains on $bc$, the escaper wins easily as $p_{h3}$ is above the
line $s_cs_b$. Otherwise, there are two cases: the pursuer leaves the segment
$bc$ last through either $b$ or $c$. It suffices to show that the separation of
their projections onto segment $bc$ is bounded away from zero, specifically
quantity $|(p_{h3} - t_a)\cdot\hat{v} - (p_{z3} - t_a)\cdot\hat{v}|$ where
$\hat{v}$ is the unit vector $(b - c)$. Let $x_{h3} = (p_{h3} -
t_a)\cdot\hat{v}$. Note that $x_{h3}$ is positive to the left of $t_a$.

\begin{enumerate}

\item (Pursuer leaves $bc$ through $c$): pursuer leaves counter-clockwise from
$p_{z2}$, so $(p_{h3} - p_{h2})\cdot\hat{v} = x_{h3} - x_{h2} \ge 0$. Then by
APLO, the pursuer travels counter-clockwise from $p_{z2}$ by distance $r(x_{h3}
- x_{h2})/\sin(\pi/3+\theta^*)$, for distance $1/2 - x_{z2}$ along edge $bc$,
and the remainder along edges $ab$ and $ac$. The largest value of $x_{h3}$
possible via this APLO strategy varies linearly with $x_{z2}$. When $x_{z2} =
1/2$, then $(x_{h3} - x_{h2})$ is bounded above by 

$$
\left(\|m_c - s_c\|\right)
\frac{\cos(\pi/6-\theta^*)}{\cos(\theta^*)} =
\frac{3 - \sqrt{5}}{4};
$$

and when $x_{z2} = 0$, then $(x_{h3} - x_{h2})$ is bounded above by

$$
\left(\|m_c - s_c\| + \frac{\sqrt{3}}{2}\|s_c-t_a\|\right)
\frac{\cos(\pi/6-\theta^*)}{\cos(\theta^*)} = 1/4; 
$$

so 
$x_{h3} \le \frac{1}{4} + x_{z2}\left(\frac{9}{2} -
2\sqrt{5}\right)$.
Using this relation and the fact that $x_{h2} = x_{z2}d_s$, yields:

\begin{gather}
\begin{align*}
(p_{h3} - t_a)\cdot\hat{v} -
(p_{z3} - t_a)\cdot\hat{v}
&= x_{h3} -
\frac{1}{2}\left(
\frac{(x_{h3}-x_{h2})r}{\sin(\pi/3 + \theta^*)} - \left(\frac{1}{2} - x_{z2}\right)
- 1\right)
\\
&\ge
\left(1-\frac{r}{r^*}\right)\left(1 - 2x_{z2}\left(\sqrt{5}-2\right)\right),
\end{align*}
\end{gather}
which is always strictly positive for $r = r^* - \varepsilon < r^*$ and $0 \le
x_{z2} \le \frac{1}{2}$, as desired.

\item (Pursuer leaves $bc$ through $b$): pursuer leaves clockwise from
$p_{z2}$, so $(p_{h3} - p_{h2})\cdot\hat{v} = x_{h3} - x_{h2} \le 0$. Then by
APLO, the pursuer travels clockwise from $p_{z2}$ by distance $r(x_{h2} -
x_{h3})/\sin(\pi/3+\theta^*)$, for distance $1/2 + x_{z2}$ along edge $bc$, and
the remainder along edges $ab$ and $ac$. The smallest value of $x_{h3}$
possible via this APLO strategy varies linearly with $x_{z2}$. When $x_{z2} =
0$, then $(x_{h3} - x_{h2})$ is bounded below by
$$
-\left(\|m_c - s_c\| + \frac{\sqrt{3}}{2}\|s_c-t_a\|\right)
\frac{\cos(\pi/6-\theta^*)}{\cos(\theta^*)} = -1/4; 
$$
and when $x_{z2} = 1/2$, then $(x_{h3} - x_{h2})$ is bounded below by
$$
-\left(\|m_c - s_c\| + \sqrt{3}\|s_c-t_a\|\right)
\frac{\cos(\pi/6-\theta^*)}{\cos(\theta^*)} = -\frac{\sqrt{5} - 1}{4}; 
$$
so 
$x_{h3} \ge -\frac{1}{4} + x_{z2}\left(\frac{9}{2} -
2\sqrt{5}\right)$.
Using this relation and the fact that $x_{h2} = x_{z2}d_s$, yields
\begin{gather}
\begin{align*}
(p_{h3} - t_a)\cdot\hat{v} -
(p_{z3} - t_a)\cdot\hat{v}
&= x_{h3} -
\frac{1}{2}\left(
-\frac{(x_{h2}-x_{h3})r}{\sin(\pi/3 + \theta^*)} + \left(\frac{1}{2} + x_{z2}\right)
+ 1\right)
\\
&\le
-\left(1-\frac{r}{r^*}\right)\left(1 - 2x_{z2}\left(\sqrt{5}-2\right)\right),
\end{align*}
\end{gather}
which is always strictly negative for $r = r^* - \varepsilon < r^*$ and $0 \le
x_{z2} \le \frac{1}{2}$, as desired. \qedhere
\end{enumerate}
\end{proof}

\subsection{Square}
\label{sec:square}

The square is perhaps the next simplest polygon after the equilateral triangle.
We show how to extend our exact techniques for this polygon as well:

\begin{theorem} 
If the escaper domain is a square, the critical
speed ratio is \\
$r^* = \sqrt{\frac{5}{2}(7 + \sqrt{41})} \approx 5.789$. 
\label{thm:square}
\end{theorem}

Similar to the triangle case, the speed ratio $r^*$ is chosen such that, if the
pursuer is a particular position $a$ (in this case at the midpoint of a side)
and the escaper is at point $s_a$ at distance $(9 - \sqrt{41})/4\approx 0.6492$
along the perpendicular bisector at $a$, then the escaper has four simultaneous
threats to exit at $p$, $p'$, $q$, and $q'$; see Figure~\ref{SquareRatioFigure}.
Specifically, the escaper distance from $s_a$ to $p$ is exactly factor $r^*$
smaller than the pursuer distance counterclockwise from $a$ to $b$ to $c$ to
$p$, i.e., $r^*\|s_a-p\| = 2 + \|c-p\|$; and the escaper distance from $s_a$ to
$q$ is exactly a factor $r^*$ smaller than the pursuer distance from $a$ to $b$
to $c$ to $q$, i.e., $r^*\|s_a-q\| = 3 - \|d-q\|$; and similarly for $p'$ and
$q'$ in the clockwise direction. 

\begin{figure}
\centering 
\includegraphics[scale=0.5]{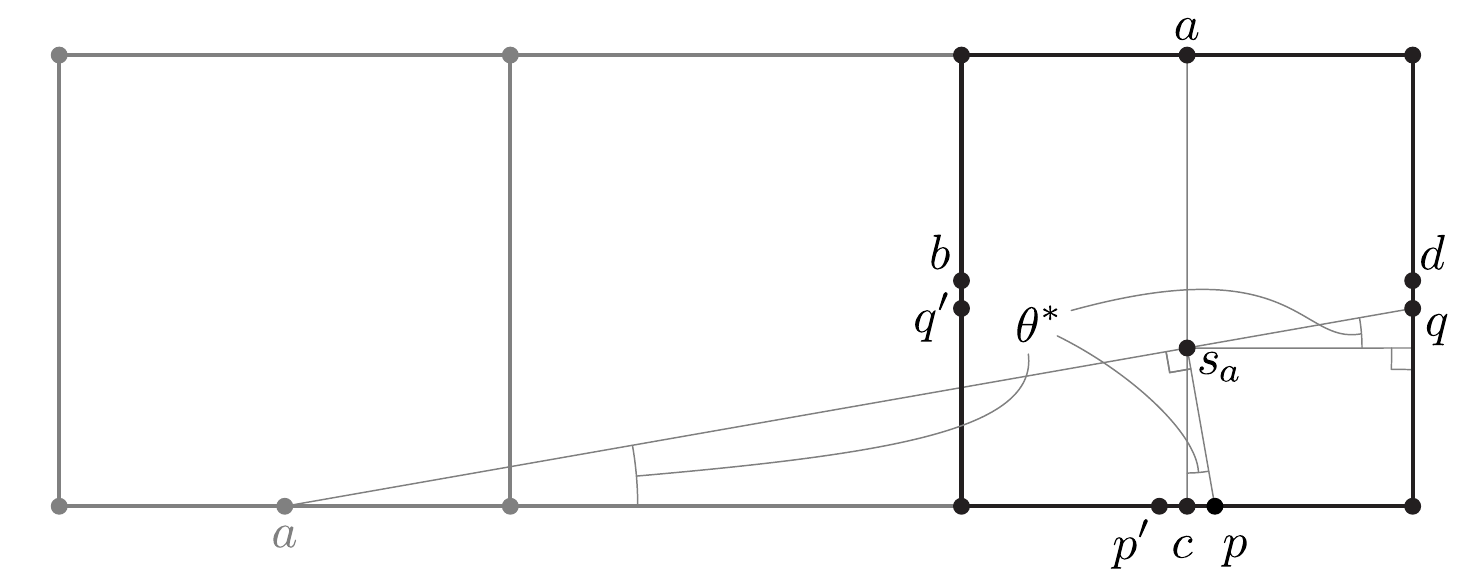}
\caption{Geometry for computing the critical speed ratio $r^* = 1/\sin\theta^*$
for a square.}
\label{SquareRatioFigure} 
\end{figure}

\begin{proof} We first provide a winning pursuer strategy when $r \geq r^*$. Our
pursuer strategy transitions between eight different strategies as the escaper
move within the triangle. These six strategies $z(h; i, j)$ are shown in
Figure~\ref{SquareTransitionsFigure}, where each strategy is associated with a
corner $i\in\{a,b,c,d\}$ and a sign $j\in\{-1, 1\}$. Each of these strategies is
identical up to rotations and reflections, so let us first focus on one of the
strategies, $z(h; a, 1)$. 

\begin{figure}
\centering 
\includegraphics[scale=0.15]{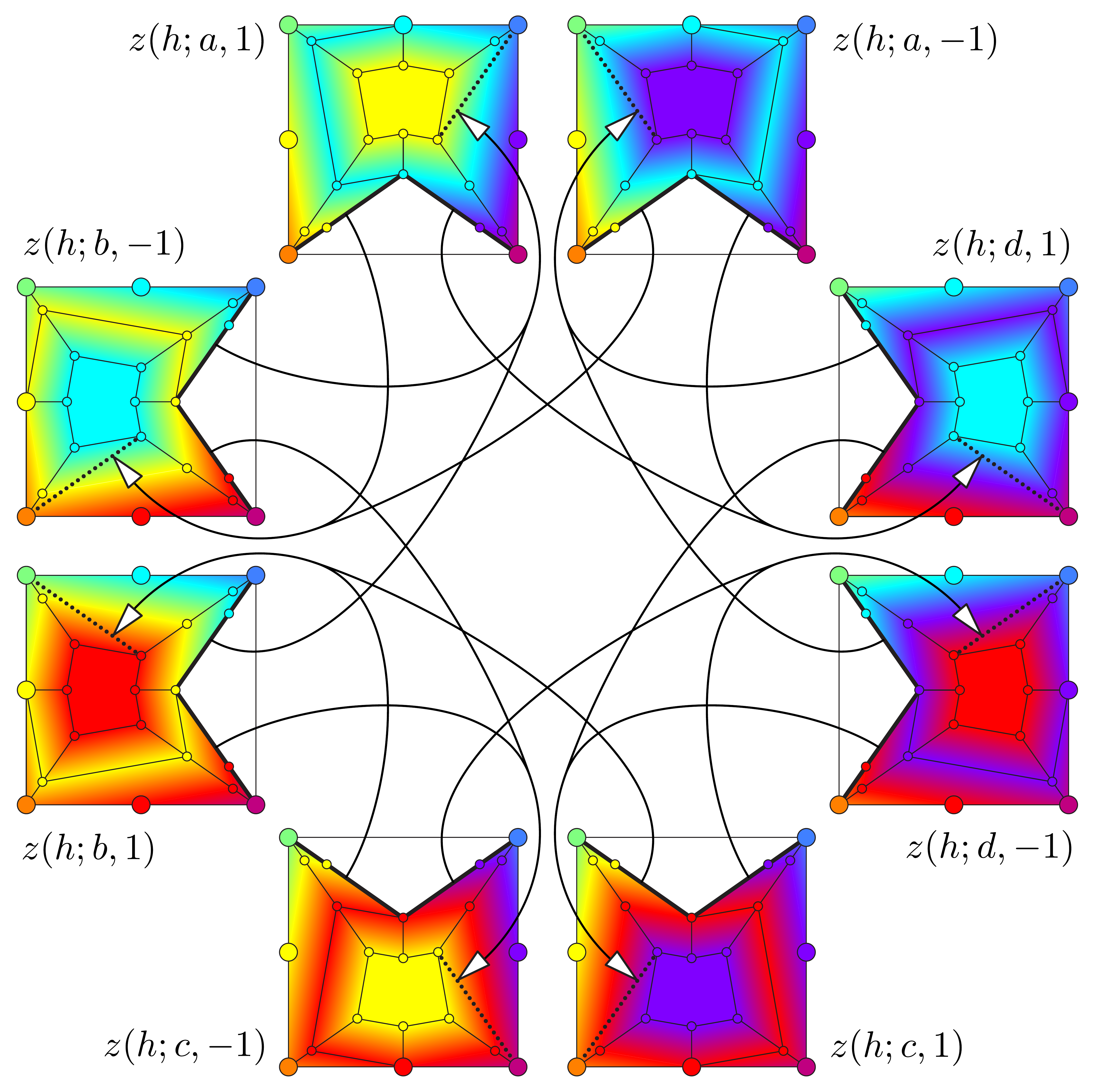}
\caption{Transitions between pursuer strategies.} 
\label{SquareTransitionsFigure} 
\end{figure}

The strategy $z(h; a, 1)$, depicted in
Figure~\ref{SquareZombieFigure}, maps each point of the colored subset of the
square to a point on the boundary via a piecewise-linear map. Wherever the
escaper is in the colored region of a strategy, the strategy will place the
pursuer at the boundary point designated by the map. To make it easier to
reference points on the boundary, we map each boundary point on edges
$p_{ab}p_{bc}$,
$p_{cd}p_{da}$, and
$p_{da}p_{ab}$
to a number, varying linearly from 
$-1.5$ at vertex $p_{bc}$ (orange), to 
$-1$ at vertex $b$ (yellow), to 
$-0.5$ at vertex $p_{ab}$ (green), to 
$0$ at vertex $b$ (cyan), to 
$0.5$ at vertex $p_{da}$ (blue), to 
$1$ at vertex $c$ (purple), to 
$1.5$ at vertex $p_{cd}$ (magenta).
The left drawing of
Figure~\ref{SquareZombieFigure} depicts the geometry of the linear patches of
this map: 

\begin{figure}
\centering 
\includegraphics[scale=0.5]{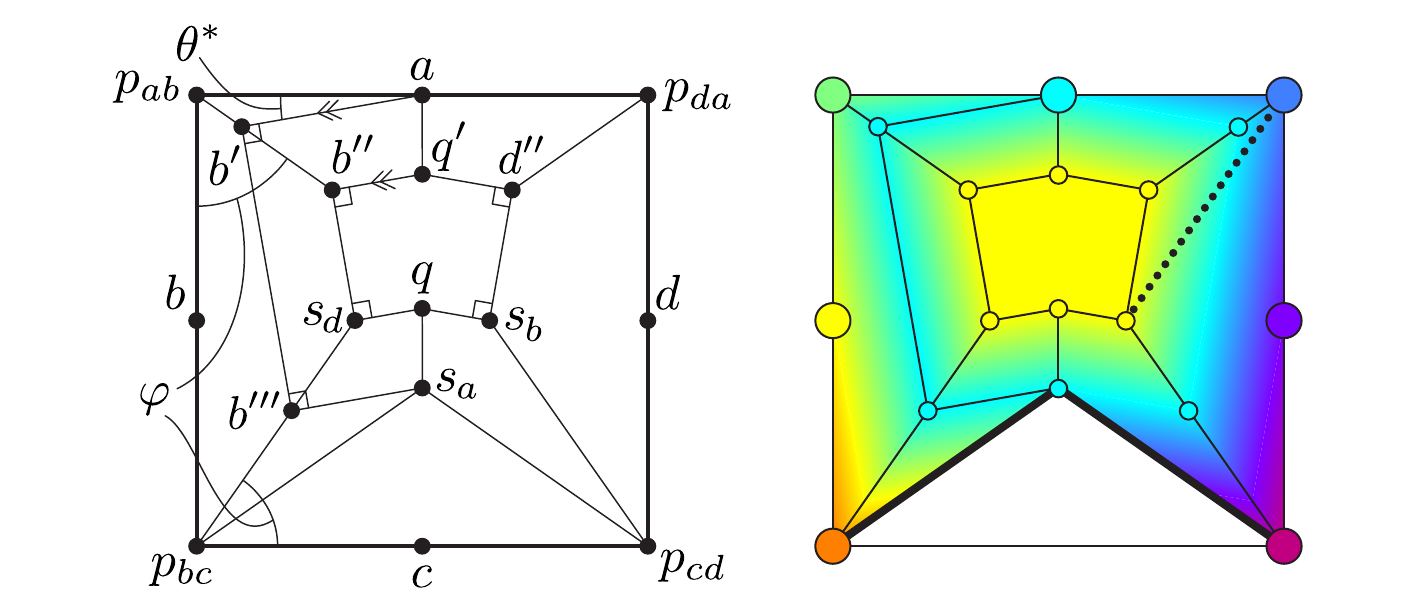}
\caption{Geometry of $\pursuer(\escaper; a, 1)$.
This function is linear in each region $ab'p_{ab}$, $aq'b''b'$,
$p_{ab}p_{bc}b'''b'$, $b'b'''s_db''$, $p_{bc}s_ab'''$, $b'''s_aqs_d$,
$q'b''s_dqs_bd''$, $aq'd''p_{da}$, $d''s_bp_{cd}p_{da}$, and $qs_ap_{cd}s_b$, where points $\{a, b', b''', s_a\}$ have value $0$ (cyan),
points $\{b,b'',s_d,q,q',s_b,d''\}$ have value $-1$ (yellow), and points 
$p_{ab}$, $p_{bc}$, $p_{cd}$, $d$, $p_{da}$ have values
$-0.5$ (green),
$-1.5$ (orange),
$1.5$ (magenta),
$1$ (purple), and
$0.5$ (blue) respectively.}
\label{SquareZombieFigure} 
\end{figure}

\begin{itemize}
\item point $s_i$ for $i\in\{a, b, c, d\}$ is distance
$(9 - \sqrt{41})/4 \approx 0.06492$
from midpoint $i$ toward the center;
\item point $b'''$ is the point on segment $p_{bc}s_d$ where 
$5\|b'''-s_d\| = 2\|s_d - p_{bc}\|$;
\item point $q$ is the intersection of the segment $as_a$ and the line
through $s_d$ parallel to segment $b'''s_a$.
\item point $q'$ is on segment $as_a$ such that $\|a-q'\|=\|q-s_a\|$;
\item point $b''$ is the intersection of two lines: the line through $t$
parallel to segment $b'''s_a$ and the line through $s_d$ perpendicular to
segment $b'''s_a$;
\item point $d''$ is the reflection of $b''$ about $as_a$; and 
\item point $b'$ is the point on segment $p_{ab}b''$ where 
$3\|p_{ab}-b'\| = \|p_{ab}-b''\|$.
\end{itemize}

We specify each linear patch by specifying the value at each vertex:

\begin{itemize}
\item points $\{a, b', b''', s_a\}$ have value $0$ (cyan),
\item points $\{b, q, q', b'', d'', s_d, s_d\}$ have value $-1$ (yellow),
\item point $d$ has value $1$ (purple), and
\item points $p_{ab}$, $p_{bc}$, $p_{cd}$, and $p_{da}$ have values 
$-0.5$, $-1.5$, $1.5$, and $0.5$ respectively.
\end{itemize}

This map has the property that the gradient at every point within each linear
patch has the same value, namely $r^*$. Thus, as the escaper moves within the
colored region, the pursuer's speed will always stay below $r^* \leq r$, so the
strategy will be valid. This map also has the property that the pursuer and the
escaper will be collocated whenever the escaper is on edges $p_{ab}p_{bc}$,
$p_{cd}p_{da}$, or $p_{da}p_{ab}$, so the escaper cannot win along those edges.
If the escaper reaches edge $p_{bc}s_a$ or edge $p_{cd}s_a$, the pursuer will
switch strategies, respectively to either $z(h; b, -1)$ or $z(h; d, 1)$. These
strategies exactly match strategy $z(h; a, 1)$ along their respective transition
edges. By transitioning between these strategies via the transition graph shown
in Figure~\ref{SquareTransitionsFigure}, the pursuer will always be collocated
with the escaper whenever the escaper is at the boundary, as desired.

Next, we provide a winning escaper strategy when $r = r^* - \epsilon$ for any
positive $\epsilon$; refer to Figure~\ref{fig:SquareHumanFigure}.
Our escaper strategy follows a similar strategy as the
triangle escaper strategy: reach a state where the escaper can win via a single
APLO strategy. In particular, when the escaper is on the boundary of square $S =
s_as_bs_cs_d$, e.g., at some point $p_h$ on $s_bs_c$, and the pursuer is
antipodal, e.g., at point $p_z$ on the boundary between $cd$ where $d_z(c,
p_z)/d_z(d, p_z) = \|s_c-p_h\|/\|s_d-p_h\|$ (recall, $d_z(u, v)$ corresponds to
the distance between $u$ and $v$ in the pursuer metric), then the escaper will
be able to win via an APLO strategy to the boundary. We will reach such a
configuration in two phases.

\begin{figure}
\centering 
\includegraphics[scale=0.45]{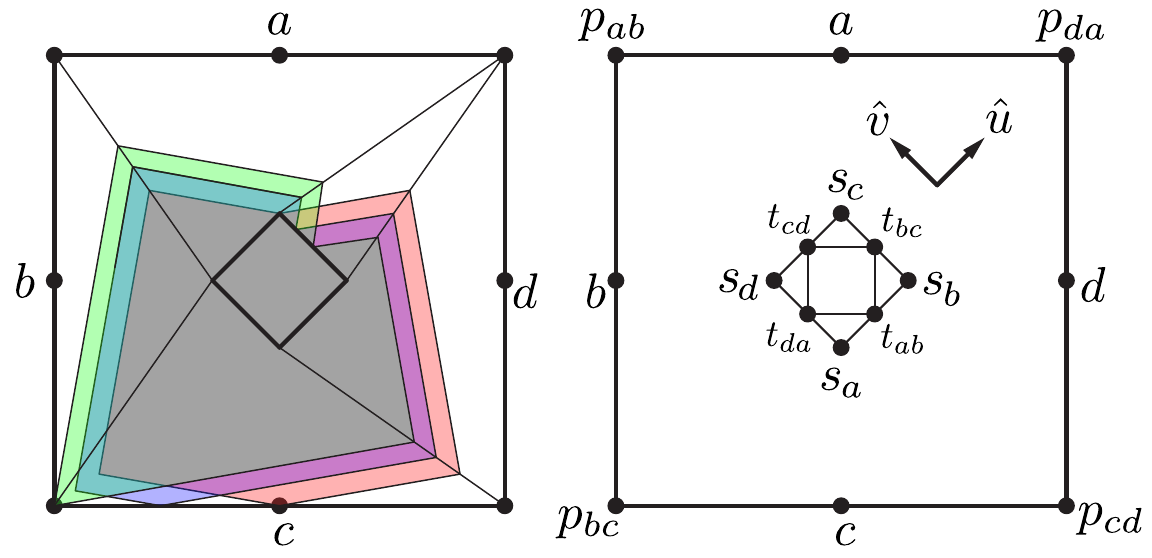}
\caption{Geometry of the escaper strategy for a square.} 
\label{fig:SquareHumanFigure} 
\end{figure}

In the first phase, the escaper starts anywhere on $S' =
t_{ab}t_{bc}t_{cd}t_{da}$, the square formed by connecting the midpoints of
square $S$. The perimeter of $S'$ has length $7 - \sqrt{41} \approx 0.5969$
which is less than $4/r^* \approx 0.6910$, so the escaper can run around $S'$
faster than the pursuer can run around the boundary. The escaper runs around
$S'$ until the escaper reaches a position $p_{h1}$ on $S'$ such that the
pursuer's position $p_{z1}$ is antipodal. Without loss of generality, assume
$p_{h1}$ is on edge $t_{cd}t_{bc}$ and $p_{z1}$ is on edge $p_{bc}p_{cd}$ such
that $\|p_{bc}-p_{z1}\| = \|t_{bc} - p_{h1}\|/\|t_{bc}-t_{cd}\|$.

Now that the escaper is antipodal to the pursuer on $S'$, the escaper enters the
second phase, executing an APLO strategy $H_{APLO}(z, t; p_{h1}, \hat{a}, r,
d_u, d_v)$ where $\hat{a}$ is the unit direction from $c$ to $a$, $d_v =
\|t_{bc} - t_{cd}\|r < 1$, and $d_u = \sqrt{1 - d_v^2} < 1$, until the escaper
reaches square $S$ at some point $p_{h2}$ (without loss of generality, assume
$p_{h2}$ is on edge $t_{bc}s_c$). By definition of APLO, during this process the
escaper's projection onto edge $t_{bc}t_{cd}$ remains antipodal to the pursuer,
so when the escaper reaches $p_{h2}$, the pursuer is at the point $p_{z2}$ on
edge $cp_{cd}$ that is also antipodal to $p_{h2}$ on $S$. Without loss of
generality, assume $p_{h2}$ is on segment $t_{bc}s_c$ and $p_{z2}$ is antipodal
on segment $cp_{bc}$ such that $\|p_{bc} - p_{z1}\|/1 = \|t_{bc} -
p_{h2}\|/\|s_b-s_c\|$. Let $x_{z2} = \|t_{bc} - p_{z1}\|$, let $d_s =
\|s_b-s_c\| = \sqrt{2}(7-\sqrt{41})/4$, and let $x_{h2} = \|t_{bc} - p_{h2}\| =
x_{z2}d_s$.

Now that the escaper is antipodal to the pursuer on $S$, the escaper enters the
third and final phase, executing an APLO strategy $H_{APLO}(z, t; p_{h2},
\hat{u}, r^*, d_u, d_v)$, where $\hat{u}$ is the unit direction from $p_{bc}$ to
$p_{da}$, $d_u= \cos(\pi/4 + \theta^*)$, and $d_v = \sin(\pi/4 + \theta^*)$,
until the escaper reaches the boundary at some point $p_{h3}$. 

If the pursuer remains in the halfplane $H$ bounded by $p_{ab}p_{cd}$ containing
$p_{bc}$, the escaper wins easily as $p_{h3}$ is in the other halfplane.
Otherwise, there are two cases: the pursuer leaves $H$ last through either
$p_{ab}$ or $p_{cd}$. It suffices to show that the separation of their
projections onto segment $p_{ab}p_{cd}$ is bounded away from zero, specifically
quantity $|(p_{h3} - t_{bc})\cdot\hat{v} - (p_{z3} - t_{bc})\cdot\hat{v}|$ where
$\hat{v}$ is the unit vector $(p_{cd} - p_{ab})$. Let $x_{h3} = (p_{h3} -
t_{bc})\cdot\hat{v}$. Note that $x_{h3}$ is positive to the upper-left of
$t_{bc}$.

\begin{enumerate}

\item (Pursuer leaves $H$ through $p_{cd}$): pursuer leaves counter-clockwise from
$p_{z2}$, so $(p_{h3} - p_{h2})\cdot\hat{v} = x_{h3} - x_{h2} \ge 0$. Then by
APLO, the pursuer travels counter-clockwise from $p_{z2}$ by distance $r(x_{h3}
- x_{h2})/\sin(\pi/4+\theta^*)$, for distance $1 - x_{z2}$ along edge
$p_{bc}p_{cd}$, and the remainder along edges $p_{cd}p_{da}$ and $p_{ab}p_{da}$.
The largest value of $x_{h3}$ possible via this APLO strategy varies linearly
with $x_{z2}$. When $x_{z2} = 1/2$, then $(x_{h3} - x_{h2})$ is bounded above by
$$
(\|s_c - a\|)
\frac{
\cos(\pi/4 - \theta^*)
}{\cos(\theta^*)} =
\frac{
\sqrt{2}(13-\sqrt{41})}
{32}
$$
and when $x_{z2} = 0$, then $(x_{h3} - x_{h2})$ is bounded above by
$$
\left(\|s_c - a\| + \frac{1}{\sqrt{2}}\|s_c-t_{bc}\|\right)
\frac{
\cos(\pi/4 - \theta^*)
}{\cos(\theta^*)} =
\frac{\sqrt{2}}{4}
$$
so 
$x_{h3} \le \frac{\sqrt{2}}{16}\left(4
+ x_{z2}\left(33 -
5\sqrt{41}\right)\right)$.
Using this relation and the fact that $x_{h2} = x_{z2}d_s$, yields
\begin{gather}
\begin{align*}
(p_{h3} - t_{bc})\cdot\hat{v} -
(p_{z3} - t_{bc})\cdot\hat{v}
&= x_{h3} -
\frac{\sqrt{2}}{2}\left(
\frac{(x_{h3}-x_{h2})r}{\sin(\pi/4 + \theta^*)} - \left(1 - x_{z2}\right)
- 1\right)
\\
&\ge
\frac{5\sqrt{2}}{16}\left(1-\frac{r}{r^*}\right)(4-x_{z2}(\sqrt{41}-5))
\end{align*}
\end{gather}
which is always strictly positive for $r = r^* - \varepsilon < r^*$ and $0 \le
x_{z2} \le \frac{1}{2}$, as desired.

\item (Pursuer leaves $H$ through $p_{ab}$): pursuer leaves clockwise from
$p_{z2}$, so $(p_{h3} - p_{h2})\cdot\hat{v} = x_{h3} - x_{h2} \le 0$. Then by
APLO, the pursuer travels clockwise from $p_{z2}$ by distance $r(x_{h2} -
x_{h3})/\sin(\pi/4+\theta^*)$, for distance $1 + x_{z2}$ along edges
$p_{bc}p_{cd}$ and $p_{ab}p_{bc}$, and the remainder along edges $p_{ab}p_{da}$
and $p_{da}p_{cd}$. The smallest value of $x_{h3}$ possible via this APLO
strategy varies linearly with $x_{z2}$. When $x_{z2} = 0$, then $(x_{h3} -
x_{h2})$ is bounded below by
$$
-\left(\|s_c - a\| + \frac{1}{\sqrt{2}}\|s_c-t_{bc}\|\right)
\frac{
\cos(\pi/4 - \theta^*)
}{\cos(\theta^*)} =
-\frac{\sqrt{2}}{4}
$$
and when $x_{z2} = 1/2$, then $(x_{h3} - x_{h2})$ is bounded below by
$$
-\left(\|s_c - a\| + \frac{2}{\sqrt{2}}\|s_c-t_{bc}\|\right)
\frac{
\cos(\pi/4 - \theta^*)
}{\cos(\theta^*)} =
-\frac{\sqrt{2}(3+\sqrt{41})}{32}
$$
so $x_{h3} \le -\frac{\sqrt{2}}{16}\left(4
- x_{z2}\left(33 -
5\sqrt{41}\right)\right)$.
Using this relation and the fact that $x_{h2} = x_{z2}d_s$, yields
\begin{gather}
\begin{align*}
(p_{h3} - t_{bc})\cdot\hat{v} -
(p_{z3} - t_{bc})\cdot\hat{v}
&= x_{h3} -
\frac{\sqrt{2}}{2}\left(
-\frac{(x_{h2}-x_{h3})r}{\sin(\pi/4 + \theta^*)} + \left(1 + x_{z2}\right)
+ 1\right)
\\
&\le
-\frac{5\sqrt{2}}{16}\left(1-\frac{r}{r^*}\right)(4-x_{z2}(\sqrt{41}-5))
\end{align*}
\end{gather}
which is always strictly negative for $r = r^* - \varepsilon < r^*$ and $0 \le
x_{z2} \le \frac{1}{2}$, as desired. \qedhere
\end{enumerate}
\end{proof}

\section{Full Model (Full Version of Section~\ref{sec:model})}
\label{appendix:model}

\iftrue
In this section, we define our model (as in Section~\ref{sec:model}),
as well as detail the motivation for the particular definitions and
the differences from past work, in Section~\ref{Continuous Game}.
Then
\else
After defining the game precisely in Section~\ref{Continuous Game},
\fi
we prove that at most one player can win in
Section~\ref{sec:BothPlayersCannotWin},
and prove that at least one player can win in
Section~\ref{Some Player Wins}.
Along the way, we introduce two important tools for analyzing these games:
$\delta$-oblivious strategies (Section~\ref{sec:oblivious strategies}) and
the $(\delta,\gamma)$-discretized game (Section~\ref{Discrete Game}).
The latter will be useful in particular for our pseudopolynomial-time
approximation scheme in \Section~\ref{sec:pseudoPTAS}.

\subsection{Continuous Game}
\label{Continuous Game}

To define the pursuit--escape game $G$, we need several ingredients:
what type of domains (regions) the escaper and pursuer traverse,
what type of motions are allowed within these domains,
what strategies are and how they can adapt to the other player's actions,
and when exactly a player wins the game.
We address each of these concepts in turn.
\iftrue
The core definitions (the overlap with Section~\ref{sec:model}) are
\hl{highlighted in yellow}.
\fi

\paragraph{Domains.} %

First we define the notion of ``player domain'', which is a play area
that either the escaper or pursuer is restricted to move within.
We choose to somewhat restrict the sets on which we analyze pursuer
evasion games, to avoid cases where escaper or pursuer running distances
(shortest-path metrics) are undefined or behave pathologically.
Even so, we give a very general definition, both to show our framework
applies very generally and so that it includes the many special cases
of interest, including a disk (with smooth boundaries), a halfplane (with
unbounded area), and the graph model (with one-dimensional features),
in addition to our primary case of a polygon with the exterior or moat model.
(Previous work on the Lion and Man game
did not deal with the issue of defining allowable domains, focusing on specific
cases, although the importance of rectifiability is mentioned in the
context of two-lion games in \cite[p.~46]{Bollobas-2006} and
\cite{Abrahamsen-Holm-Rotenberg-Wulff-Nilsen-2017}.)

Specifically, \hl{a \defn{player domain} is a closed subset $D$ of Euclidean space
$\mathbb R^k$
that is \defn{locally finitely} \defn{rectifiable},
meaning that its intersection $D \cap B$
with any bounded closed Euclidean ball $B$
is ``finitely rectifiable''} (which intuitively
means ``finite total surface area'').  \hl{Formally,
$R \subseteq \mathbb R^k$ is \defn{finitely rectifiable} if it is
the union of the images of finitely many functions of the form
$S : [0,1]^k \to R$ satisfying the Lipschitz condition
$d(S(u),S(v)) \leq d(u,v)$ for all $u, v \in [0,1]^k$.}%
\footnote{Throughout this paper, we use Euclidean as
  the default metric unless otherwise specified, so $d(u,v)$ denotes
  the Euclidean distance $\|u-v\|_2$. We use a subscript (such as
  $d_\escaper$ and $d_\pursuer$ introduced soon) to denote a different metric.}
We call the functions $S$ constituting $R$ the \defn{patches} of~$R$.

This definition forbids player domains with fractal boundary of
nontrivial fractal dimension, and forbids the ``Harmonic comb'' ---
the union of line segments from the origin to $(1/i,1)$ for all
$i > 0$, together with the segment from the origin to $(0,1)$.%
\footnote{The Harmonic comb would have been allowed if we required the
  weaker property that $D$ is the union of the images of countably
  many Lipschitz functions (the countable analog of ``finitely rectifiable'').
  Notably, this compact set has an infinite sequence of points $(1/i,1)$ that
  converge in the Euclidean metric but not when measured according to
  shortest paths within $D$ (contrary to Lemma~\ref{lem:intrinsic
    metric compact}), so we choose to forbid it from being a valid domain.}
But the definition still allows a boundary of infinite total
length/surface area so long as the infinity comes from being unbounded
in $\mathbb R^k$.  For example, the following are valid domains:
\begin{itemize}
\item Polygons (interior plus boundary), possibly with holes,
  of finite total perimeter (but having possibly infinitely many edges).
\item Unbounded polygons, where finitely many edges extend to infinite rays,
  while the finite-length edges have bounded total length.
  For example, 2D linear programs define convex unbounded polygons,
  including half-planes and wedges
  (which are studied in \Section~\ref{sec:exact}).
\item The exterior (including the boundary) of one or more polygons,
  each of finite perimeter.
\item Generalizations of the above to higher dimensions (polyhedra).
\item Any closed semi-algebraic set, or more generally,
  closed semi-analytic or closed subanalytic set
  \cite[Theorem~6.10]{Bierstone-Milman-1988}.
\item Any embedding (not necessarily straight-line) of a graph
  into $\mathbb R^k$ of finite total edge length.
  In particular, any graph can be embedded into $\mathbb R^3$,
  even while matching specified edge lengths,
  so this lets us represent the pursuit--escape game on weighted graphs
  (the \defn{graph model}).
  In this case, the entire domain is its own boundary.
\end{itemize}

\hl{The input to the pursuit--escape problem consists of an \defn{escaper
  domain} $D_\escaper$ and a \defn{pursuer domain} $D_\pursuer$, and an \defn{exit
  set} $\Exit$.  The escaper and pursuer domains must be \emph{player domains} as
described above.
The exit set $X$ must also be a player domain,
and a subset of the player domains: $\Exit \subseteq D_\escaper \cap D_\pursuer$.
The goal of the escaper will be to reach an \defn{exit} --- any point of the
exit set $\Exit$ --- while being sufficiently away from the pursuer.}
Typically, we imagine the entire escaper--pursuer domain intersection as the exit
set ($\Exit = D_\escaper \cap D_\pursuer$), but we allow the more general form to represent
e.g.\ that the escaper must reach an escape vehicle which are only at certain
points where the escaper and pursuer could meet.

Two natural cases captured by this framework are as follows:
\begin{itemize}
\item \defn{Exterior model}: when $D_\pursuer = \overline{\mathbb R^k - D_\escaper}$
  (the closure of the complement of~$D_\escaper$), i.e., the pursuer can be
  anywhere the escaper cannot, plus the shared boundary
  $\partial D_\pursuer = \partial D_\escaper$.%
  \footnote{Here $\partial D = D \setminus \interior D$ is the \defn{boundary}
    of~$D$, where $\interior D$ is the \defn{interior} of $D$, i.e., the
    set of all points of $D$ having an open neighborhood within~$D$.}
\item \defn{Moat model}: when $D_\pursuer = \partial D_\escaper$, i.e.,
  the pursuer can only walk around the boundary of the escaper domain.
\end{itemize}

For any domain $D$, let $d_D$ denote the
\defn{intrinsic (shortest-path) metric} of~$D$.
This metric measures how someone restricted to the domain would travel.
In particular, define the \defn{escaper metric} $d_\escaper = d_{D_\escaper}$ and
\defn{pursuer metric} $d_\pursuer = d_{D_\pursuer}$.

\paragraph{Motion paths.}

\hl{A \defn{motion path} with maximum speed $s \geq 0$ in metric domain $D$
is a function $\player : [0,\infty) \to D$
satisfying the \defn{speed-limit constraint} (Lipschitz condition)}
$$ \colorbox{hl}{$\displaystyle
d_D(\player(t_1), \player(t_2)) \leq s \cdot |t_1 - t_2| \text{ for all }t_1,t_2 \geq 0.
$}
$$
(This definition matches the definitions of ``lion path'' and ``man path''
in \cite{Bollobas-Leader-Walters-2012}, generalized to arbitrary maximum speed
and arbitrary domain.)
The speed constraint implies that all motion paths are continuous.
This definition can also represent finite motion paths
by letting $\player(t)$ be constant for $t \geq T$ for some~$T$.

\hl{We consider a model where the pursuer maximum speed is a factor of $r$
larger than the escaper maximum speed, which we assume is $1$ for simplicity.
Thus an \defn{escaper motion path} is a motion path of maximum speed $1$ in
the escaper domain $D_\escaper$,
while a \defn{pursuer motion path} is a motion path of maximum speed $r$ in
the pursuer domain $D_\pursuer$.}

\paragraph{Symmetric terminology for player vs.\ opponent.}

\hl{For symmetry, the following definitions refer to a \defn{player}
(either escaper and pursuer) and their \defn{opponent} (pursuer or escaper,
respectively).}  For example, we use ``player motion path'' $\player$ and
``opponent motion path'' $\opponent$ to refer to two cases symmetrically:
(1)~an escaper motion path $\player$ and a pursuer motion path~$\opponent$; and
(2)~a pursuer motion path $\player$ and an escaper motion path~$\opponent$.

\paragraph{Strategies.}

\hl{A \defn{player strategy} is a function $\Player$ mapping an opponent motion path
$\opponent$ to a player motion path $\Player(\opponent)$ satisfying the following \defn{nonbranching-lookahead
constraint}:}
\begin{quote}
\hl{for any two opponent motion paths $\opponent_1,\opponent_2$ agreeing on $[0,t]$,
the strategy's player motion paths $\Player(\opponent_1),\Player(\opponent_2)$ also agree on $[0,t]$.}
\end{quote}
Effectively, this definition constrains $\Player(\opponent)(t)$ to depend only on
$\opponent(t')$ for earlier times $t' \leq t$, or equivalently by continuity of
motion plans, for strictly earlier times $t' < t$.

This definition matches the clever definition of ``lion/man strategy''
and ``no lookahead'' in \cite{Bollobas-Leader-Walters-2012}.
We use the term ``nonbranching-lookahead'' to more accurately
reflect that the strategy can depend on the opponent motion path, including
the future, so long as it does so in a nonbranching way.
This is useful for defining strategies such as ``move along a straight line
to where the opponent will go'', but it can allow for certain kinds of
``cheating''; see Lemma~\ref{lem:terrible} below.

This definition correctly defines a \defn{pursuer strategy}~$\Pursuer$.
\hl{An \defn{escaper strategy}~$\Escaper$ must satisfy one additional constraint,
the \defn{escaper-start constraint}:}
\begin{quote}
\hl{all paths $\Escaper(\pursuer)$ (over all pursuer motion paths~$\pursuer$)
must start at a common point $\Escaper(\pursuer)(0)$.}
\end{quote}
This constraint is necessary in our case because, if the escaper can choose their
starting position depending on the pursuer's start position, then the escaper
can trivially win (by starting at a far-away exit).
(In the man-and-lion problem, the man and lion's starting positions
are given, so \cite{Bollobas-Leader-Walters-2012}
did not have to deal with this asymmetry.)

Notationally, we use lower-case letters $\player,\escaper,\pursuer$ for motion paths
and upper-case letters $\Player,\Escaper,\Pursuer$ for strategies of the player, escaper, and
pursuer, respectively.

\paragraph{Win condition.}

It remains to define a win condition for the pursuit--escape game~$G$.
We do so in terms of \hl{an infinity family of games $G_\epsilon$
for all $\epsilon > 0$.}

\hl{An escaper strategy $\Escaper$ \defn{wins $G_\epsilon$} or
\defn{wins $G$ by $\epsilon$} if,
for every pursuer motion path $\pursuer$,
there is a time $t$ at which $\Escaper(\pursuer)(t)$ is on an exit and at distance
$\geq \epsilon$ from $\pursuer(t)$ in the pursuer metric.}
Intuitively, the escaper needs a small amount of time to exit
(e.g., to break into the getaway car),
during which the pursuer can run $\epsilon$ distance and catch the escaper.

This \defn{no-capture} definition allows the escaper and pursuer
to collocate at time $< t$
without the escaper being captured; in other words, the escaper has the
ability to choose to exit, and only then must be away from the pursuer.
As mentioned in Section~\ref{IntroductionSection}, our no-capture model differs
from the Lion and Man problem, where collocation implies immediate capture.
Indeed, our no-capture model is a significant deviation because,
if we used the Lion-and-Man notion of ``escaper win''
\cite{Bollobas-Leader-Walters-2012},
then the escaper would always win in many natural instances
(e.g., polygon, Jordan, and polyhedron models):

\begin{lemma} \label{lem:terrible}
  Assuming the exit set $\Exit$ contains a one-dimensional curve,
  there is an escaper strategy $\Escaper$ such that,
  for any pursuer motion path $\pursuer$,
  $\Escaper(\pursuer)$ wins $G_{\epsilon(\pursuer)}$
  for some function $\epsilon(\pursuer)$.
\end{lemma}

\begin{proof}
  Parameterize the curve as $C(t)$ for $0 \leq t \leq T$
  with unit speed in the escaper metric $d_\escaper$.
  The escaper starts at $C(0)$, i.e., $\Escaper(\pursuer)(0) = C(0)$.
  Thus $\Escaper$ satisfies the escaper-start constraint.

  If $\pursuer(0) \neq C(0)$, then the escaper wins immediately
  by $d_\pursuer(C(0), \pursuer(0)) > 0$.
  So assume $\pursuer(0) = C(0)$.
  (The escaper can still continue from this position
  because of the no-capture aspect of our model.)
  Either the pursuer stays at $C(0)$ for positive time, or they move away.
  We define the rest of the escaper strategy according to these two cases:
  $$
  \Escaper(\pursuer)(t) =
    \begin{cases}
      C(t) & \text{if } \pursuer(t') = C(0) \text{ for all } t' \in [0,T']
                                            \text{ for some } T' > 0, \\
      C(0) & \text{if } \pursuer(t') \neq C(0) \text{ for some } t' \in [0,1].
    \end{cases}
  $$
  By the unit-speed parameterization of $C$,
  $\Escaper(\pursuer)$ is a valid escaper motion path.
  In the first case, the escaper wins by $d_\pursuer(C(0),C(T')) > 0$.
  In the second case, the escaper wins by $d_\pursuer(C(0),\pursuer(t')) > 0$.

  Finally, we prove that $\Escaper$ satisfies the
  nonbranching-lookahead constraint.
  Consider two pursuer motion paths $\pursuer_1,\pursuer_2$
  that agree on $[0,t]$ for some $t \geq 0$.
  If $t = 0$, then $\Escaper(\pursuer_1)$ and $\Escaper(\pursuer_2)$
  also agree on $[0,t]$ (by the escaper-start constraint).
  If $t > 0$ and $\pursuer_i(t') = C(0)$ for all $t' \in [0, T']$
  for some $T' > 0$, then $\pursuer_{3-i}(t') = C(0)$
  for all $t' \in [0, \min\{t,T'\}]$.
  Thus, if $t > 0$, then
  $\pursuer_1$ and $\pursuer_2$ are in the same case among the two cases,
  so $\Escaper(\pursuer_1)$ and $\Escaper(\pursuer_2)$ also agree on $[0,t]$.
\end{proof}

To avoid this problem, we use the following notion of an escaper win
for a pursuit--escape game~$G$.
\hl{The \defn{escaper wins $G$} if, for some $\epsilon > 0$,
there is an escaper strategy that wins $G$ by~$\epsilon$, i.e.,
wins~$G_\epsilon$.}
Notably, unlike Lemma~\ref{lem:terrible}, this condition requires a
\emph{uniform} $\epsilon$ for all pursuer motion paths.
Equivalently, we are taking a uniform limit of winning strategies
in the games~$G_\epsilon$ as $\epsilon \to 0$.
This is a key difference from the definitions for Lion and Man
in \cite{Bollobas-Leader-Walters-2012}; as we will show,
it implies the existence of ``oblivious'' strategies, which are a stronger form
of ``locally finite'' strategies from \cite{Bollobas-Leader-Walters-2009-arXiv},
and perhaps a more natural notion of ``no lookahead''.
Note that, for the Lion-and-Man game, the locally finite property is already
known to imply a unique winner \cite{Bollobas-Leader-Walters-2009-arXiv}.

\medskip

Next we define pursuer wins.
\hl{A pursuer strategy $\Pursuer$ \defn{wins $G_\epsilon$} if,
for every escaper motion path~$\escaper$,
and every time $t$ at which $\escaper(t)$ is on an exit,
$\escaper(t)$ is at distance $< \epsilon$ from $\Pursuer(\escaper)(t)$ in the pursuer metric:
$d_\pursuer(\escaper(t), \Pursuer(\escaper)(t)) < \epsilon$.}
Intuitively, such a pursuer strategy prevents the escaper
from winning by~$\epsilon$.
\hl{The \defn{pursuer wins $G$} if, for all $\epsilon > 0$,
there is a pursuer strategy that wins $G_\epsilon$.}
The latter definition allows the pursuer strategy to depend on~$\epsilon$,
and our proofs will rely on this.
Under the Axiom of Choice, however, it is equivalent to a simpler definition:

\begin{lemma} \label{lem:choice}
  Assuming the Axiom of Choice,
  the pursuer wins $G$ if and only if there is a pursuer strategy
  that, for all $\epsilon > 0$, wins $G_\epsilon$.
\end{lemma}

To prove this lemma,
we need a version of the Arzel\`a--Ascoli Theorem \cite{AA-wikipedia}.
This theorem is sometimes stated for bounded functions over bounded intervals
and guaranteeing uniform convergence;
we need a version over unbounded intervals and only local boundedness,
at the cost of guaranteeing only pointwise instead of uniform convergence.
Known generalizations
\cite[p.~231]{Kelley-1955},
\cite[Theorems~3.4.20 and~8.2.10]{Engelking-1989}
imply this version, but for clarity and completeness,
we translate the theorem and proof from topological language.
Our proof is roughly a subset of the proof described in \cite{AA-wikipedia}
(skipping the finite-cover step needed for uniform convergence
but which does not work for unbounded domains).

\begin{lemma}[Arzel\`a--Ascoli Theorem] \label{lem:pointwise A-A}
  For any metric domain~$D$,
  and for any sequence of functions $f_1, f_2, \ldots : [0,\infty) \to D$
  that satisfies the following two properties,
  there is a subsequence $f_{i_1}, f_{i_2}, \dots$ that converges pointwise.
  \begin{enumerate}
  \item \textbf{Uniformly locally bounded}: there is an origin $O$ and
    a function $M : [0,\infty) \to [0,\infty)$
    such that, for all $i$ and $t$,
    we have $d_D(O, f_i(t)) \leq M(t)$.
  \item \textbf{Uniformly equicontinuous}:
    for every $\epsilon > 0$, there is a $\delta > 0$ such that,
    for all~$i$,
    we have $|s - t| \leq \delta$ implies $d_D(f_i(s), f_i(t)) \leq \epsilon$.
  \end{enumerate}
\end{lemma}

\begin{proof}
  Fix an enumeration $t_1, t_2, \dots$ of the nonnegative rational numbers.
  Start by applying every function to $t_1$,
  forming the sequence $f_1(t_1), f_2(t_1), \dots$.
  By uniform local boundedness, this sequence is bounded,
  so by the Bolzano--Weierstrass Theorem,
  it has a convergent subsequence $f_{i_{1,1}}(t_1), f_{i_{1,2}}(t_1), \dots$.
  Now change the parameter from $t_1$ to $t_2$, forming the sequence
  $f_{i_{1,1}}(t_2), f_{i_{1,2}}(t_2), \dots$.
  This sequence is also bounded,
  so by the Bolzano--Weierstrass Theorem,
  it has a convergent subsequence $f_{i_{2,1}}(t_2), f_{i_{2,2}}(t_2), \dots$.
  By induction, we obtain a sequence of progressively nested subsequences
  $\{i_{1,1}, i_{1,2}, \dots\} \supseteq \{i_{2,1}, i_{2,2}, \dots\} \supseteq \cdots$
  such that $f_{i_{k,1}}(t_k), f_{i_{k,2}}(t_k), \dots$ converges
  for each~$k$.

  Now diagonalize to form the subsequence
  $f_{i_{1,1}}, f_{i_{2,2}}, f_{i_{3,3}}, \dots$
  of the given functions $f_1, f_2, \dots$.
  We claim that this subsequence converges pointwise.
  For any nonnegative rational $t_k$, the sequence
  $f_{i_{1,1}}(t_k), f_{i_{2,2}}(t_k), \dots$ converges
  because the suffix
  $f_{i_{k,k}}(t_k), f_{i_{k+1,k+1}}(t_k), \dots$ is a subsequence of the
  convergent sequence $f_{i_{k,1}}(t_k), f_{i_{k,2}}(t_k), \dots$.
  Thus, for any $k$, any $\epsilon > 0$,
  and any sufficiently large $p,q$, we have
  $d_D(f_{i_{p,p}}(t_k), \allowbreak f_{i_{q,q}}(t_k)) \leq \epsilon/3$.
  By uniform equicontinuity, there is a $\delta = \delta(\epsilon) > 0$
  such that, for all~$i$,
  we have $|s-t| \leq \delta$ implies $d_D(f_i(s), f_i(t)) \leq \epsilon/3$.
  For any $t$, we can find a rational $t_k$ such that $|t - t_k| \leq \delta$.
  By the triangle inequality, for any $\epsilon > 0$ and
  sufficiently large $p,q$, we have
  \begin{align*}
  d_D\big(f_{i_{p,p}}(t), f_{i_{q,q}}(t)\big)
  &\leq
  d_D\big(f_{i_{p,p}}(t), f_{i_{p,p}}(t_k)\big) +
  d_D\big(f_{i_{p,p}}(t_k), f_{i_{q,q}}(t_k)\big) +
  d_D\big(f_{i_{q,q}}(t_k), f_{i_{q,q}}(t)\big)
  \\
  &\leq \epsilon/3 + \epsilon/3 + \epsilon/3 = \epsilon.
  \end{align*}
  Therefore, for any $t$,
  the sequence $f_{i_{1,1}}(t), f_{i_{2,2}}(t), \dots$
  is a Cauchy sequence, so it converges, as desired.
\end{proof}

\begin{proof}[Proof of Lemma~\ref{lem:choice}]
  One direction is obvious: if a single pursuer strategy wins $G_\epsilon$
  for all $\epsilon > 0$, then we satisfy the definition of winning~$G$.
  To prove the other direction, assume the pursuer wins~$G$, i.e.,
  for every $\epsilon > 0$, there is a pursuer strategy $\Pursuer_\epsilon$
  that wins~$G_\epsilon$.
  To construct a single pursuer strategy $\Pursuer_0$
  that wins all $G_\epsilon$,
  we roughly follow the proof of \cite[Lemma~3]{Bollobas-Leader-Walters-2012}
  which shows how to take limits of strategies in the Lion and Man game.
  (Our proof differs in a few ways:
  we need to check a different notion of winning;
  our result works for infinite time and unbounded domains;
  as in Lemma~\ref{lem:pointwise A-A}, we use pointwise instead of uniform
  convergence; and our proof is more detailed.)
  Specifically,
  we use Zorn's Lemma (which is equivalent to the Axiom of Choice):
  for any partially ordered set, if every chain has a maximal element,
  then there is a global maximum element.

  We define a partially ordered set of ``good partial pursuer strategies''.
  A \defn{partial pursuer strategy} is a \emph{partial} function $\Pursuer$
  from escaper motion paths to pursuer motion paths
  satisfying the nonbranching-lookahead constraint where it is defined, i.e.,
  for any two escaper motion paths $\escaper_1,\escaper_2 \in \domain(\Pursuer)$
  agreeing on $[0,t]$, the pursuer motion paths $\Pursuer(\escaper_1),
  \Pursuer(\escaper_2)$ also agree on $[0,t]$.
  A partial pursuer strategy $\Pursuer$ is \defn{good} if,
  for every escaper motion path $\escaper \in \domain(\Pursuer)$,
  there is an infinite sequence $\epsilon_1, \epsilon_2, \dots$
  converging to $0$ such that
  $\Pursuer_{\epsilon_1}(\escaper), \Pursuer_{\epsilon_2}(\escaper), \dots$
  converges pointwise to $\Pursuer(\escaper)$ in the pursuer metric.
  As we show below,
  $\Pursuer$ being good implies that $\Pursuer$ wins $G_\epsilon$
  for all $\epsilon > 0$,
  if the escaper is restricted to motion paths in $\domain(\Pursuer)$.
  The partial order $\leq$ is defined as follows:
  for two good partial pursuer strategies $\Pursuer^1,\Pursuer^2$,
  \defn{$\Pursuer^1 \leq \Pursuer^2$} if
  $\domain(\Pursuer^1) \subseteq \domain(\Pursuer^2)$ and
  $\Pursuer^1$ and $\Pursuer^2$ agree on their common $\domain(\Pursuer^1)$.

  Zorn's Lemma applies to this partial order because
  any chain $\Pursuer^1, \Pursuer^2, \dots$
  of good partial strategies has a maximal element, namely,
  $\Pursuer^1 \cup \Pursuer^2 \cup \cdots$.
  Thus we obtain a maximum good partial pursuer strategy~$\Pursuer_0$.
  We will show that $\Pursuer_0$ is in fact a (full) pursuer strategy,
  and by goodness, wins $G_\epsilon$ for all $\epsilon > 0$ as desired.

  Suppose for contradiction that $\Pursuer_0$ is not defined on
  some escaper motion path~$\escaper'$.
  We will show how to extend $\Pursuer_0$ to a good partial pursuer strategy
  $\Pursuer_0'$
  where $\domain(\Pursuer_0') = \domain(\Pursuer_0) \cup \{\escaper'\}$,
  contradicting maximality of~$\Pursuer_0$.
  To ensure preservation of the nonbranching-lookahead constraint,
  we look for an escaper motion path
  $\escaper \in \domain(\Pursuer_0)$ that agrees with $\escaper'$
  for the longest interval $[0,t^*]$.
  To this end, define
  $$
  t^* = \sup \{ t \geq 0 \mid
    \text{there exists } \escaper \in \domain(\Pursuer_0)
    \text{ such that } \escaper,\escaper'
    \text{ agree on } [0,t] \}.
  $$
  Beyond time $t^*$, we can define $\Pursuer_0'(\escaper')$ arbitrarily,
  while preserving the nonbranching-lookahead property.
  There are three cases according to whether the supremum $t^*$ is realized
  or undefined.
  \begin{description}
  \item[Case 0: $t^*$ is undefined.]
    This case happens when there is no $h \in \domain(\Pursuer_0)$
    for which $h(0) = h'(0)$,
    so no matter how we define $\Pursuer_0'(\escaper')$,
    we will satisfy nonbranching lookahead.

    Define $\epsilon_i = 1/i$, and take the sequence
    $\Pursuer_{\epsilon_1}(\escaper'), \Pursuer_{\epsilon_2}(\escaper'), \dots$.
    Now we apply Lemma~\ref{lem:pointwise A-A} to this sequence of functions.
    Our functions $\Pursuer_{\epsilon_i}(\escaper')$ are
    uniformly equicontinuous because they are Lipschitz
    with uniform constant~$r$.
    Our functions $\Pursuer_{\epsilon_i}(\escaper')$ are
    uniformly locally bounded because they are uniformly Lipschitz and
    start at points
    $\Pursuer_{\epsilon_1}(\escaper')(0), \Pursuer_{\epsilon_2}(\escaper')(0),
    \allowbreak \dots$
    which we know converge to a point,
    and thus are all within a bounded distance from that point.
    Thus
    $\Pursuer_{\epsilon_1}(\escaper'), \Pursuer_{\epsilon_2}(\escaper'), \dots$
    has an infinite subsequence
    $\Pursuer_{\epsilon_{i_1}}(\escaper'),
    \Pursuer_{\epsilon_{i_2}}(\escaper'), \dots$ 
    that converges pointwise to some function,
    which we define to be $\Pursuer_0'(\escaper')$.

    It remains to check that $\Pursuer_0'$ is a (larger)
    good partial pursuer strategy.
    By construction, $\Pursuer_0'$ is good and
    satisfies the nonbranching-lookahead constraint.
    $\Pursuer_0'(\escaper')$ satisfies the speed constraint because the
    pointwise limit of $r$-Lipschitz functions is $r$-Lipschitz.

  \item[Case 1: $t^*$ is realized.]
    Then we have an escaper path $\escaper \in \domain(\Pursuer_0)$
    such that $\escaper,\escaper'$ agree on $[0,t^*]$.
    Because $\Pursuer_0$ is good, we have a sequence
    $\Pursuer_{\epsilon_1}(\escaper), \Pursuer_{\epsilon_2}(\escaper), \dots$
    that converges pointwise to $\Pursuer_0(\escaper)$.
    The given strategies $\Pursuer_\epsilon$ are defined on all escaper paths,
    so we can form the corresponding sequence
    $\Pursuer_{\epsilon_1}(\escaper'), \Pursuer_{\epsilon_2}(\escaper'), \dots$.

    As in Case~0,
    we can apply Lemma~\ref{lem:pointwise A-A} to this sequence of functions
    to get an infinite subsequence
    $\Pursuer_{\epsilon_{i_1}}(\escaper'),
    \Pursuer_{\epsilon_{i_2}}(\escaper'), \dots$ 
    that converges pointwise to some function,
    which we define to be $\Pursuer_0'(\escaper')$.
    As in Case~0,
    $\Pursuer_0'$ is good and
    $\Pursuer_0'(\escaper')$ is a pursuer motion path.

    To prove that $\Pursuer_0'$ satisfies the nonbranching-lookahead constraint,
    it suffices to check that
    $\Pursuer_0'(\escaper'),\Pursuer_0'(\escaper)=\Pursuer_0(\escaper)$
    agree on $[0,t^*]$ (because $t^*$ is maximum).
    The subsequence $\Pursuer_{\epsilon_{i_1}}(\escaper), \allowbreak
    \Pursuer_{\epsilon_{i_1}}(\escaper), \allowbreak \dots$ converges pointwise
    to $\Pursuer_0(\escaper)$ because it is a subsequence of the sequence
    $\Pursuer_{\epsilon_1}(\escaper), \allowbreak \Pursuer_{\epsilon_2}(\escaper), \allowbreak \dots$
    which we assumed converged to $\Pursuer_0(\escaper)$,
    and the corresponding subsequence $\Pursuer_{\epsilon_{i_1}}(\escaper'),
    \allowbreak \Pursuer_{\epsilon_{i_1}}(\escaper'), \allowbreak \dots$
    converges pointwise to $\Pursuer_0'(\escaper')$ by definition.
    For each $j$,
    the given strategy $\Pursuer_{\epsilon_{i_j}}$
    satisfies the nonbranching-lookahead constraint,
    so $\Pursuer_{\epsilon_{i_j}}(\escaper),
    \Pursuer_{\epsilon_{i_j}}(\escaper')$ agree on $[0,t^*]$.
    Taking the two limits over the identical sequence
    $\epsilon_{i_1}, \epsilon_{i_2}, \dots$,
    we obtain that $\Pursuer_0(\escaper), \Pursuer_0'(\escaper')$
    also agree on $[0,t^*]$.

  \item[Case 2: $t^*$ is not realized.]
    By definition of $\sup$, we have an infinite sequence of escaper paths
    $\escaper_1, \escaper_2, \ldots \in \domain(\Pursuer_0)$
    such that $\escaper',\escaper_i$ agree on $[0,t^*_i]$
    where $t^*_i \to t^*$ and $t^* > 0$.
    We can apply Lemma~\ref{lem:pointwise A-A} to this sequence:
    uniform equicontinuity follows from escaper paths being
    Lipschitz with constant~$1$, and uniform local boundedness follows
    because all escaper motion paths $\escaper_i$ agree at time $0$,
    so at time $t$ they remain within distance $t$ of that starting point.
    By Lemma~\ref{lem:pointwise A-A},
    we obtain a subsequence $\escaper_{i_1}, \escaper_{i_2}, \ldots$
    that converges pointwise to some~$\escaper^*$.
    This~$\escaper^*$ is an escaper motion path because
    the pointwise limit of $1$-Lipschitz functions is $1$-Lipschitz.
    We claim that $\escaper^*, \escaper'$ agree on $[0,t^*]$:
    for any $t < t^*$, for sufficiently large~$i$, $h_i(t)$ agrees with $h'(t)$,
    and thus so does $h^*(t)$;
    and for $t^*$, for any $\epsilon > 0$,
    $h'(t^*)$ is within $\epsilon$ of $h'(t^*-\epsilon)$
    (by 1-Lipschitz of~$h'$),
    which is $h^*(t^* - \epsilon)$ for sufficiently large $i$,
    which is within $\epsilon$ of $h^*(t^*)$
    (by 1-Lipschitz of~$h^*$),
    so $h'(t^*)$ is within $2 \epsilon$ of $h^*(t^*)$.

    Because we are in Case~2, $\escaper^* \notin \domain(\Pursuer_0)$.
    Because each $h_i \in \domain(\Pursuer_0)$, we can construct the sequence
    $\Pursuer_0(\escaper_{i_1}), \Pursuer_0(\escaper_{i_2}), \dots$.
    We can apply Lemma~\ref{lem:pointwise A-A} to this sequence:
    uniform equicontinuity follows from pursuer paths being
    Lipschitz with constant~$r$, and uniform local boundedness follows
    because all escaper motion paths $\escaper_i$ agree at time~$0$,
    and $\Pursuer_0$ satisfies the nonbranching-lookahead constraint,
    so all pursuer motion paths $\Pursuer_0(\escaper_i)$ agree at time~$0$,
    so at time $t$ they remain within distance $t$ of that starting point.
    By Lemma~\ref{lem:pointwise A-A},
    we obtain a subsequence
    $\Pursuer_0(\escaper_{i'_1}), \Pursuer_0(\escaper_{i'_2}), \dots$
    that converges pointwise to some function,
    which we define to be $\Pursuer_0'(\escaper^*)$.
    As in Cases~0 and~1,
    $\Pursuer_0'$ is good and
    $\Pursuer_0'(\escaper^*)$ is a pursuer motion path.

    To prove that $Z_0'$ satisfies the nonbranching-lookahead constraint,
    consider an escaper motion path $\escaper \in \domain(\Pursuer_0)$,
    and suppose that $\escaper,\escaper^*$ agree on $[0,t]$,
    where $t$ is necessarily less than the supremum $t^*$
    (because we are in Case~2 and $\escaper^*, \escaper'$ agree on $[0,t^*]$).
    Take the infinite subsequence $i''_1, i''_2, \dots$ of
    $i'_1, i'_2, \dots$ where $t^*_{i''_j} \geq t$.
    Thus $\escaper,\escaper^*,\escaper_{i''_1},\escaper_{i''_2},\dots$
    agree on $[0,t]$.
    Because $\Pursuer_0$ satisfies the nonbranching-lookahead constraint,
    $\Pursuer_0(\escaper),
    \Pursuer_0(\escaper_{i''_1}), \Pursuer_0(\escaper_{i''_2}), \dots$
    agree on $[0,t]$.
    Because $\Pursuer_0(\escaper_{i''_1}),\Pursuer_0(\escaper_{i''_2}),\dots$
    converges pointwise to $Z_0'(\escaper^*)$,
    we obtain that $\Pursuer_0(\escaper), \Pursuer_0'(\escaper^*)$
    agree on $[0,t]$.

    If $h' = h^*$, we have achieved our goal.
    Otherwise, we are now in Case 1:
    the supremum $t^*$ is realized by $h^*$.
    By Case 1, we can add $h'$ to $\domain(Z_0')$ as well.

  \end{description}

  Finally, we show that $\Pursuer_0$ wins $G_\epsilon$ for all $\epsilon > 0$,
  or more generally, any good partial strategy~$\Pursuer$ wins all $G_\epsilon$
  if the escaper is restricted to motion paths in $\domain(\Pursuer)$.
  Take any $\epsilon > 0$ and any escaper motion path
  $\escaper \in \domain(\Pursuer)$.
  Because $\Pursuer$ is good, $\Pursuer(\escaper)$ is the limit of
  $\Pursuer_{\epsilon_1}(\escaper), \Pursuer_{\epsilon_2}(\escaper), \dots$
  for some sequence $\epsilon_1, \epsilon_2, \dots$ converging to~$0$.
  For all $\epsilon_i < \epsilon/2$,
  $\Pursuer_{\epsilon_i}(\escaper)$ prevents the escaper
  (following path~$\escaper$)
  from exiting $\epsilon_i < \epsilon/2$ away from the pursuer
  (in the pursuer metric), i.e., for any time $t \geq 0$,
  $\escaper(t) \in \Exit$ implies
  $d_\pursuer(\escaper(t), \Pursuer_{\epsilon_i}(\escaper)(t)) < \epsilon_i < \epsilon/2$.
  By (pointwise) convergence, for any time $t \geq 0$,
  for sufficiently large~$i$,
  $\Pursuer_{\epsilon_i}(\escaper)(t)$ is within $\epsilon/2$ of
  $\Pursuer(\escaper)(t)$ (in the pursuer metric).
  By triangle inequality, for any time $t \geq 0$,
  $\escaper(t) \in \Exit$ implies
  $d_\pursuer(\escaper(t), \Pursuer(\escaper)(t))
  < \epsilon/2 + \epsilon/2 = \epsilon$,
  i.e.,
  $\Pursuer(\escaper)(t)$ prevents the escaper from exiting
  $\epsilon/2 + \epsilon/2 = \epsilon$ away from the pursuer
  (in the pursuer metric).
  Therefore, $\Pursuer$ wins $G_\epsilon$
  when restricted to motion paths $\escaper \in \domain(\Pursuer)$,
  for all $\epsilon > 0$.
  In particular, $\Pursuer_0$ wins $G_\epsilon$ for all $\epsilon > 0$.
\end{proof}

Thus, under the Axiom of Choice, our definition of winning $G$
is equivalent to the existence of a single winning strategy for that player.
An escaper strategy \defn{wins $G$} if it wins by $\epsilon$
for some $\epsilon > 0$.
A pursuer strategy \defn{wins $G$} if it prevents the escaper from winning
by $\epsilon$ for all $\epsilon > 0$.
Henceforth, we will use the notions of winning $G_\epsilon$ instead of $G$,
so as to not rely on the Axiom of Choice.

\subsection{Both Players Cannot Win: Oblivious Strategies and Unique Playthroughs}
\label{sec:BothPlayersCannotWin}
\label{sec:oblivious strategies}

In this section, we prove that our definitions prevent
both players from having ``winning strategies'',
similar to stronger result about locally finite strategies for Lion and Man
\cite{Bollobas-Leader-Walters-2009-arXiv}.
Our main approach is to construct a valid \defn{playthrough} that can result
from a given pursuer strategy $\Pursuer$ and escaper strategy $\Escaper$, that is, an actual
pursuer path $\pursuer$ and escaper path $\escaper$ consistent with the strategies:
$\pursuer = \Pursuer(\escaper)$ and $\escaper = \Escaper(\pursuer)$.  Any playthrough has a clear winner.
We show that any winning player strategy can be modified to induce
unique playthroughs, no matter what path/strategy the opponent chooses,
while preserving the winning property.

\paragraph{Oblivious strategies.}
Our main tool is the idea of ``$\delta$-oblivious'' player strategies,
where the player can only see and react to where the opponent was at times
at least $\delta$ ago.
Formally, a player strategy $\Player$ is \defn{$\delta$-oblivious} if it satisfies
the following strengthening of the nonbranching-lookahead constraint:
\begin{quote}
for any two opponent motion paths $\opponent_1,\opponent_2$ agreeing on $[0,t]$,
the strategy's player motion paths $\Player(\opponent_1),\Player(\opponent_2)$ agree on $[0,t+\delta]$.
\end{quote}
This definition is a stronger form of the nonbranching-lookahead constraint
that guarantees a positive ($\delta$) amount of no lookahead.

Oblivious strategies are a stronger notion than ``locally finite
strategies'' introduced in \cite[Section~6]{Bollobas-Leader-Walters-2009-arXiv},
which effectively allow $\delta$ to adapt (in particular, get smaller) as time
advances.  (For example, the classic Lion and Man solution is locally finite
but not $\delta$-oblivious for any $\delta > 0$, because the lion gets
arbitrarily close to the man, so the man must react faster and faster.)
If either player uses a locally finite strategy, then the game has a unique
playthrough \cite[Proposition~14]{Bollobas-Leader-Walters-2012}.
For completeness, we prove
\ifabstract in Appendix~\ref{appendix:model} \fi
the weaker (and simpler) version we need:
one oblivious strategy implies unique playthrough.

\begin{lemma} \label{UniquePlaythroughLemma}
  If one player uses a $\delta$-oblivious strategy $\Player$ for any $\delta > 0$,
  then for any opponent strategy $\Opponent$, the game has a unique playthrough.
\end{lemma}

\begin{proof}
We will prove that strategies $(\Player,\Opponent)$ have a unique playthrough $(\player,\opponent)$
defined up until time $k\delta$, by induction on~$k$.

In the base case $k = 0$, the unique playthrough consists of trivial paths
where neither player moves, but we need to define the starting point for
both players.  The escaper strategy defines a unique starting point
for the escaper path (by the escaper start constraint), and thus the
pursuer strategy defines a unique starting point for the pursuer path
(by the nonbranching-lookahead constraint).

Now suppose we have determined a unique playthrough $(\player,\opponent)$ up until time
$k \delta$, i.e., we have determined $\player([0, k \delta])$ and $\opponent([0, k \delta])$.
By the $\delta$-obliviousness of~$\Player$, $\Player(\opponent)([0, (k+1)\delta])$ is a function
just of the opponent path $\opponent([0, k \delta])$, and is therefore uniquely
determined by the partial playthrough determined so far.
Thus we can set $\player([0, (k+1)\delta])$ accordingly.
Then the opponent's strategy $\Opponent(\player)([0, (k+1)\delta])$ is determined,
being a function of $\player([0, (k+1)\delta])$ (by the nonbranching-lookahead constraint).
Therefore we determine $\player$ and $\opponent$ uniquely and consistently by induction.
\end{proof}

Crucially, we do not require that strategies be $\delta$-oblivious.
(Such a restriction is rightly rejected in
\cite{Bollobas-Leader-Walters-2009-arXiv}
because it forbids natural strategies such as ``run in the direction of the
escaper''.)
But we can exploit the $\epsilon$ distance tolerance that we incorporated into
the definition of the pursuer winning to show that any winning player strategy
can be made oblivious, with some tweaking of the parameters:

\begin{lemma}[Obliviate Lemma] \label{InformationDelayLemma}
If a player has a winning strategy in $G_\epsilon$ with speed ratio $r$,
then that player has a $\delta$-oblivious winning strategy in $G_{\epsilon'}$,
where $\delta = {\epsilon \over 2r}$, and
where $\epsilon' = {1 \over 2} \epsilon$ if the player is the escaper
and $\epsilon' = {3 \over 2} \epsilon$ if the player is the pursuer.
\end{lemma}

\begin{proof}
Given a player winning strategy $\Player$ for~$G_\epsilon$,
we construct a $\delta$-oblivious player winning strategy~$\Player_\delta$.
Given an opponent motion path~$\opponent$, we construct a player motion path
$\Player_\delta(\opponent)$ that stands still for $\delta$ time, then mimics strategy $\Player$
but with a shifted version of~$\opponent$:
\begin{align*}
  \Player_\delta(\opponent)([0,\delta]) &= \Player(\opponent)(0), \\
  \Player_\delta(\opponent)(t+\delta) &= \Player(\opponent([0,t]))(t).
\end{align*}
This player strategy $\Player_\delta$ is clearly $\delta$-oblivious.
We show that it wins $G_{\epsilon'}$ in two cases.

First, if the player is the escaper, then for any pursuer motion path~$\opponent$,
the given winning strategy $\Player$ for $G_\epsilon$
has a time $t$ such that $\Player(\opponent)(t)$ is at an
exit while $\opponent(t)$ is at least $\epsilon$ away in the pursuer metric.
We obtain a similar time $t+\delta$ for the constructed $\delta$-oblivious
strategy $\Player_\delta$: $\Player_\delta(\opponent)(t+\delta) = \Player(\opponent)(t)$ is at an exit,
and by the speed-limit constraint, $\opponent(t+\delta)$ is at most
$\delta r = \epsilon/2$ closer than $\opponent(t)$ was.

Second, if the player is the pursuer, then for any escaper motion path $\opponent$,
and for any time $t$ where $\opponent(t)$ is on an exit,
the given winning strategy $\Player$ for $G_\epsilon$
guarantees that $\Player(\opponent)(t)$ is $< \epsilon$ distance from $\opponent(t)$.
We prove the analogous result for~$\Player_\delta$:
if $\opponent(t)$ is at an exit, then $\Player_\delta(\opponent)(t+\delta) = \Player(\opponent)(t)$
is $< \epsilon$ distance from $\opponent(t)$, and by the speed-limit constraint,
$\Player_\delta(\opponent)(t)$ is at most $\delta r = \epsilon/2$ away from $\Player(\opponent)(t)$.
The farthest it can be from $\opponent(t)$ is then ${3 \over 2} \epsilon$.
\end{proof}

\begin{corollary} \label{unique playthrough epsilon'}
  If a player has a winning strategy $\Player$ for $G_\epsilon$ with speed ratio~$r$,
  then that player has a winning strategy $\hat \Player$ for $G_{\epsilon'}$
  (where $\epsilon' = {1 \over 2} \epsilon$ if the player is the escaper
  and $\epsilon' = {3 \over 2} \epsilon$ if the player is the pursuer)
  such that, for every opponent strategy~$\Opponent$, the game of $\hat \Player$
  against $\Opponent$ has a unique playthrough (where the player wins).
\end{corollary}

\begin{proof}
By Lemma~\ref{InformationDelayLemma}, the player has a $\delta$-oblivious
winning strategy $\hat \Player$ for $G_{\epsilon'}$.
By Lemma~\ref{UniquePlaythroughLemma}, there is a unique playthrough
$(\hat \player,\opponent)$ such that $\hat \player = \hat \Player(\opponent)$ and $\Opponent(\hat \player) = \opponent$.
Because $\hat \Player$ wins against all opponent paths, it wins against~$\opponent$.
\end{proof}

Now it follows that both players cannot win in the pursuit--escape game $G$,
given that our definition of the escaper winning by a uniform $\epsilon > 0$.
(Again, a stronger result for locally finite strategies in the Lion-and-Man
game is mentioned in
\cite[after Proposition~14]{Bollobas-Leader-Walters-2009-arXiv}.)

\begin{corollary} \label{can't both win}
  For no pursuit--escape game $G$ can both the escaper and pursuer win.
\end{corollary}

\begin{proof}
  Suppose the escaper wins $G$.
  By definition, there is an escaper winning strategy $\Escaper$ for $G_\epsilon$
  for some $\epsilon > 0$.
  By Corollary~\ref{unique playthrough epsilon'}, 
  there is an escaper winning strategy $\hat \Escaper$ for $G_{\epsilon'}$,
  for some $\epsilon' > 0$, that has unique playthroughs against all
  pursuer strategies where the escaper wins.

  If the pursuer also wins $G$, then for all $\epsilon > 0$,
  there is a pursuer winning strategy $Z_\epsilon$ for~$G_\epsilon$;
  in particular, we obtain $Z_{\epsilon'}$ for $G_{\epsilon'}$.
  But $\hat \Escaper$ and $Z_{\epsilon'}$ have a unique playthrough where the escaper
  wins, contradicting that $Z_{\epsilon'}$ is a pursuer winning strategy.
\end{proof}

\paragraph{Specified starting points.}
Next we consider a variant $G(s_\escaper, s_\pursuer)$ of the game $G$ where we are
given the starting points $s_\escaper$ and $s_\pursuer$ for the escaper and pursuer,
respectively (like the Lion and Man problem).  This game naturally arises when
analyzing strategies in the middle of a game~$G$; in particular, we did so in
Section~\ref{sec:halfplane}.
A similar proof technique to the Obliviate Lemma gives us another interesting result
about robustness over starting points:

\begin{lemma}\label{lem:neighborhood}
  Suppose the escaper has a winning strategy for $G_\eps(s_\escaper,s_\pursuer)$, and
  that $s'_\pursuer$ is another point in the pursuer domain with
  $d_\pursuer(s_\pursuer,s'_\pursuer) = \delta < \eps$. Then the escaper has a winning
  strategy for $G_{\eps-\delta}(s_\escaper,s'_\pursuer)$.
\end{lemma}

\begin{proof}

Let $\Escaper$ be the assumed escaper strategy that wins $G_\eps(s_\escaper, s_\pursuer)$. We define a
new escaper strategy $\Escaper'$ that wins $G_{\eps-\delta}(s_\escaper, s'_\pursuer)$: for any pursuer
path $\pursuer'(t)$ starting at $s'_\pursuer$, the escaper strategy will return an escaper path
$\Escaper'(s'_\pursuer, t)$ defined as follows. Let $\pursuer_s(t)$ be the pursuer path starting at
$s_\pursuer$ running at full speed along a shortest path in the pursuer metric to $s'_\pursuer$
(in exactly $\delta/r$ seconds), and then for $t > \delta/r$ let $\pursuer_s(t) =
\pursuer'(t-\delta/r)$. Define $\Escaper'(\pursuer'_s, t) = \Escaper(\pursuer_s, t)$. Observe that strategy $\Pursuer'$
satisfies:
\begin{itemize}
\item the nonbranching-lookahead constraint because $\Escaper'(\pursuer_b, t)$
depends only on $\pursuer'_s$ restricted to the closed interval $[0, t - \delta/r]$
(unless $t < \delta/r$, in which case $\Escaper'(\pursuer_b, t)$ is independent of $\pursuer_b$), and
\item the speed-limit constraint because $\Escaper$ does
and $\pursuer_s$ obeys speed limit $r$.
\end{itemize}

To see that strategy $\Pursuer'$ wins $G_{\eps-\delta}(s_\escaper, s'_\pursuer)$, consider a particular
pursuer path $\pursuer'(t)$, and define $\pursuer_s(t)$ as above. Because $\Pursuer$ is a winning
strategy for $G_\eps(s_\escaper, s_\pursuer)$, there exists some time $u$ at which the escaper
wins at boundary point $\escaper_s = \Pursuer(\pursuer_s, u)$ where $d_\pursuer(\escaper_s, \pursuer_s(u)) \geq \eps$.
According to strategy $\Pursuer'$, the escaper at time $u$ reaches the same boundary point
$\Pursuer'(\pursuer', u) = \Pursuer(\pursuer_s, u) = \escaper_s$, and the pursuer is at point $\pursuer'(u)$. We claim
that $\pursuer'(u)$ has distance at least $\eps - \delta$ from $\escaper_s$ in the pursuer
metric, so the escaper wins at time $u$. 

Because pursuer has speed at most $r$, $d_\pursuer(\pursuer_s(u), \pursuer'(u)) =
d_\pursuer(\pursuer'(u-\delta/r), \pursuer'(u)) \leq \delta$. And because $d_\pursuer(\escaper_s, \pursuer_s(u)) \geq
\eps$, by the triangle inequality, $d_\pursuer(\escaper_s, \pursuer'(u))\geq \eps - \delta$ as
desired.
\end{proof}

\begin{corollary}
If the escaper can win $G(s_\escaper,s_\pursuer)$, then the escaper can win $G(s_\escaper, s'_\pursuer)$
for all $s'_\pursuer$ in some open $d_\pursuer$-neighborhood of~$s_\pursuer$.
\end{corollary}

\subsection{Discrete Game}
\label{Discrete Game}

In this section, we show how to discretize the (continuous) pursuit--escape game
while closely approximating winning strategies.  This tool will enable us
to prove that some player always wins (in Section~\ref{Some Player Wins})
and to obtain a pseudopolynomial-time approximation scheme
(in \Section~\ref{sec:pseudoPTAS}).
Bollobas et al.~\cite{Bollobas-Leader-Walters-2009-arXiv}
define a discrete pursuit--evasion game,
which discretizes time into steps,
but players still move in the original continuous domains.
By contrast, we discretize both time and space.
Combining this discretization with the stronger oblivious property that
we obtained in Section~\ref{sec:oblivious strategies}
enables us to obtain finite approximation algorithms
in Section~\ref{sec:pseudoPTAS}.
Our discrete game is similar in spirit to a discretization of pursuit--evasion
games given by Reif and Tate \cite[Section~4]{Reif-Tate-1993},
but the difference in models means that
we need to prove our own results about approximating the continuous game.

\paragraph{Discretization.}
Given a pursuit--escape game consisting of an escaper domain $D_\escaper$,
pursuer domain $D_\pursuer$, exit set $\Exit$, and speed ratio~$r$, we define the
\defn{$(\delta,\gamma)$-discretized game} $\hat G_{\delta,\gamma}(r)$
as follows.  We write an explicit ``$(r)$'' for the intended speed ratio,
as we will need to adjust it when relating to the continuous game $G = G(r)$.

First we define a \defn{$\gamma$-sampling algorithm} which,
given a locally finitely rectifiable set~$Q$
(such as $D_\escaper$, $D_\pursuer$, or $\Exit$),
produces a countable set $S_{Q,\gamma}$ of sample points
such that every point $q \in Q$ has a $\gamma$-nearby sample point.
In the special case that $Q$ is finitely rectifiable,
the sample set $S_{Q,\gamma}$ is in fact finite.
We define the \defn{$\gamma$-sample} $S_{Q,\gamma}$ of $Q$ in two cases:

\begin{itemize}
\item For a finitely rectifiable set $R$, the $\gamma$-sample of $R$ is
  the union, over every Lipschitz patch $S : [0,1]^k \to R$ constituting~$R$,
  of the finite point set
  $\big\{S(i_1/m,i_2/m,\dots,i_k/m) \mid i_1,i_2,\dots,i_k \in \{0,1,\dots,m\}\big\}$ 
  where $m = \left\lceil 1 / \left( {\sqrt k \over 2} {\gamma} \right) \right\rceil$.
  Because $R$ is bounded, this sample set is finite.
\item For a locally finitely rectifiable set $Q$,
  the $\gamma$-sample of $Q$ is the union,
  over every positive integer~$\rho$, of the $\gamma$-sample of $Q$
  intersected with the radius-$\rho$ Euclidean ball centered at the origin.
  (Each such intersection is finitely rectifiable, so its
   $\gamma$-sample is defined above.)
  This $\gamma$-sample consists of countably many points.
\end{itemize}

\begin{lemma} \label{lem:gamma sampling}
  Every point $q \in Q$ is within distance $\gamma$ of a sample point
  in $S_{Q,\gamma}$,
  where distance is measured via the
  Euclidean shortest-path metric $d_Q$ in~$Q$.
\end{lemma}

\begin{proof}
  First restrict to the integer-radius-$\lceil \|q\| + \gamma \rceil$ ball $A$
  centered at the origin,
  so that $Q \cap A$ is finitely rectifiable and has an associated sample set
  $S_{Q \cap A,\gamma} \subseteq S_{Q,\gamma}$.
  Let $S$ be a Lipschitz patch of $Q \cap A$ containing $q \in Q$.
  Consider the closed radius-$\gamma$ ball $B$
  centered at $q$ which is intrinsic to surface~$S$
  (the ball's distance is measured with respect to the metric on~$S$),
  which is contained in~$A$ (by the construction of~$A$).
  By construction of the $\gamma$-sample $S_{Q \cap A,\gamma}$,
  and by the Lipschitz property of~$S$,
  $B$ contains a point $b$ of $S_{Q \cap A,\gamma} \subseteq S_{Q,\gamma}$.
  By definition of the ball $B$, $d_Q(q,b) \leq \gamma$ as desired.
\end{proof}

Now we define a graph for the $(\delta,\gamma)$-discretized game
$\hat G_{\delta,\gamma}(r)$:
\begin{itemize}
\item %
  Define escaper vertex set
  $V_\escaper = S_{D_\escaper,\gamma} \cup S_{\Exit,\gamma}$
  and pursuer vertex set
  $V_\pursuer = S_{D_\pursuer,\gamma} \cup S_{\Exit,\gamma}$.
  Notably, both players share the exit sample~$S_{\Exit,\gamma}$.
\item The escaper edge set $E_\escaper$
  contains edges between all pairs $p,q \in V_\escaper$
  such that $d_\escaper(p,q) \leq \delta$.
\item The pursuer edge set $E_\pursuer$
  contains edges between all pairs $p,q \in V_\pursuer$
  such that $d_\pursuer(p,q) \leq r \delta$.
\end{itemize}

Finally we can define the game $\hat G_{\delta,\gamma}(r)$ which has discrete
alternation between the players.
To start, the escaper chooses a point $\escaper_0$ from $V_\escaper$;
and then the pursuer chooses a point $\pursuer_0$ from $V_\pursuer$.
In turn $i \in \{1,2,\dots\}$, the escaper chooses a point $\escaper_i$ from
$V_\escaper$ such that $(\escaper_{i-1}, \escaper_i) \in E_\escaper$;
and then the pursuer chooses a point $\pursuer_i$ from
$V_\pursuer$ such that $(\pursuer_{i-1}, \pursuer_i) \in E_\pursuer$.
The escaper wins if, in some turn~$j$,
there is a discrete exit point $\exit \in B_\exit$
such that $(\escaper_j, \exit) \in E_\escaper$ yet $(\pursuer_{j+1}, \exit) \notin E_\pursuer$;
and the pursuer wins if there is no such turn.
In other words, in the discrete game, the pursuer gets two turns
($\pursuer_j$ and $\pursuer_{j+1}$) to respond to an escaper threat $\escaper_j$ to exit
(analogous to the pursuer getting an extra reach of $\epsilon$ in the
continuous game). 
It may seem strange that the escaper wins without ever actually reaching the boundary. This captures a moment when it is clear the escaper has a forced win. Using this definition, rather than when the escaper actually reaches a boundary vertex, will be useful in future proofs when we want to consider a moment when the escaper is 'close enough' to just run to the boundary and win, or the pursuer always stays close enough to the escaper to prevent this. 

\paragraph{Approximation.}
Now we argue that winning strategies for the discrete game $\hat G$ can be
adapted to winning strategies for the continuous game $G$ with slightly
different parameters.

\begin{theorem} \label{discrete -> continuous}
  If the discrete game $\hat G_{\delta,\gamma}(r)$ has a player winning strategy
  where $\gamma \leq \min\{{1 \over 4}, {r \over 2}\} \delta$,
  then the continuous game $G_\epsilon(r')$ has a player winning strategy,
  where $\epsilon = {1 \over 2} r \delta$
  and $r' = r - 2 {\gamma \over \delta}$ if the player is the escaper, and
  $\epsilon = {5 \over 2} r \delta$
  and $r' = r / (1 - 2 {\gamma \over \delta})$ if the player is the pursuer.
\end{theorem}

\begin{proof}
  For $\player \in \{\escaper,\pursuer\}$ and for any point $p \in D_\player$,
  define
  $$
    \round_\player p = \begin{cases}
      \text{point} \in S_{\Exit,\gamma} \text{ nearest to } p
      \text{ in } d_\Exit
      & \text{if } p \in \Exit, \\
      \text{point} \in V_\player = S_{D_\player,\gamma} \cup S_{\Exit,\gamma} \text{ nearest to } p
      \text{ in } d_\player
      & \text{if } p \notin \Exit. \\
    \end{cases}
  $$
  By Lemma~\ref{lem:gamma sampling} and because $\Exit \subseteq D_\player$,
  $d_\player(p,\round_\player p) \leq \gamma$ for any point $p \in D_\player$.

  \textbf{Case 1:} The player is the escaper.
  We construct a continuous escaper winning strategy $\Escaper(\pursuer)$
  for $G_{r \delta - \gamma}(r - 2 {\gamma \over \delta})$,
  given a pursuer motion path~$\pursuer$.
  The continuous escaper starts at $\Escaper(\pursuer)(0) = \escaper_0$,
  the discrete point where the discrete escaper strategy starts.
  We give the discrete escaper strategy as input the pursuer move sequence
  $\pursuer_i = \round_\pursuer \pursuer(i \delta)$ for $i \in \{0, 1, \dots\}$.
  To confirm that this sequence satisfies $(\pursuer_i, \pursuer_{i+1}) \in E_Z$
  for all~$i$, we can use the triangle inequality, the claim above, and that
  $\pursuer$ satisfies the $r - 2 {\gamma \over \delta}$ speed-limit constraint:
  \begin{align*}
  &      d_\pursuer( \pursuer_i, \pursuer_{i+1}) \\
  &    = d_\pursuer\big(\round_\pursuer \pursuer(i \delta), \round_\pursuer \pursuer((i+1) \delta)\big) \\
  & \leq d_\pursuer\big(\round_\pursuer \pursuer(i \delta), \pursuer(i \delta)\big) +
         d_\pursuer\big(\pursuer(i \delta), \pursuer((i+1) \delta)\big) +
         d_\pursuer\big(\pursuer((i+1) \delta), \round_\pursuer \pursuer((i+1) \delta)\big) \\
  & \leq 2 \gamma + \big(r - 2 \textstyle {\gamma \over \delta}\big) \delta \\
  &    = r \delta.
  \end{align*}

  Suppose turn $i$ of the discrete escaper strategy tells us to move the
  escaper to $\escaper_i$ (dependent on only $\pursuer_0, \pursuer_1, \dots, \pursuer_{i-1}$).
  Then we extend the continuous escaper strategy by letting
  $\Escaper(\pursuer)([(i-1)\delta, i\delta])$ interpolate a shortest path in $d_\escaper$
  between $\escaper_{i-1}$ and~$\escaper_i$.
  By construction of $E_\escaper$, $d_\escaper(\escaper_{i-1},\escaper_i) \leq \delta$, so
  this interpolation satisfies the escaper speed-limit constraint.
  Because $\escaper_i$ depends on only $\pursuer_0, \pursuer_1, \dots, \pursuer_{i-1}$,
  $\Escaper(\pursuer)([(i-1)\delta, i\delta])$ depends on only $\pursuer([0,(i-1)\delta])$,
  so $\Escaper$ satisfies the nonbranching-lookahead constraint.
  (Because we are in the no-capture model,
  we do not need to worry about the escaper being captured during this motion.)

  In the final turn $j$ of the discrete game, there is an exit point
  $\exit \in \Exit$ such that
  $(\escaper_j, \exit) \in E_\escaper$ yet $(\exit, \pursuer_{j+1}) \notin E_\pursuer$.
  (Here we use that $\Exit \subseteq D_\escaper$, so that
  $d_\escaper(\escaper_j, \exit) \leq d_\Exit(\escaper_j, \exit)$.)
  Thus $d_\escaper(\escaper_j,\exit) \leq \delta$ yet $d_\pursuer(\pursuer_{j+1},\exit) > r \delta$.
  We finish the continuous escaper winning strategy by letting
  $\Escaper(\pursuer)([j \delta, (j+1)\delta])$ interpolate a shortest path in $d_\escaper$
  from $\escaper_j$ to~$\exit$.
  As above, $\Escaper$ satisfies the escaper speed-limit constraint and
  nonbranching-lookahead constraint.
  Furthermore, $\Escaper(\pursuer)$ is a continuous escaper winning strategy for
  $G_{r \delta - \gamma}$ because, at time $t = (j+1) \delta$,
  $\Escaper(\pursuer)(t)$ is at an exit $\exit$
  yet $\pursuer_{j+1} = \round_\pursuer \pursuer(t)$ is at a distance $> r \delta$ away,
  so by the claim above,
  $\pursuer(t)$ is at distance $> r \delta - \gamma$ away.
  By our assumption that $\gamma \leq {r \over 2} \delta$,
  $r \delta - \gamma \geq {r \over 2} \delta$.

  \textbf{Case 2:} The player is the pursuer.
  We construct a continuous pursuer winning strategy $\Pursuer(\escaper)$
  for $G_{2 r \delta + \gamma}(r')$, given an escaper motion path~$\escaper$.
  Let $\delta' = \delta (1 - 2 {\gamma \over \delta})$.
  We give the discrete pursuer strategy as input the escaper move sequence
  $\escaper_i = \round_\escaper \tilde \escaper(i \delta')$ for $i \in \{0, 1, \dots\}$.
  To confirm that this sequence satisfies $(\escaper_i, \escaper_{i+1}) \in E_H$
  for all~$i$, we can use the triangle inequality, the claim above, and that
  $\escaper$ satisfies the $1$ speed-limit constraint:
  \begin{align*}
  &      d_\escaper( \escaper_i, \escaper_{i+1}) \\
  &    = d_\escaper\big(\round_\escaper \escaper(i \delta'), \round_\escaper \escaper((i+1) \delta')\big) \\
  & \leq d_\escaper\big(\round_\escaper \escaper(i \delta'), \escaper(i \delta')\big) +
         d_\escaper\big(\escaper(i \delta'), \escaper((i+1) \delta')\big) +
         d_\escaper\big(\escaper((i+1) \delta'), \round_\escaper \escaper((i+1) \delta')\big) \\
  & \leq 2 \gamma + \delta' \\
  &    = 2 \gamma + \delta \big(1 - 2 \textstyle {\gamma \over \delta}\big) \\
  &    = \delta.
  \end{align*}

  The continuous pursuer starts at $\Pursuer(\escaper)(0) = \pursuer_0$, which depends on
  $\escaper_0 = \round_\escaper \escaper(0)$ (satisfying the nonbranching-lookahead constraint).
  Suppose turn $i$ of the discrete pursuer strategy tells us to move the
  pursuer to $\pursuer_i$ (dependent on only $\escaper_0, \escaper_1, \dots, \escaper_{i-1}$).
  Then we extend the continuous pursuer strategy by letting
  $\Pursuer(\escaper)([(i-1)\delta', i\delta'])$ interpolate a shortest path in $d_\pursuer$
  from $\pursuer_{i-1}$ to~$\pursuer_i$.
  By definition of $E_\pursuer$, $d_\pursuer(\pursuer_{i-1},\pursuer_i) \leq r \delta = r' \delta'$, so
  this interpolation satisfies the $r'$ pursuer speed-limit constraint.
  Because $\pursuer_i$ depends on only $\escaper_0, \escaper_1, \dots, \escaper_{i-1}$,
  $\Pursuer(\escaper)([(i-1)\delta', i\delta'])$ depends on only $\escaper([0,(i-1)\delta'])$,
  so $\Pursuer$ satisfies the nonbranching-lookahead constraint.

  To see that $\Pursuer$ is a continuous pursuer winning strategy
  for $G_{2 r \delta + \gamma}(r')$,
  consider a time $t$ where $\escaper(t) = \exit \in \Exit$.
  Let $i \delta'$ be the integer multiple of $\delta'$ nearest~$t$,
  so $|t - i \delta'| \leq {1 \over 2}$.
  By the escaper speed-limit constraint,
  $d_\escaper(\escaper(t), \escaper(i \delta')) \leq {\delta' \over 2}$.
  By the triangle inequality and the claim above,
  $d_\escaper(\escaper(t), \round_\escaper \escaper(i \delta')) \leq {\delta' \over 2} + \gamma$,
  i.e., $d_\escaper(\exit, \escaper_i) \leq {\delta' \over 2} + \gamma$.
  By the definition of $\round_\escaper$, $\hat \exit  = \round_\escaper \exit \in S_{\Exit,\gamma}$.
  By the triangle inequality and the claim above,
  $d_\escaper(\hat \exit, \escaper_i) \leq {\delta' \over 2} + 2 \gamma
                    \leq {\delta \over 2} + 2 \gamma$.
  (Here we use that $\Exit \subseteq D_\pursuer$, so that
  $d_\pursuer(\hat \exit, \escaper_i) \leq d_\Exit(\hat \exit, \escaper_i)$.)
  By our assumption that $\gamma \leq {\delta \over 4}$,
  $d_\escaper(\hat \exit, \escaper_i) \leq \delta$, so $(\hat \exit, \escaper_i) \in E_\escaper$.
  By the discrete win condition, $(\hat \exit, \pursuer_{i+1}) \in E_\pursuer$,
  so $d_\pursuer(\hat \exit, \pursuer_{i+1}) \leq r \delta$.
  Thus $d_\pursuer(\exit, \pursuer_{i+1}) \leq r \delta + \gamma$,
  i.e., $d_\pursuer(\escaper(t), \Pursuer(\escaper)((i+1) \delta)) \leq r \delta + \gamma$.
  By the pursuer speed-limit constraint,
  $d_\pursuer(\escaper(t), \Pursuer(\escaper)(t)) \leq 2 r \delta + \gamma$.
  By our assumption that $\gamma \leq {r \over 2} \delta$,
  $2 r \delta + \gamma \leq {5 \over 2} r \delta$.
\end{proof}

\subsection{Some Player Wins}
\label{Some Player Wins}

\paragraph{Discrete game.}
We start by proving that the discrete game $\hat G_{\delta,\gamma}(r)$
(defined in Section~\ref{Discrete Game}) always has a winner.
This result follows from known results, but is nontrivial because the
vertex set $V$ can have countably many vertices (when either domain is
unbounded).

\begin{lemma} \label{discrete game winner}
  The discrete game $\hat G_{\delta,\gamma}(r)$ always has a unique winner,
  i.e., either has an escaper winning strategy or a pursuer winning strategy
  but not both.
\end{lemma}

\begin{proof}
  We show that any $\hat G_{\delta,\gamma}(r)$ is an instance of an open
  Gale--Stewart game \cite{Gale-Stewart-1953},
  where two players alternate moves (with perfect
  information about past moves), a move is an element of a discrete set~$M$,
  the first player wins if the sequence of moves has a prefix
  in a known set $\Player$ of finite prefixes, and the second player wins
  if they can prevent ever having a prefix in~$\Player$.
  (The prefix notion of winning is what makes the game ``open'' in the
  product topology of $M^\omega$.)

  We can represent $\hat G$ by setting $M = V$,%
  \footnote{We could even make $M$ finite by mapping the finite number of
  choices available at any state to either player (by finite rectifiability)
  to bounded integers.}
  and letting a finite prefix $\escaper_0, \pursuer_0, \escaper_1, \pursuer_1, \dots, \escaper_k, \pursuer_k$ represent
  the state of the game if the escaper starts at $\escaper_0 \in V_\escaper$,
  the pursuer starts at $\pursuer_0 \in V_\pursuer$,
  then the escaper moves along $(\escaper_0,\escaper_1) \in E_\escaper$,
  then the pursuer moves along $(\pursuer_0,\pursuer_1) \in E_\pursuer$, etc.
  If any $\pursuer_i \notin V_\pursuer$, or any $(\pursuer_i,\pursuer_{i+1}) \notin E_\pursuer$,
  then we declare the prefix winning for the escaper.
  Conversely, if any $\escaper_i \notin V_\escaper$, or any $(\escaper_i,\escaper_{i+1}) \notin E_\escaper$,
  then we forbid the prefix from being winning for the escaper.
  Otherwise, we define the prefix as winning if and only if there is an
  $\exit \in B_\exit$ such that $(\escaper_{k-1}, \exit) \in E_\escaper$ yet $(\pursuer_k, \exit) \notin E_\pursuer$.

  Thus the discrete pursuit--escape game $\hat G$ is an open Gale--Stewart game.
  By open determinacy theorem \cite{Gale-Stewart-1953}
  this game is strictly determined, meaning that it has a unique winner.
\end{proof}

\paragraph{Continuous game.}
Now we can combine Theorem~\ref{discrete -> continuous} with
Lemma~\ref{discrete game winner} to derive results about the continuous
pursuit--escape game:

\begin{theorem} \label{nearby continuous wins}
  For any escaper domain~$D_\escaper$, pursuer domain~$D_\pursuer$, exit set~$\Exit$,
  and speed ratio~$r$,
  either the escaper wins $G(r')$ for all $r' < r$
  or the pursuer wins $G(r')$ for all $r' > r$
  (or both).
\end{theorem}

\begin{proof}
  Construct an infinite sequence by, for each $i = 1, 2, \dots$,
  taking the $(\delta_i,\gamma_i)$-discretized game
  $\hat G_{\delta_i,\gamma_i}(r)$ induced by $(D_\escaper, D_\pursuer, \Exit, r)$
  and parameters $\delta_i = 1/i$ and
  $\gamma_i = \min\{{1 \over 4}, {r \over 2}\} \delta_i / i$.
  By Lemma~\ref{discrete game winner}, every discrete game
  $\hat G_{\delta_i,\gamma_i}(r)$ has a unique winner $w_i$ (escaper or pursuer).
  We split into two cases, both of which could happen
  (and indeed will happen at the critical speed ratio):

  \textbf{Case 1:}
  If infinitely many $w_i$ are escaper,
  then by Theorem~\ref{discrete -> continuous},
  we can convert each discrete escaper winning strategy for
  $\hat G_{\delta_i,\gamma_i}(r)$ into a continuous escaper winning strategy for
  $G_{{1 \over 2} r \delta_i} (r - 2 {\gamma_i \over \delta_i}) =
  G_{{1 \over 2} r/i} (r - \min\{{1 \over 2}, {r \over 2}\} / i)$.
  By definition, the escaper wins $G(r - \min\{{1 \over 2}, {r \over 2}\} / i)$
  (as well as at all smaller speed ratios).
  Because this holds for infinitely many~$i$, and
  $\min\{{1 \over 2}, {r \over 2}\} / i \to 0$ as $i \to \infty$,
  we obtain that the escaper wins $G(r-\epsilon)$ for all $\epsilon > 0$.

  \textbf{Case 2:}
  If infinitely many $w_i$ are pursuer,
  then by Theorem~\ref{discrete -> continuous},
  we can convert the discrete pursuer winning strategy for
  $\hat G_{\delta_i,\gamma_i}(r)$ into a continuous pursuer winning strategy for
  $G_{{5 \over 2} r \delta_i}(r / (1 - 2 {\gamma_i \over \delta_i}) =
  G_{{5 \over 2} r/i}(r / (1 - \min\{{1 \over 2}, {r \over 2}\} / i))$.
  Each such strategy also wins
  $G_\epsilon(r / (1 - \min\{{1 \over 2}, {r \over 2}\} / i)$ for all
  $\epsilon \geq {5 \over 2} r/i$.
  (as well as at all larger speed ratios).
  Because this holds for infinitely many~$i$, and
  ${5 \over 2} r/i \to 0$ and
  $\min\{{1 \over 2}, {r \over 2}\} / i \to 0$ as $i \to \infty$,
  we obtain that the pursuer wins $G_\epsilon(r+\epsilon)$
  for all $\epsilon > 0$.
  Thus the pursuer wins $G(r+\epsilon)$ for all $\epsilon > 0$.
\end{proof}

\begin{corollary}
  Any (continuous) pursuit--escape instance $(D_\escaper, D_\pursuer, \Exit)$
  has a critical speed ratio $r^* \geq 0$ (possibility $\infty$)
  such that the escaper wins $G(r)$
  for all speed ratios $r < r^*$ and the pursuer wins $G(r)$
  for all speed ratios $r > r^*$.
\end{corollary}

The critical speed ratio $r^*$ can be $\infty$.  For example, consider a cusp
$\prec$ where the escaper domain is (locally) on the right and the pursuer
domain is (locally) on the left.  No matter what speed $r$ the pursuer has,
a unit-speed escaper can get sufficiently close to the cusp vertex, threaten
to leave on the top side, and then run to the bottom side and escape.
Thus the escaper always wins in such examples.

\begin{theorem}[pursuer wins at critical speed ratio] For any region $R$ and speed $r^*$, if for all $r > r^*$ the pursuer wins the game at speed ratio $r$, then the pursuer wins at $r^*$.

Equivalently, the interval of speeds for which the pursuer wins is closed.

Equivalently, if the critical speed ratio $r^*$ is finite, the pursuer wins at speed $r^*$.
\end{theorem}

\begin{proof}
For every $\epsilon > 0$,
we will construct an $\frac{\epsilon}{4(r^*+1)}$-oblivious winning strategy for the pursuer
in the game $G_\epsilon$ with speed ratio $r^*$. 
By Lemma~\ref{InformationDelayLemma},
since the pursuer has a winning strategy in $G_{\frac{\epsilon}3}$ with speed ratio $r$,
for every $r > r^*$ and every $\epsilon > 0$,
the pursuer has an $\frac{\epsilon}{6r}$-oblivious winning strategy $Z_{r,\epsilon}$
for~$G_{\frac{\epsilon}{2}}$ with speed ratio $r$. 
So for every $r \in (r^*,r^*+1]$ and every $\epsilon > 0$,
the pursuer has an $\frac{\epsilon}{6(r^*+1)}$-oblivious winning strategy $Z_{r,\epsilon}$
for~$G_{\frac{\epsilon}{2}}$ with speed ratio $r$. 
We will simulate those games for every $r$ in the sequence
$\langle r^* + {1 \over k} \mid k = 1, 2, \dots \rangle$, and in all of them we will have the escaper move as it does in the $G_{\epsilon}$ game.\footnote{Note that here we take advantage of the asymmetry between the definitions of escaper and pursuer wins: for the escaper to win, it needs only be the case that there exists one $\epsilon > 0$ for which the escaper wins, so a similar strategy of simulating infinitely many games would not be possible for the escaper.} We will define a winning strategy $Z(\escaper)$ in $G_\epsilon$ for every escaper strategy $\escaper$.

Consider the set of starting locations $\{Z_{r,\epsilon}(h)(0)\}_{r}$ chosen by pursuers in those simulated games. 
There are infinitely many of them, and they all lie within a pursuer-metric disk of radius
$\frac{\epsilon}2 + ((r^*+1))\cdot d_{\escaper}(\escaper(0), \Exit)$
(or else the escaper could win the simulated games by running directly to $\Exit$).
So, by Lemma~\ref{lem:intrinsic metric compact},
they have a limit point $p_0$ in the pursuer metric;
the pursuer chooses to start at $p_0$.
We will continue the simulations only of those simulated games for which the pursuer starts within $\frac{\epsilon}{4}$ of $p_0$, 
of which there are infinitely many since $p_0$ was a limit point.

We prove by induction on $k$ that at time $k \frac{\epsilon}{12(r^*+1)}$, we can guarantee that the pursuer is at distance at most $\frac{\epsilon}{2}(1 - 2^{-1-k})$ from the positions of the pursuers in infinitely many of the simulated games. This is true for $k=0$, as above.

At time $k \frac{\epsilon}{12(r^*+1)}$, the pursuer decides its movement for the next $\frac{\epsilon}{12(r^*+1)}$ time as follows:
simulate each game until time $(k+1)\frac{\epsilon}{12(r^*+1)}$.
The pursuers in the simulated games follow $\frac{\epsilon}{6(r^*+1)}$-oblivious strategies, so their strategies until that time depend on the position of the escaper no later than $(k+1)\frac{\epsilon}{12(r^*+1)} - \frac{\epsilon}{6(r^*+1)} = (k-1)\frac{\epsilon}{12(r^*+1)}$. At time $k \frac{\epsilon}{12(r^*+1)}$, the pursuer (for whom we are constructing an $\frac{\epsilon}{12(r^*+1)}$-oblivious strategy) knows that much of the escaper's motion, so it can in fact simulate all those games.

Consider the set of positions at which pursuers in those simulated games are at time $(k+1)\frac{\epsilon}{12(r^*+1)}$.
There are infinitely many of them, and they all lie within a disk of radius
$\frac{\epsilon}2 (1 - 2^{-1-k}) + (r^*+1) \frac{\epsilon}{12(r^*+1)}$
centered at $p_k$,
so by Lemma~\ref{lem:intrinsic metric compact},
they have a limit point $p_{k+1}$.
All the simulated pursuers are within distance $\frac{\epsilon}{2}(1 - 2^{-1-k})$ of the actual pursuer at time $k\frac{\epsilon}{12(r^*+1)}$, and for any $\delta > \frac{\epsilon}{12(r^*+1)} r^*$, only finitely many of the simulated pursuers are fast enough to travel a distance greater than $\delta$, so $p_{k+1}$ is within $\frac{\epsilon}{12(r^*+1)}r^* + \frac{\epsilon}{2}(1 - 2^{-1-k})$ of the pursuer's position at time $k\frac{\epsilon}{12(r^*+1)}$. The pursuer chooses to move toward $p_{k+1}$, so by time $(k+1)\frac{\epsilon}{12(r^*+1)}$, the pursuer is within $\frac{\epsilon}{2}(1 - 2^{-1-k})$ of $p_{k+1}$. Continue the simulations only of those games in which the simulated pursuer is within $\frac{\epsilon}{2}(2^{-2-k})$ of the limit point,
of which there are infinitely many since $p_{k+1}$ was a limit point.
By the triangle inequality, the pursuer's distance from the pursuer in each of those games at time $(k+1)\frac{\epsilon}{12(r^*+1)}$ is at most $\frac{\epsilon}{2}(1 - 2^{-1-k}) + \frac{\epsilon}{2}(2^{-2-k}) = \frac{\epsilon}{2}(1 - 2^{-2-k})$, completing the induction.

Whenever the escaper is at an exit, the pursuers in the simulated games are within distance $\frac{\epsilon}{2}$, since they are following winning strategies for~$G_{\frac{\epsilon}{2}}$. The (unsimulated) pursuer is within $\frac{\epsilon}{2}$ of those (simulated) pursuers, so by the triangle inequality it is within $\epsilon$ of the escaper, so the pursuer wins $G_{\epsilon}$, as claimed.
\end{proof}

\finalmodelresult

\section{Pseudopolynomial-Time Approximation Scheme}
\label{sec:pseudoPTAS}
\label{appendix:pseudoPTAS}

In this section, we give a \defn{pseudopolynomial-time approximation scheme}
for approximating the critical speed ratio $r^*$
when the escaper domain is the interior and boundary of a simple polygon $P$
with integer coordinates, the pursuer domain is the boundary and
optional exterior of~$P$, and the exit set $\Exit = \partial P$.
More precisely, given $D_\escaper$, $D_\pursuer$, and $\epsilon > 0$, the scheme approximates
$r^*$ to within a factor of $1 + \epsilon$
in time polynomial in $1/\epsilon$ and the polygon coordinates.
Our main tool is the $(\delta,\gamma)$-discretized game defined and analyzed
in \Section~\ref{Discrete Game}.  (In fact, we initially developed the
discretization idea in the context of this pseudopolynomial-time
approximation scheme, and later realized it could be useful to prove that the
continuous game always has a winner.)
We showed in \Section~\ref{Discrete Game} that the discrete game approximates
the continuous game in some sense, but we need substantially more effort to
turn this into an efficient approximation algorithm.

\subsection{Restricting to Convex Hull}

One challenge with applying the discretization tool is that the vertex set $V$
has infinitely many points whenever $D_\escaper$ or $D_\pursuer$ is unbounded.
Even in very natural models (e.g., the exterior model),
$D_\pursuer$ is typically unbounded.
Luckily, we can focus our attention to the convex hull of all boundaries:

\begin{lemma} \label{ConvexHullLemma}
  If a player in domain $D$ has a winning strategy %
  that leaves the convex hull of $\partial D$,
  then they have a winning strategy %
  that does not.
\end{lemma}

\begin{proof}
  Let $C$ be the convex hull of $\partial D$ (i.e., its interior and boundary),
  and let $\Player$ be a player winning strategy for~$G_\epsilon$.
  For any opponent motion path $\opponent$ and time~$t$,
  define $\hat \Player(\opponent)(t)$ to be the nearest point $\in C$ to $\Player(\opponent)(t)$.
  Because this modification is a contraction\xxx{ideally prove this},
  $\hat \Player$ will still satisfy the speed-limit constraint.
  Because the modification is independent of~$\opponent$,
  $\hat \Player$ will still satisfy the nonbranching-lookahead constraint and
  (for the escaper player) the escaper-start constraint.
  Because $\Player$ won against every opponent strategy $\opponent$, so will~$\hat \Player$.
\end{proof}

\subsection{Margin of Victory}

Another challenge with applying the discretization tool is that,
while Theorem~\ref{discrete -> continuous} relates
discrete winning strategies to continuous winning strategies, it does so only
for $G_\epsilon$ for some $\epsilon > 0$.  But we want an algorithm to compute
the critical speed ratio for~$G$, not some $G_\epsilon$.  To resolve this
discrepancy, we develop a tool for trading off the pursuer winning distance
$\epsilon$ with the speed ratio.

First we need a simpler lemma:

\begin{lemma} \label{RunStraightLemma}
If the escaper has a winning strategy for $G_\epsilon$,
then
the escaper has a winning strategy for $G_{\epsilon/(2r+3)}$
satisfying that the last $\epsilon/(2r+3)$ time of their motion
(in response to any pursuer motion path) is along a shortest path.

If the escaper domain $D_\escaper$ is a polygon (interior and boundary)
and $\Exit = \partial D_\escaper$,
then the escaper can further restrict to a straight-line motion
for the last $\epsilon/(2r+3)$ time of their motion.
\end{lemma}

\begin{proof}
Suppose the escaper has a winning strategy $\Escaper$ for $G_\epsilon$, i.e.,
for any pursuer motion path~$\pursuer$, there is a time $t_\pursuer$ such that
$\Escaper(\pursuer)(t_\pursuer)$ is an exit $x_\pursuer$ and $d_\pursuer(\Escaper(\pursuer)(t_\pursuer), \pursuer(t_\pursuer)) \geq \epsilon$.
Define an escaper--pursuer distance by
$$d_{\escaper \pursuer}(p_\escaper, p_\pursuer) = \min_{\exit \in \Exit} d_\escaper(p_\escaper, \exit) + d_\pursuer(p_\pursuer, \exit). $$
At any time $t \geq t_\pursuer - \epsilon\frac{1}{2r+3}$,
$d_\escaper(\Escaper(\pursuer)(t), x_\pursuer) \leq \epsilon\frac{1}{2r+3}$
(by the escaper speed-limit constraint)
and
$d_{\escaper \pursuer}(\Escaper(\pursuer)(t), \pursuer(t)) \geq \epsilon\frac{r+2}{2r+3}$
(because in time $\leq \epsilon\frac{1}{2r+3}$, the pursuer and escaper travel a
total distance of $\leq \epsilon\frac{r+1}{2r+3}$, so $d_{\escaper \pursuer}(\Escaper(\pursuer)(t), \pursuer(t))$
can change by at most that much, yet it reaches at least
$\epsilon\frac{2r+3}{2r+3}$ at the end).
Thus, if we replace $\Escaper(\pursuer)([t_\pursuer - \epsilon\frac{1}{2r+3}, t_\pursuer])$ with the escaper
moving along a shortest path to~$x_\pursuer$, then $d_{\escaper \pursuer}(\Escaper(\pursuer)(t), \pursuer(t))$
can decrease by $\leq \epsilon\frac{r+1}{2r+3}$ over that time
from its initial value of $\geq \epsilon\frac{r+2}{2r+3}$,
leaving a distance of $\geq \epsilon\frac{1}{2r+3}$.
Thus we obtain an escaper winning strategy for $G_{\epsilon/(2r+3)}$.

If shortest paths in $d_\escaper$ are polygonal with vertices only at points of~$\Exit$,
as when $D_\escaper$ is a polygon and $\Exit = \partial D_\escaper$, then we can stop the
shortest-path motion whenever it hits a point of $\Exit$ and thereby guarantee a
straight-line motion (which can be spread out over the final
$\epsilon\frac{1}{2r+3}$ time interval).  By the same argument, this strategy
will still win.
\end{proof}

\begin{lemma}
\label{MarginOfVictoryLemma}
Suppose $P$ is a simple polygon and $\epsilon > 0$ satisfies
\begin{enumerate}
\item \label{point to boundary}
  there is a point in $P$ at distance more than $\epsilon$
  from the nearest boundary;
\item \label{edge to edge}
  no disk of radius $2\sqrt{\epsilon}$
  intersects two edges not sharing a vertex; and
\item \label{epsilon r^*}
  $\epsilon < 1/(2 (r^*)^2)$ where $r^*$ is the critical speed ratio for the game with escaper domain~$D_{\escaper} = P$, exit set $X = \partial P$, and pursuer domain $D_{\pursuer}$ either $\partial P$ or $\overline{\mathbb{R}^2 - P}$.
\end{enumerate}
If the escaper wins the continuous game $G$ in a polygon $P$ (with $D_{\escaper}$, $X$, and $D_{\pursuer}$ as above) at a speed ratio $r$, then the escaper wins the game $G_{\epsilon^3}$ at speed ratio $r/(1+\epsilon)$.
\end{lemma}

While it is easy to prove that such an $\epsilon$ \emph{exists} via
Lemma~\ref{InformationDelayLemma},
the point is that we can efficiently compute a valid such $\epsilon$.
Specifically, we can compute an $\epsilon_0$ such that all
$\epsilon \in (0, \epsilon_0]$ satisfy the three conditions of
Lemma~\ref{MarginOfVictoryLemma} by taking the minimum of the following
three lower bounds:
\begin{enumerate}
\item We can compute a lower bound on Condition~\ref{point to boundary} by triangulating $P$, choosing any of that triangulation's triangles, and using the inradius of that triangle. The inradius is the area divided by half the perimeter, and both of those are polynomial functions of the input coordinates, so this bound on $\epsilon_0$ is polynomial in the coordinates of $P$.
\item We can compute a lower bound on Condition~\ref{edge to edge}: the minimum distance between two
edges not sharing a vertex is attained either by a pair of vertices (and we can
compute the minimum distance between pairs of vertices) or by the perpendicular from a vertex $v$ to an edge $(u,w)$. The length of that perpendicular is the area of the triangle with vertices $u$, $v$, and $w$ divided by the distance from $u$ to $w$, and those are both polynomial in $u$, $v$, and $w$, so this bound on $\epsilon_0$ has length (in bits) polynomial in the length (in bits) of $P$.
\item \label{easy upper bound for r^*}
We can compute a lower bound on Condition~\ref{epsilon r^*} via an
upper bound on $r^*$ depending only on $P$:
by Theorem~\ref{ConstantApproximationTheorem}, $r^* \leq 10.89898 \max_{p,q \in \partial P}\frac{d_\pursuer(p,q)}{d_\escaper(p,q)}$.
Instead of computing $\frac{d_\pursuer(p,q)}{d_\escaper(p,q)}$ directly, we can easily
upper-bound it by the maximum of two easy-to-compute quantities:
\begin{enumerate}
\item $F/f$ where $F$ is the perimeter of $\partial D_\escaper$ and
  $f$ is minimum distance between two nonincident edges (minimum feature size).
  This is an upper bound on $\frac{d_\pursuer(p,q)}{d_\escaper(p,q)}$ for $p,q$ on any two
  nonincident edges.
\item $\csc {\theta_{\min} \over 2}$ where $\theta_{\min}$ is the smallest
  interior angle of a vertex of~$D_\escaper$.  This is an upper bound on
  $\frac{d_\pursuer(p,q)}{d_\escaper(p,q)}$ for $p,q$ along two edges sharing an endpoint.
  (When $p$ and $q$ are on the same edge, we get a ratio of $1$, so we do not
  need to consider this case.)
\end{enumerate}
\end{enumerate}

\begin{proof}[Proof of Lemma~\ref{MarginOfVictoryLemma}]
The escaper should start at some point $H$ at distance more than $\epsilon$ from the nearest boundary (escaper-start constraint); $\epsilon_0$ was chosen small enough that such a place exists. The escaper can still win $G$: if they could win by some other starting position, the escaper can immediately run to that position; wherever the pursuer is after that run, the pursuer could have started there, so the escaper can simulate their winning strategy from that starting position to win.

If, from that starting position, there is a point $T$ on the boundary such that the escaper can win $G$ with speed ratio $r$ by committing to running straight to $T$ (that is, if there is a point $T \in \partial P$ such that $r \cdot d_\escaper(H,T) < d_\pursuer(Z,T)$, where $Z$ is the pursuer's starting position), then the escaper can win $G_{\epsilon^3}$ with speed ratio $r\frac{1}{1+\epsilon}$ by running straight to that point. The escaper's time to get there is $d_\escaper(H,T)$, in which time the pursuer moves at most $r\frac{1}{1+\epsilon} \cdot d_\escaper(H,T) < \frac{1}{1+\epsilon} \cdot d_\pursuer(Z,T)$, leaving a distance of at least $\frac{\epsilon}{1+\epsilon} \cdot d_\pursuer(Z,T) > \frac{\epsilon}{2r} \cdot d_\escaper(H,T) > \frac{\epsilon^2}{2r} > \epsilon^3$, as desired.

Otherwise, the escaper cannot immediately win $G$ with speed ratio $r$ by picking a point on $\partial P$ within $\epsilon$ of their location
and running straight to it. However, the escaper can eventually win $G$ by using the strategy in Lemma~\ref{RunStraightLemma}. Consider the escaper's position $H$ and pursuer's position $Z$ at 
a time $t$ such that for all later times, the escaper can win by picking a point on $\delta P$ and running along a shortest path to it, and for all earlier times, the escaper cannot so win.

Let $W = \{W_1, W_2, \ldots\}$ be the set of points on $\partial P$ that the escaper can reach in the same time as the pursuer if both of them run on a shortest path.

If there is any point in $W$ at distance more than $\epsilon$ from $h$, then by the same calculation as above, the escaper can win $G_{\epsilon^3}$ at speed ratio $r\frac{1}{1+\epsilon}$ by running straight to it. Otherwise, every such boundary point is within $\epsilon$ of~$h$. By the choice of $\epsilon$, there are at most two edges within $\epsilon$ of~$h$, and if there are two such edges, they share a vertex, so all points in $W$ are on one or two adjacent edges.

When the escaper is at $h$ and the pursuer at $z$, for every point $x$ on the boundary, the time it would take the pursuer to reach $x$ is at most the time it would take the escaper to reach $x$. Suppose not, and suppose that the escaper's shortest path to $x$ has length $\ell$ and the pursuer's shortest path to $x$ has length $\ell + r\epsilon_x$. Then at time $t - \frac{\epsilon_x}{2r+3}$, the length of the escaper's shortest path to $x$ is at most $\ell + \frac{\epsilon_x}{2r+3}$ and the length of the pursuer's shortest path to $x$ is at least $\ell + r\epsilon_x - \frac{r\epsilon_x}{2r+3} > \ell + \frac{\epsilon_x}{2r+3} + \frac{\epsilon_x}{2}$, so at time $t - \frac{\epsilon_x}{2r+3}$, the escaper can win $G_{\epsilon'}$ for all $\epsilon' < \frac{\epsilon_x}{2}$ by picking the point $x$ and running along a shortest path to it. Hence the escaper wins $G$ by the same strategy, which contradicts the choice of $t$.

If $W$ is empty, then consider, for each point on $\partial P$, the time it takes the escaper to reach that point minus the time it takes the pursuer to reach that point. That's function is always nonnegative, is nowhere 0 by assumption, and is a continuous function of a parameterization of the boundary, which is closed and bounded. Therefore, it attains a minimum $\epsilon_x$. At time $t + \frac{\epsilon_x}{2r+3}$, that function is still everywhere at least $\frac{\epsilon_x}{2}$, by the same calculation as above. So even at time $t + \frac{\epsilon_x}{2r+3}$, there's no point on the boundary such that the escaper can win $G_{\epsilon'}$ for any $\epsilon' < \frac{\epsilon_x}{2}$ by running on a shortest path to it. So at time $t + \frac{\epsilon_x}{2r+3}$, there's no point on the boundary such that the escaper can win $G$ by running on a shortest path to it, contradicting the choice of $t$.

If there are two points of $W$ on the same edge, let that edge be the $x$-axis. The escaper's shortest-path time to a point $(x,0)$ is a function of the form $f(x) = \sqrt{x^2 + ax + b}$ for some $a$ and $b$, and the pursuer's shortest-path time is $g(x) = \sqrt{x^2 + cx + d}/r$. Since $r>1$, if those two functions are equal at two points, then there's some point such that the escaper's shortest-path time to it is strictly less than the pursuer's, say by $\epsilon_x$. At time $t - \frac{\epsilon_x}{2r+3}$, that difference is still at least $\frac{\epsilon_x}{2}$, by the same calculation as above. So even at time $t - \frac{\epsilon_x}{2r+3}$, the escaper can win $G$, contradicting the choice of $t$.

Now, suppose that there's at most one point of $W$ per edge, and $W$ is nonempty. We claim that the escaper cannot win at all with at most $\epsilon$ more movement, much less by committing to moving straight to one of those boundary points. 

First, if there are points of $W$ on only one edge $e$, the pursuer can use an APLO strategy until there's a point of $W$ on another edge: the pursuer can match the escaper's speed perpendicular to $HZ$ and, conditioned on that, use as much of its speed as possible to move toward $e$ in the direction of $HZ$. If it's the case that, for every direction $\theta$ that the escaper runs, the escaper's distance to $e$ decreases at most $\frac{d(H,e)}{d(Z,e)}$ times faster than the pursuer's distance to $e$ does; then the pursuer reaches the boundary first and wins by Theorem~\ref{WedgeTheorem}. Otherwise, there's some direction $\theta$ such that the escaper's distance to $e$ decreases more than $\frac{d(H,e)}{d(Z,e)}$ times faster than the pursuer's distance to $e$ does. So, if the escaper runs straight in the direction $\theta$ toward a point $W_3$ on $e$, the escaper reaches $W_3$ before the pursuer reaches $e$ (since if the pursuer reaches $e$, it wins by Theorem~\ref{WedgeTheorem}) following this strategy. But if the pursuer runs straight toward $W_3$, it gets there first by assumption; so if the pursuer runs straight toward $W_3$ but slows down enough to keep the line between it and the escaper parallel to $HZ$, it still wins. If the pursuer does that, but uses any extra movement to move toward the boundary, that brings it strictly closer to $W_3$, so it still wins. But that's exactly the pursuer strategy for which we claimed that the pursuer would lose, contradiction. So there are points of $W$ on at least two edges within $\epsilon$ of $H$.

So, there are two edges $e$ and $f$ (there cannot be more, by the definition of $\epsilon$) with one point of $W$, within $\epsilon$ of $H$, on each. To deal with this case, we prove three lemmas about the geometry of the situation.

\begin{lemma}\label{lem:angleinequality}
The angle between the escaper's shortest paths to the points in $W$ is at least the angle between the pursuer's shortest paths to the points in $W$, with equality only if both angles are $\pi$.
\end{lemma}

\begin{proof}
Let the points of $W$ be $W_1 = (x,y)$ on edge $e$ and $W_2$ on edge $f$ and the $x$ axis, which meet at $P = (0,t)$ with $t > 0$. The escaper and pursuer are on opposite sides of at least one of the supporting lines of $e$ and $f$. We divide into two cases: either they are on opposite sides of both, or they are on opposite sides of just one.

If the escaper and pursuer are on opposite sides of both supporting lines of $e$ and $f$, as in Figure~\ref{fig:marginofvictoryeasy}, then $ZW_1 > HW_1$ because $ZW_1$ and $HW_1$ are the shortest paths for the pursuer and escaper, respectively, to $W_1$, both players reach $W_1$ in the same time, and the pursuer is faster, so the pursuer's path is longer. Therefore, $\angle W_1ZH < \angle W_1HZ$. Similarly, $\angle W_2ZH < \angle W_2HZ$, so $\angle W_1ZW_2 < \angle W_1HW_2$, as desired.

\begin{figure}[h]
	\centering
	\includegraphics[scale=0.7]{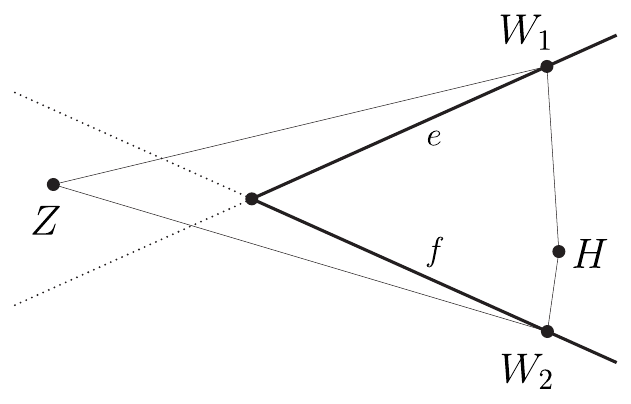}
	\caption{Coordinates and variables used in the easy case of the proof of Lemma~\ref{lem:angleinequality}.}
	\label{fig:marginofvictoryeasy}
\end{figure}

If the escaper and pursuer are on opposite sides of just one of the edges $e$ and $f$, then, without loss of generality, let them be on opposite sides of the supporting line of $f$. Let the escaper's shortest path to $W_2$ be horizontal, let the pursuer's position be $Z = (z, \zeta)$, and let the escaper's position be $H = (h,0)$, so $O = (0,0)$ is the foot of the perpendicular from $P$ to the escaper's shortest path to $W_2$, as in Figure~\ref{fig:marginofvictorymain}.

\begin{figure}[h]
	\centering
	\includegraphics[scale=0.3]{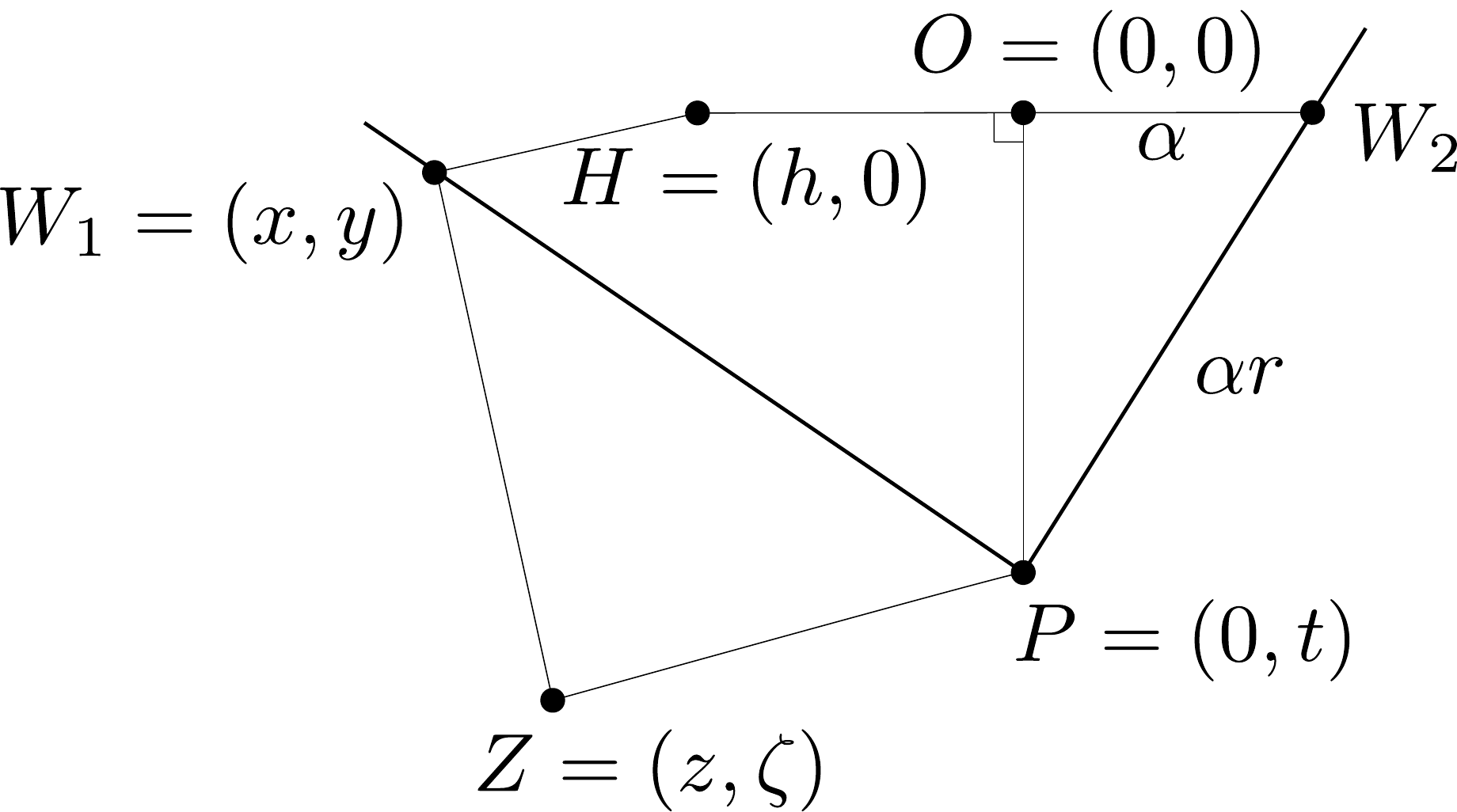}
	\caption{Coordinates and variables used in the hard case of the proof of Lemma~\ref{lem:angleinequality}.}
	\label{fig:marginofvictorymain}
\end{figure}

\paragraph{Ramchundra's intercept problem.}

We first show that, when the pursuer is at $P$, the escaper's shortest-path time to some point on $f$ equals the pursuer's shortest-path time if and only if the escaper is on $PO$. (This is the ``perpendicular to first sighting'' rule for naval pursuit, also known as ``Ramchundra's intercept problem''~\cite[Section~1.5]{Nahin-2007}.) When the escaper and pursuer both run by their shortest paths to $W_2$, and the pursuer is at $P$, let the escaper be at a point $O'$, and let $W'$ be any point on $f$. By the law of sines, $\frac{\overline{O'W'}}{\overline{PW'}} = \frac{\sin(\angle O'PW')}{\sin(\angle PO'W')}$, and $\angle O'PW'$ is fixed, so $\frac{\overline{O'W'}}{\overline{PW'}}$ is maximized over choices of $W'$ when $\angle PO'W' = \frac{\pi}{2}$. But if $O'$ is the point where the escaper is when the pursuer's at $P$ as both follow their shortest paths to $W_2$, then, by the choice of coordinate system, the escaper can tie only by running horizontally to $W_2$, so $W' = W_2$ and $O'$ is the point on the escaper's shortest path with $PO'W_2 = \frac{\pi}2$, that is, $O'=O$.

By the perpendicular to first sighting rule, if the escaper can escape, they can escape by running perpendicular to the direction to the pursuer. So, when the pursuer is at $P$, the escaper must be at the foot of the perpendicular from $P$ to the escaper path, that is, at $O$, and that's after the escaper has traveled a distance of $HO$ and the pursuer has traveled a distance of $ZP$, so $\frac{\overline{ZP}}{\overline{HO}} = r$.

\paragraph{Tied time to $W_1$, in coordinates.}

Writing $(\overline{ZW_1})^2 = (r \overline{HW_1})^2$ out in coordinates, $$(z-x)^2 + (\zeta - y)^2 = r^2[(x-h)^2+y^2].$$
Consider the function from a point $p$ on $e$ to the difference between the escaper's shortest-path time to $p$ and the pursuer's shortest-path time to $p$. At $W_1$, that difference is $0$, and near $W$, it's nonnegative, so the derivative is $0$ at $W_1$. In coordinates, the difference at a point near $W_1$ is $$(z-x- x d\ell)^2 + (\zeta - y - (y-t)d\ell)^2 = r^2[(x+xd\ell -h)^2+(y + (y-t)d\ell)^2],$$ so the derivative gives us $$(z-x)x + (\zeta - y)(y-t) = r^2[(h-x)x+(t-y)(y)].$$ Adding the equation $(\overline{ZW_1})^2 = (r \overline{HW_1})^2$ gives $$(z-x)z + (\zeta - y)(\zeta-t) = r^2[(h-x)h+ty],$$
an equation that will be useful in two cases:

If $y = 0$, then the angle between the escaper's shortest paths is $\pi$, so the conclusion of Lemma~\ref{lem:angleinequality} is satisfied.

If $y > 0$, then $yt > 0$, so $h(h-x) < h(h-x)+yt$, so $$\frac{h(h-x)}{\overline{HW_1}\overline{HO}} < \frac{h(h-x)+yt}{\overline{HW_1}\overline{HO}}.$$ By Ramchundra's intercept problem, $\frac{\overline{ZP}}{\overline{HO}} = r$, and by the definition of $W$, $\frac{\overline{ZW_1}}{\overline{HW_1}} = r$, so $$\frac{h(h-x)}{\overline{HW}\overline{HO}} < \frac{r^2(h(h-x)+yt)}{\overline{ZW_1}\overline{ZP}}.$$ By the tied time to $W_1$ in coordinates, that's $$\frac{h(h-x)}{\overline{HW_1}\overline{HO}} < \frac{(z-x)z + (\zeta - y)(\zeta-t)}{\overline{ZW_1}\overline{ZP}}.$$ Each of those numerators is a dot product:
$$\frac{HW_1 \cdot HO}{\overline{HW_1}\overline{HO}} < \frac{ZW_1 \cdot ZP}{\overline{ZW_1}\overline{ZP}}.$$
That is, $\cos(\angle OHW_1) < \cos(\angle W_1ZP)$, so $\angle OHW_1 > \angle W_1ZP$, as claimed.

If $y < 0$, we again have, by the tied time to $W_1$ in coordinates, that $$(z-x)z + (\zeta - y)(\zeta-t) = r^2[(h-x)h+ty].$$ By the Cauchy-Schwarz inequality, $$((z-x)z + (\zeta - y)(\zeta-t))^2 \le [(z-x)^2 + (\zeta - y)^2][z^2 + (\zeta - t)^2] = \overline{ZW_1}^2\overline{ZP}^2.$$ By Ramchundra's intercept problem and the fact that $\frac{\overline{ZW_1}}{\overline{HW_1}} = r$, that's $$r^4[(h-x)h+ty]^2 \le r^4h^2[(x-h)^2+y^2],$$ so $$h^2y^2 \ge 2(h-x)hty + t^2y^2 \ge 2(h-x)hty.$$ Since $h > 0$ and $y < 0$, $hy \le 2(h-x)t$. But $(h,0)$ is on the escaper side of edge $e$, so $hy > (h-x)t$, contradiction.

Therefore, in every surviving case, the conclusion of Lemma~\ref{lem:angleinequality} is satisfied.
\end{proof}

We know that, if the escaper moves straight toward a point $W_1$, there exists a pursuer strategy (a direction of pursuer movement) such that the pursuer does not fall behind in the race toward $W_1$. We now prove that that strategy is stable: if the escaper moves at an angle of $\theta$ from $W_1$, and the pursuer moves at an angle less than $\theta$ from its shortest path to $W_1$, then for a positive time, the invariant that the pursuer's distance to \emph{every} point on the edge containing $W_1$ remains at most $r$ times the escaper's distance.

\begin{lemma}\label{lem:marginofvictoryharderhaplo}
Suppose the escaper and pursuer are on the same side of (the supporting line of) an edge $f$ containing a point $W_2$ such that the the pursuer's shortest path to $W_2$ is $r$ times longer than the escaper's shortest path to $W_2$. If the escaper moves a short distance $dt$ at an angle of $\theta$ from its shortest path to $W_2$, and the pursuer moves a short distance $rdt$ at an angle at most $\theta$ from its shortest path to $W_2$, then, for every point on $f$, the pursuer's shortest-path time to it remains at most the escaper's shortest-path time to it.
\end{lemma}

\begin{proof}
Let the end of $f$ to which the pursuer runs be $(0,0)$; let the perpendicular $\ell$ from $(0,0)$ to the escaper's shortest path be at an angle $\theta_\ell$ (so points $(x,y)$ on it have $x \sin \theta_\ell - y \cos \theta_\ell = 0$), let the pursuer's position be $(a,b)$, and let the escaper's position be $(c,d)$, all as in Figure~\ref{fig:marginofvictorywrapup}. Then the escaper's distance to $\ell$ is $c \sin \theta_\ell - d \cos \theta_\ell$ and the pursuer's distance to $(0,0)$ is $\sqrt{a^2+b^2}$. By Ramchundra's intercept problem, if the escaper is on $\ell$ at the same time as the pursuer reaches $(0,0)$, the escaper cannot win a race to anywhere on $f$. So, the pursuer's distance to $(0,0)$ is currently $r$ times the escaper's distance to $\ell$, and it suffices for the pursuer to maintain that invariant.

\begin{figure}[h]
	\centering
	\includegraphics[scale=0.4]{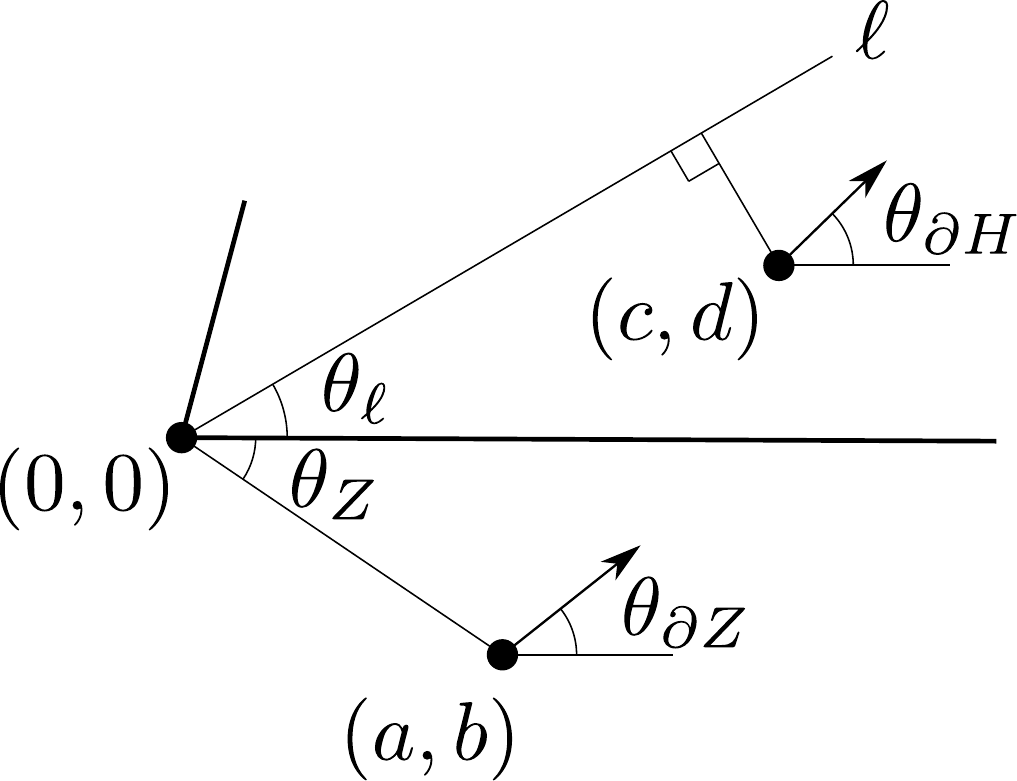}
	\caption{Coordinates and variables used in the proof of Lemma~\ref{lem:marginofvictoryharderhaplo}.}
	\label{fig:marginofvictorywrapup}
\end{figure}

Suppose the escaper moves in a direction $\theta_{\partial H}$; that is, $(\partial c, \partial d) = (\cos \theta_{\partial H}, \sin \theta_{\partial H})$. (If the escaper moves at less than full speed, the pursuer can reduce its speed proportionally.) The direction directly toward $\ell$ is $\ell + \frac{\pi}2$, so the escaper's angle from that direction is $|\frac{\pi}2 + \theta_\ell - \theta_{\partial H}|$. If the pursuer's angle from $(0,0)$ is $\theta_Z$ (so $(a,b) = (\sqrt{a^2+b^2}\cos\theta_Z, \sqrt{a^2+b^2}\sin\theta_Z)$ and the pursuer's direction toward $(0,0)$ is $\pi + \theta_Z$), we will have the pursuer move in any direction $\theta_{\partial Z}$ (that is, $(\partial a, \partial b) = (r\cos \theta_{\partial Z}, r\sin \theta_{\partial Z})$) such that $|\theta_{\partial Z} - (\pi + \theta_Z)| \le |\frac{\pi}2 + \theta_\ell - \theta_{\partial H}|$. Then $\cos(\pi + \theta_Z - \theta_{\partial Z}) \ge \cos(\frac{\pi}2 + \theta_\ell - \theta_{\partial H})$, so $\cos\theta_Z \cos\theta_{\partial Z} + \sin \theta_Z\sin\theta_{\partial Z} \le \cos\theta_{\partial H}\sin\theta_\ell - \sin\theta_{\partial H}\cos\theta_\ell$, or $a \partial a + b\partial b \le r\sqrt{a^2+b^2} (\sin\theta_\ell \partial c - \cos\theta_\ell\partial d)$. The escaper's distance to $\ell$ is $c \sin \theta_\ell - d \cos \theta_\ell$ and the pursuer's squared distance to $(0,0)$ is $\sqrt{a^2+b^2}$, so $(c \sin \theta_\ell - d \cos \theta_\ell)r = \sqrt{a^2+b^2}$, and $a \partial a + b\partial b \le r^2(c \sin \theta_\ell - d \cos \theta_\ell) (\sin\theta_\ell \partial c - \cos\theta_\ell\partial d)$. The left side is the derivative of the pursuer's squared distance to $(0,0)$ and the right side is $r^2$ times the derivative of the escaper's squared distance to $\ell$, so the pursuer's shortest-path time to $(0,0)$ decreases at least as fast as the escaper's shortest-path time to $\ell$, as desired.
\end{proof}

\begin{lemma}\label{lem:marginofvictoryoffangle}
Suppose the escaper and pursuer are on opposite sides of (the extensions of) an edge $e$ containing a point $W_1$ such that the the pursuer's shortest path to $W_1$ is $r$ times longer than the escaper's shortest path to $W_1$. If the escaper moves a short distance $dt$ at an angle of $\theta$ from its shortest path to $W_1$, and the pursuer moves a short distance $rdt$ at an angle at most $\theta$ from its shortest path to $W_1$, then, for every point on $e$, the pursuer's shortest-path time to it remains at most the escaper's shortest-path time to it.
\end{lemma}

\begin{proof}
Let edge $e$ be the $x$-axis, let the pursuer's position be $(a,b)$, and let the escaper's position be $(c,d)$, as in Figure~\ref{fig:marginofvictoryoffangle}. Also, place the origin such that $a = r^2 c$ (which may be a translation from the coordinates used in the proof of the previous lemma); this is possible since $r^2 \neq 1$.

\begin{figure}[h]
	\centering
	\includegraphics[scale=0.9]{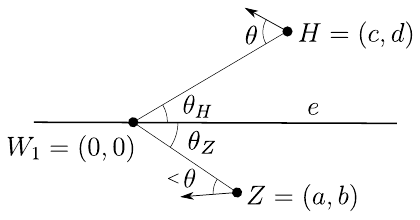}
	\caption{Coordinates and variables used in the proof of Lemma~\ref{lem:marginofvictoryoffangle}.}
	\label{fig:marginofvictoryoffangle}
\end{figure}

We first claim that the pursuer's shortest-path distance to every point on $e$ is at most $r$ times the escaper's shortest-path distance if and only if $r^2d - b \ge r\sqrt{(a-c)^2 + (b-d)^2}$, with equality if and only if there's a point on $e$ for which the pursuer's shortest-path distance equals $r$ times the escaper's shortest-path distance. Indeed, a point $(x,y)$ has the pursuer's plane distance more than $r$ times the escaper's plane distance if and only if $(x-a)^2 + (y-b)^2 \ge r^2\left[(x-c)^2 + (y-d)^2\right]$, or $\left(x - \frac{r^2c-a}{r^2-1}\right)^2 + \left(y - \frac{r^2d-b}{r^2-1}\right)^2 \le \frac{r^2}{(r^2-1)^2}\left((c-a)^2 + (d-b)^2\right)$. That describes a circle of radius $\frac{r}{r^2-1} \sqrt{(c-a)^2 + (d-b)^2}$ centered at $(\frac{r^2c-a}{r^2-1},\frac{r^2d-b}{r^2-1})$, which is strictly above the $x$-axis if $r^2 d-b > \sqrt{(c-a)^2 + (d-b)^2}$ and is tangent to it at $(\frac{r^2c-a}{r^2-1},0)$ if they are equal, as claimed.

Since we chose $r^2c=a$, the point of tangency (that is, $W_1$) is $(0,0)$. Let $\theta_H$ and $\theta_Z$ be the angles from the origin to $H$ and $Z$, respectively, so the escaper's direction to the origin is $\pi+\theta_H$ and the pursuer's is $\pi+\theta_Z$. 
Suppose the escaper moves in a direction $\theta_{\partial H}$; that is, $(\partial c, \partial d) = (\cos \theta_{\partial H}, \sin \theta_{\partial H})$. (If the escaper moves at less than full speed, the pursuer can reduce its speed proportionally.) The direction directly toward $(0,0)$ is $\pi + \theta_H$, so the escaper's angle from that direction is $|\pi + \theta_H - \theta_{\partial H}|$. We will have the pursuer move in any direction $\theta_{\partial Z}$ (that is, $(\partial a, \partial b) = (r\cos \theta_{\partial Z}, r\sin \theta_{\partial Z})$) such that $|\pi + \theta_Z - \theta_{\partial Z}| \le |\pi + \theta_H - \theta_{\partial H}|$. Then $\cos(\theta_{\partial Z} -\theta_Z) \le \cos(\theta_{\partial H} - \theta_H)$. Also, $\sqrt{a^2+b^2} = r\sqrt{c^2+d^2}$, so $\sqrt{a^2+b^2}\cos(\theta_{\partial Z} -\theta_Z) \le r\sqrt{c^2+d^2}\cos(\theta_{\partial H} - \theta_H)$. Those are the coordinate expansions of dot products: $a \cos \theta_{\partial Z} + b \sin \theta_{\partial Z} \le r^2 \left[c \cos \theta_{\partial H} + d \sin \theta_{\partial H}\right]$. Plugging in $\partial a = r \cos \theta_{\partial Z}$ and so on gives $a \partial a + b \partial b \le r^2 \left[ c \partial c + d \partial d\right]$. Plugging in $a = r^2c$, multiplying by $1-r^2$ (which is negative), and rearranging gives $(r^2d-b)(r^2 \partial d - \partial b) \ge r^2 (a-c)(\partial a - \partial c) + r^2 (b-d)(\partial b - \partial d)$. The left side is the derivative of $(r^2d-b)^2$ and the right side is the derivative of $\left(r\sqrt{(a-c)^2 + (b-d)^2}\right)$, so the chosen direction of pursuer movement maintains $r^2 - d \ge r\sqrt{(a-c)^2 + (b-d)^2}$, as desired.
\end{proof}

Finally, we can complete the proof of Lemma~\ref{MarginOfVictoryLemma},
by describing an APLO-like strategy for the pursuer to win $G_{\epsilon_G}$ (for all $\epsilon_G > 0$) as long as the escaper moves a distance at most $\eps$, contradicting the assumption that the escaper could win $G$ with at most $\eps$ more movement. To win $\epsilon_G$, the pursuer will use an $\epsilon_G$-oblivious strategy, so it can respond to at least $\epsilon_G$'s worth of escaper movement, and so define a direction of escaper movement for each time step. By the previous two lemmas, if the pursuer can, at all times, move in a direction closer to its shortest path to each of $W_1$ and $W_2$ than the escaper's direction of movement is to its shortest path to each of $W_1$ and $W_2$, the escaper cannot win a race to any point on either $e$ or $f$. By Lemma~\ref{lem:angleinequality}, the pursuer can do so, so the pursuer wins $G_{\epsilon_G}$.

\xxx{This proof could probably use some diagrams.}

So as long as the escaper and pursuer stay within that circle of radius $2\sqrt{\epsilon_0}$, the escaper cannot win, contradicting the assumption that, a moment later, the escaper could win by running straight a distance at most~$\epsilon$.

In every surviving case, the escaper can win $G_{\epsilon^3}$ with speed ratio $r\frac{1}{1+\epsilon}$, as desired.
\end{proof}

\subsection{Algorithm}

\begin{theorem}[pseudopolynomial-time approximation scheme] \label{pseudoPTAS}
  Given a polygon with integer vertex coordinates $\in [0,N]$,
  defining the escaper domain $D_\escaper$ as its interior and boundary,
  the exit set $X$ as its boundary,
  and the pursuer domain $D_\pursuer$ as its boundary and optionally its exterior,
  there is an $(N/\epsilon)^{O(1)}$-time approximation algorithm for
  $\epsilon$-approximating the critical speed ratio $r^*$ in~$G$:
  the algorithm computes a speed ratio $r$ such that
  $(1-\epsilon) r \leq r^* \leq (1+\epsilon) r$.
\end{theorem}

\begin{proof}
At the top level, our algorithm uses a binary search to evaluate~$r^*$.
To this end, first we give easily computable bounds on the range of~$r^*$.
\xxx{The following bounds are all exploiting the polygon case.  We should be
  able to do it for bounded regions...}
As a lower bound, $r^* \geq 1$; otherwise, the escaper can win along
a single edge, as in the halfplane analysis (Theorem~\ref{HalfplaneTheorem}).
As an upper bound, $r^* \leq 10.89898 \max_{p,q \in \Exit} \frac{d_\pursuer(p,q)}{d_\escaper(p,q)}$
by Theorem~\ref{ConstantApproximationTheorem}.  Instead of computing this
quantity directly, we can easily compute an upper bound as described in
Point~\ref{easy upper bound for r^*} after Lemma~\ref{MarginOfVictoryLemma}.
As both quantities are pseudopolynomial, we get an interval containing $r^*$
of pseudopolynomial length.  The overhead for binary search will be a factor
logarithmic in this interval length, which is even polynomial.

It thus remains to give an approximate binary decider for binary search:
given a speed ratio $r$ (from binary search), decide in pseudopolynomial time
whether $r^* < (1-\epsilon) r$ or $r^* > (1+\epsilon) r$,
with the freedom to return either answer if
$(1-\epsilon) r \leq r^* \leq (1+\epsilon) r$.

A key ingredient is that we can compute the winner for the discrete game
$\hat G_{\delta,\gamma}(r)$ for any $\delta,\gamma,r$ in pseudopolynomial time.
First, in the exterior model\xxx{generalize}, we restrict to the convex hull
of $R_\pursuer$ by Lemma~\ref{ConvexHullLemma}.  
Then we compute the graph with vertices $V$ and edges $E_\escaper \cup E_\pursuer$.
This graph has pseudopolynomial size, as the area of the convex hull of $R_\pursuer$
and the perimeter of $R_\escaper$ are both pseudopolynomial.
Thus the number of states --- consisting of the current escaper and pursuer
positions, the previous escaper and pursuer positions to check the win condition,
and whose move is next --- is also pseudopolynomial.  
We can thus compute all winning positions in the discrete game by marking all
game states for which the escaper immediately wins (being adjacent to a vertex
$\exit$ of $B_\exit$ for two moves such that the pursuer still is not adjacent to~$\exit$),
then repeatedly, mark any game state as an escaper win if either it is the
escaper's turn and they can move to any game state already marked an escaper win, or
it is the pursuer's turn and every game state they can move to is already
marked an escaper win.
After at most as many rounds as the pseudopolynomial number of game states,
every game state from which the escaper wins will be so marked because,
at each round, either at least one game state not previously marked as an escaper
win will be so marked or no new game states will be marked and every following
round will be the same.
(This is essentially the finite case of the open determinacy theorem
\cite{Gale-Stewart-1953} exploited in Lemma~\ref{discrete game winner}.)
Then the escaper wins the discrete game if and only if there is an escaper starting
position $s_\escaper$ such that, for every pursuer starting position $s_\pursuer$, the state
with the escaper at $s_\escaper$, the pursuer at $s_\pursuer$, and the pursuer to move
is marked as an escaper win.

First suppose that the discrete game $\hat G_{\delta,\gamma}(r)$
has an escaper winning strategy, where
$\gamma < \min\{{1 \over 4}, {r \over 2}, {1 \over 2} \epsilon r\} \delta$.
By Theorem~\ref{discrete -> continuous},
the continuous game $G_\epsilon(r')$ and thus $G(r')$
has an escaper winning strategy where
$r' = r - 2 {\gamma \over \delta} > r - \epsilon r = (1-\epsilon) r$,
so $r^* > (1-\epsilon) r$.

On the other hand, if $r^* > (1+\epsilon) r$,
then the escaper wins $G((1+\epsilon) r)$.
By Lemma~\ref{MarginOfVictoryLemma}, there is an escaper winning strategy for
$G_{\hat \epsilon^3}(\frac{1+\epsilon}{1+\hat \epsilon} r)$
for any $\hat \epsilon \leq \epsilon_0$,
where $\epsilon_0$ is computed according
to the algorithm after Lemma~\ref{MarginOfVictoryLemma}.
By Corollary~\ref{can't both win},
there is no pursuer winning strategy for the same game.
Let $\delta = 2 \epsilon_0^3 / r$,
so that $\hat \epsilon^3 \leq {1 \over 2} r \delta$.
By the contrapositive of Theorem~\ref{discrete -> continuous},
the discrete game $\hat G_{\delta,\gamma}(\frac{1+\epsilon}{1+\hat \epsilon} (1 - {2 \gamma \over \delta}) r)$ has no pursuer winning strategy.
By Lemma~\ref{discrete game winner}, the same game has an escaper winning strategy.
Because decreasing the speed ratio only removes pursuer moves,
$\hat G_{\delta,\gamma}(r)$ has an escaper winning strategy provided
$\frac{1+\epsilon}{1+\hat \epsilon} (1 - {2 \gamma \over \delta})
\geq 1$.
If we further constrain that
$\gamma \leq \delta({\epsilon \over 4} (1 + \hat \epsilon) - {\hat \epsilon \over 2})$
(which we can make positive by setting $\hat \epsilon$ small enough),
then
${2 \gamma \over \delta} \leq {\epsilon \over 2} (1 + \hat \epsilon) - \hat \epsilon$,
so
$1 - {2 \gamma \over \delta} \geq 1 - {\epsilon \over 2} (1 + \hat \epsilon) + \hat \epsilon
= (1 + \hat \epsilon) (1 - {\epsilon \over 2})$,
so $\frac{1+\epsilon}{1+\hat \epsilon} (1 - {2 \gamma \over \delta}) \geq
(1+\epsilon) (1 - {\epsilon \over 2}) = 1 + {\epsilon \over 2} - {\epsilon^2 \over 2} \geq 1$
provided $\epsilon \leq 1$.

Therefore, assuming $r^*$ is not in $((1-\epsilon) r, (1+\epsilon) r)$,
we have $r^* > (1+\epsilon) r$ if and only if
$\hat G_{\delta,\gamma}(r)$ has an escaper winning strategy.
So we can compute the winner of $\hat G_{\delta,\gamma}(r)$ to decide whether
$r^* > (1+\epsilon) r$ or $r^* < (1-\epsilon) r$,
enabling the binary search.
\end{proof}

For a related pursuit--evasion problem (can a polyhedral evader reach a goal
point while avoiding a polyhedral pursuer, given maximum speeds for each?),
Reif and Tate \cite{Reif-Tate-1993} give what might seem like a
pseudopolynomial-time approximation scheme.
Specifically, they give an $(n/\epsilon)^{O(1)}$-time algorithm to find an
evasion strategy if there is an ``$\epsilon$-safe'' evasion strategy
that stays $\epsilon$ away from the pursuer and all obstacles.
They also prove this result with a similar approach to discretizing the
continuous game.  However, to turn such an algorithm into an approximation
algorithm for computing the critical speed ratio requires a relation between
tweaking the speed ratio and guaranteeing a safety distance.  This relation is
precisely the point of our margin-of-victory Lemma~\ref{MarginOfVictoryLemma},
which is the bulk of our proof.

\section{NP-hardness for Two Players in 3D}
\label{sec:NPhard}
\label{appendix:NPhard}

  In this section,
we prove that the pursuit--escape problem is NP-hard
for polyhedral domains in 3D.
Our proof is an easy extension of the famous result by Canny and Reif
\cite{Canny-Reif-1987} that it is weakly NP-hard to find shortest paths in 3D
amidst polyhedral obstacles.

\begin{theorem} \label{thm:NPhard}
  It is weakly NP-hard to calculate the critical speed ratio $r^*$
  for a pursuit-escape problem with polyhedral domains in 3D,
  with or without specified starting positions,
  and even if $D_\escaper$ and $D_\pursuer$ are disjoint
  except at $\Exit$ which consists of at most two points.
\end{theorem}

\begin{proof}
  We begin by showing the problem hard with specified starting positions for
  the players, and with arbitrary intersections between $D_\escaper$ and
  $D_\pursuer$.  Then we adapt the construction to work without
  specified starting positions, with minimal intersection between
  $D_\escaper$ and $D_\pursuer$, and to make both $D_\escaper$ and $D_\pursuer$
  proper polyhedra (without lower-dimensional degeneracy).

\paragraph{Specified starting positions.}

  Our reduction follows Canny and Reif's reduction from 3SAT to finding a path
  of length $\leq \ell$ from $s$ to $t$ in a 3D polyhedral environment
  under any $L_p$ metric \cite{Canny-Reif-1987}.
  The escaper domain $D_\escaper$ is exactly the polyhedral environment
  in Canny and Reif's construction.
  The escaper's start location is the start location $s$, and the exit set $X$
  consists of a single point, namely, the target location $t$.
  Next, the pursuer domain $D_\pursuer$ is a straight line between $t$
  and any point $s_\pursuer$ at distance $\ell + \epsilon$
  (slightly more than the target path length) from~$t$.
  The pursuer's start location is $s_\pursuer$.

  The pursuer's optimal strategy is to run directly from $s_\pursuer$ to the
  unique exit location $t$ and staying there.  If the pursuer arrives
  $\epsilon$ before the escaper, then the pursuer wins, and vice versa.
  Thus, if the escaper can find a path of length $\leq \ell$ between $s$ and
  $t$, then the escaper can win and the critical speed ratio is greater
  than $1$.  Conversely, if all paths have length $\geq \ell+\epsilon$,
  then any escaper strategy cannot arrive before the pursuer,
  so the pursuer wins with a speed ratio of~$1$.
  As argued in \cite[Corollary~2.3.4]{Canny-Reif-1987}, there is a gap
  of at least $2^{-2 n m - 3 n - 3}$ in path length between positive and
  negative instances (where $n$ is the number of variables in $m$ is the
  number of clauses in the 3SAT formula), so setting
  $\epsilon = 2^{-2 n m - 3 n - 3}$ completes the reduction.

\paragraph{Unspecified starting positions.}

  The construction of the escaper domain $D_\escaper$ remains exactly the
  polyhedral environment in Canny and Reif's construction.
  But now the exit set $\Exit = \{s, t\}$ consists of both the start
  and target locations.
  The pursuer domain $D_\pursuer$ is a Manhattan from $s$ to $t$,
  so that its length $\ell_\pursuer$ is easy to compute as the sum of
  coordinate differences between $s$ and~$t$.
  We set $r = \ell_\pursuer / (\ell + \epsilon)$,
  at which an escaper path of length $\ell+\epsilon$
  takes the same time as the pursuer traversing the entire path of $D_\pursuer$,
  and ask whether the critical speed ratio $r^* \leq r$.

  If there is a path in $D_\escaper$ from $s$ to $t$ of length $\leq \ell$,
  then we construct a winning escaper strategy with speed ratio~$r$.
  The strategy starts at $s$ which, because $s \in X$, forces the pursuer
  to also start at~$s$.
  Then the strategy runs to $t$ along the path of length $\leq \ell$,
  oblivious to movement by the pursuer.
  The pursuer will remain at least $\epsilon$ away from~$t$,
  so the escaper escapes.

  If all paths in $D_\escaper$ from $s$ to $t$ have length $\geq \ell+\epsilon$,
  then we construct a winning pursuer strategy with speed ratio~$r$.
  For any escaper location $\escaper \in D_\escaper$, the pursuer computes
  the shortest-path distances $d_\escaper(s,\escaper)$ and $d_\escaper(\escaper,t)$.
  (This strategy is expensive to compute, but all we need is that it exists.)
  By the triangle inequality,
  $$d_\escaper(s,\escaper) + d_\escaper(\escaper,t) \geq d_\escaper(s,t) \geq \ell+\epsilon.$$
  We define the pursuer strategy $\Pursuer(\escaper)$ to be the unique point
  along the path $D_\pursuer$ that satisfies
  $$\frac{d_\pursuer(\Pursuer(\escaper),s)}{d_\pursuer(\Pursuer(\escaper),s)+d_\pursuer(\Pursuer(\escaper),t)} = \frac{d_\escaper(\escaper,s)}{d_\escaper(\escaper,s)+d_\escaper(\escaper,t)}.$$
  If $\escaper$ varies with speed $\leq 1$, then $\Pursuer(\escaper)$ varies
  with speed $\leq r$.  The strategy is history-independent, so is a valid
  pursuer strategy.
  Because $\Pursuer(s) = s$ and $\Pursuer(t) = t$, $\Pursuer$ is in fact
  a winning pursuer strategy.
	
\paragraph{Disjoint regions.}

  Next we achieve the property that $D_\escaper \cap D_\pursuer = X$,
  whereas currently the line segment $D_\pursuer$ might intersect $D_\escaper$
  at other intermediate points.
  In Canny and Reif's construction, almost all of the polyhedral region
  $D_\escaper$ is ``thin'',
  with a maximum width of $w = 1/2^{\Theta(n m)}$.
  They show that an additive change of $O(w)$ to the path length does not
  affect the hardness reduction.
  Thus, we can safely move the exits in $\Exit$ from $s,t$ to the nearest
  boundary faces of the polyhedral region~$D_\escaper$.
  Then we can modify $D_\pursuer$ to a path between the two points of $\Exit$
  that avoids otherwise intersecting~$D_\escaper$.
  Again we set $\ell_\pursuer$ to the length of this path,
  and set $r = \ell_\pursuer / (\ell + \epsilon)$ as before.
  The rest of the argument works as above.

\paragraph{Polyhedral domains.}
	
  Finally, we show how to thicken the pursuer path so that the pursuer domains
  $D_\pursuer$ is a proper polyhedron instead of a one-dimensional path.
  When we lay out Canny and Reif's construction, we ensure that the
  first path splitter visited after the start location has no other gadget
  above it,
  and that the final clause filter visited has no other clause below it.
  These properties ensure that the start and end positions $s,t$ each has
  an orthogonal ray that does not intersect the rest of the construction.
  We set $X$ to the intersection of $\partial D_\escaper$ with these rays;
  these two points are still within $O(w)$ of $s$ and $t$ respectively.
  Now we can route the path $D_\pursuer$ orthogonality out and around the
  bounding box of $D_\escaper$, keeping it at least $O(w)$ distance away
  from any part of $D_\escaper$ and using at most six turns.

  We now construct a polyhedral pursuer domain $D'_\pursuer$ based on
  the orthogonal path $D_\pursuer$.  For all parts of $D_\pursuer$ more than
  $w$ away from the bounding box of $D_\escaper$, we make $D'_\pursuer$ an
  orthotube centered on $D_\pursuer$ with orthogonal thickness $w/24$.
  Then we connect the ends of these tubes to their respective closer
  point in $X$ via two pyramid caps that do not intersect $D_\escaper$.
  The new pursuer shortest path has gotten smaller from the ability to
  shortcut corners in the orthotube, but the change in distance remains
  less than $w$, and so still within the additive factor for which
  Canny and Reif's construction works.
\end{proof}

\section{Multiple Escapers and Pursuers}
\label{MultipleZombiesSection}

In this section, we prove stronger computational hardness of computing or
approximating critical speed ratio in broader models of pursuit--escape
problems. All of the hardness proofs require that there be multiple pursuers,
not just one, such that any one of them can block the escaper's escape.
Some will also require that there be multiple escapers,
who win if at least one escapes. 
First we generalize our model to allow multiple escapers and pursuers (Section~\ref{sec:multi-model}).
To make the hardness proofs more interesting, we discuss some positive
results as well, in Sections~\ref{sec:multi-human}--\ref{sec:slow-pursuers}.
Then Section~\ref{HardnessSection} describes the hardness results.

\subsection{Model}
\label{sec:multi-model}

First we describe the necessary extensions to the single-escaper single-pursuer
model of Section~\ref{sec:model} and \Section~\ref{appendix:model}
to handle multiple escapers and pursuers.
Suppose there are $n_\escaper$ escapers and $n_\pursuer$ pursuers.
We define a two-player game where the escaper player controls all $n_\escaper$ escapers
and the pursuer player controls all $n_\pursuer$ pursuers.
We refer to the $n_\escaper + n_\pursuer$ escapers and pursuers as \defn{individuals}.

\paragraph{Domains.}
The definition of ``domain'' remains unchanged, but now instead of a single domain for each player, the input specifies a set of domains for each player and an integer \defn{capacity} for each domain representing an upper bound on the number of individuals a player can place on the respective domain.
We assume that every escaper domain and every pursuer domain intersect in a measure-zero set (possibly empty).
We allow two domains of the same player to intersect, but still forbid
individuals from jumping across domains at such intersections;
they must remain in their originally assigned domain.
We are also given a set of (escaper) exit locations, which must be a subset of the
union of all pursuer domains.
In this setting, the polygon model restricts the escaper domain set to contain a single simple polygon with infinite capacity.
Similarly the external and moat models are defined with infinite capacity.

\paragraph{Strategies.}
The definitions of ``pursuer motion path'' and ``escaper motion path'' remain
unchanged, but now a player strategy involves multiple such paths.
Suppose the player has $n_p$ individuals (either $n_\escaper$ or $n_\pursuer$)
and the opponent has $n_o$ individuals (either $n_\pursuer$ or $n_\escaper$).
A \defn{player strategy} is a function $\Player$ mapping $n_o$ opponent motion paths
$\opponent_1, \opponent_2, \dots, \opponent_{n_o}$ to $n_p$ player motion paths
$\Player_i(\opponent_1, \opponent_2, \dots, \opponent_{n_o})$ for $i \in \{1, 2, \dots, n_p\}$
satisfying the following \defn{nonbranching-lookahead constraint}:
\begin{quote}
for any opponent motion paths $\opponent_1, \opponent_2, \dots, \opponent_{n_o},
\tilde \opponent_1, \tilde \opponent_2, \dots, \tilde \opponent_{n_o}$ such that
$\opponent_j$ and $\tilde \opponent_j$ agree on $[0,t]$ for all $j \in \{1, 2, \dots, n_o\}$,
the strategy's player motion paths $\Player_i(\opponent_1, \opponent_2, \dots, \opponent_{n_o})$
and $\Player(\tilde \opponent_1, \tilde \opponent_2, \dots, \tilde \opponent_{n_o})$
also agree on $[0,t]$ for all $i \in \{1, 2, \dots, n_p\}$.
\end{quote}
In addition, an \defn{escaper strategy} must satisfy the
\defn{escaper-start constraint}:
\begin{quote}
for each $i \in \{1, 2, \dots, n_p\}$,
all paths $H_i(\pursuer)$ (over all pursuer motion paths~$\pursuer$)
must start at a common point $H_i(\pursuer)(0)$.
\end{quote}

\paragraph{Win condition.}
We model the escaper player's natural goal of maximizing the number of escapers
that escape, i.e., reach an exit sufficiently far from any pursuer.
Thus we define winning relative to an integer goal $g \in [1, n_\escaper]$
for the number of escapers that escape.

Given escaper motion paths $\escaper_1, \escaper_2, \dots, \escaper_{n_\escaper}$ and
pursuer motion paths $\pursuer_1, \pursuer_2, \dots, \pursuer_{n_\escaper}$,
we say that escaper $i$ \defn{escapes by~$\epsilon$} if,
for some time~$t$, $\escaper_i(t)$ is on an exit and,
for all $j \in \{1, 2, \dots, n_\pursuer\}$,
$\pursuer_j(t)$ is at least $\epsilon$ away from $\escaper_i(t)$ in the pursuer metric.

A pursuer strategy $\Pursuer$ \defn{wins $G_\epsilon$} if,
for all escaper motion paths $\escaper_1, \escaper_2, \dots, \escaper_{n_\escaper}$,
the resulting pursuer motion paths
$Z_1(\cdots), Z_2(\cdots), \dots, Z_{n_\pursuer}(\cdots)$
let $< g$ escapers to escape.
An escaper strategy $\Escaper$ \defn{wins $G_\epsilon$} if,
for all pursuer motion paths $\pursuer_1, \pursuer_2, \dots, \pursuer_{n_\pursuer}$,
the resulting escaper motion paths
$H_1(\cdots), H_2(\cdots), \dots, H_{n_\escaper}(\cdots)$
let $\geq g$ escapers to escape.
As before, a pursuer strategy \defn{wins $G$}
if it wins $G_\epsilon$ for all $\epsilon > 0$, and
an escaper strategy \defn{wins $G$} if it wins $G_\epsilon$
for some $\epsilon > 0$.

By straightforward extensions of the previous proofs, we can show that
exactly one player wins any instance of game~$G$.

\subsection{Multiple Escapers}
\label{sec:multi-human}

In this section, we give simple strategies that narrow the interesting cases
for multiple escapers.
First we show that we can restrict to the goal of $g=1$ escaper escaping
(perhaps to call for help).

\begin{proposition}
Every escaper can escape in a game with multiple escapers if and only if the single escaper could escape in the same game with only one escaper.
\end{proposition}

\begin{proof} If one escaper can escape in a game with only one escaper, all the
escapers can stay together, moving as one escaper would to escape. If the
pursuers can keep a lone escaper from escaping, they can ignore all but one of the
escapers and keep that escaper from escaping. \looseness=-1
\end{proof}

Next we identify some simple scenarios where multiple escapers can always win
(with $g=1$).

\begin{proposition}
If there is only one escaper domain, the cardinality of the exit set is at least $n_\escaper$, and if escapers outnumber pursuers, then one escaper can always escape.
\end{proposition}

\begin{proof}
Each of the escapers can stand at a distinct point in the exit locations.
At at least $n_\escaper-n_\pursuer$ of those spots, there is no pursuer, so the escapers at those spots escape.
\end{proof}

\subsection{Approximation Algorithms}

In this section, we describe some simple extensions of our approximation
algorithms to the case of multiple escapers and/or pursuers.

First, Theorem~\ref{pseudoPTAS} still gives a pseudopolynomial approximation scheme if there are multiple (but $O(1)$) escapers and/or pursuers. The proof is essentially the same: we can solve a discrete game with $O(1)$ pursuers, and the critical speed ratio is bounded above by the critical speed ratio for one pursuer.

Second, the $O(1)$-approximation algorithm from Section~\ref{sec:O(1)}
seems more difficult to generalize.
One approach is to restrict to a pursuer strategy where the pursuers divide up regions to guard and then individually follow a strategy akin to the one used in Section~\ref{sec:O(1)}.
One side of Theorem~\ref{ConstantApproximationTheorem} has an analogue:

\begin{corollary}
	\label{cor:multi-bound}
Consider the game where a polygon $P$ is designated as the only escaper domain of capacity $n_\escaper=1$, and there is a single pursuer domain of capacity $n_\pursuer$ in the moat or exterior model.
Consider partitions of the boundary of $P$ into $n_z$ (not necessarily connected) regions $\mathcal R = \{R_1, R_2, \dots, R_{n_z}\}$.
The pursuers win if their speed is at least
$$
\min_{\text{partition }\mathcal R}
\left(
10.89898
~
\max_{p,q\text{ in same region }R_i \in \mathcal R}
~
\frac{d_\pursuer(p,q)}{d_\escaper(p,q)}
\right).$$
\end{corollary}

\begin{proof}
Each pursuer can ignore all of the boundary but the part assigned to it and use the strategy of Theorem~\ref{ConstantApproximationTheorem}.
\end{proof}

However, for the other side we have no analogue.
Does there exist $c > 0$ such that, for the game described in Corollary~\ref{cor:multi-bound}, the escaper wins if the pursuers' speed is less than the minimum over partitions of the boundary into (not necessarily connected) regions of
$$c \cdot \max_{p,q\text{ in same region}} \frac{d_\pursuer(p,q)}{d_\escaper(p,q)}?$$
We leave this question as an open problem.

\subsection{Slow Pursuers}
\label{sec:slow-pursuers}

In this section, we prove some simple results about pursuers running slower than or equal speed to the escapers, i.e., the speed ratio $r \leq 1$.
Assume the polygon model (exterior or moat).
First we show that the escaper always wins for $r < 1$:

\begin{proposition}
	\label{thm:slow zombies lose}
	For finitely many pursuers whose speed is strictly less than the escaper's ($r < 1$), the escaper wins in a polygonal domain $P$.
\end{proposition}
\begin{proof}
	The intuition is as follows. When close to an edge, the escaper can outrun a single pursuer and escape. Thus there must be other pursuers nearby to catch the escaper. However, how close they need to be depends on how close the escaper is to the edge, and thus the escaper can force the pursuers to guard an arbitrarily small portion of an edge. Once clustered the escaper can outrun the whole group and escape. We now formalize such a strategy and show there is always a region of the polygon in which it can be executed.
	
	First we describe the escaper strategy. Let $R$ be a $\delta\times\Delta$ rectangle (i) contained in $P$, and (ii) whose edge of length $\Delta$ is contained by the longest edge $e$ of $P$.
	We determine $\Delta$ later as a function of $\delta$.
	We choose $\delta$ to be small enough to satisfy properties (i) and (ii).
	Without loss of generality, $e$ is horizontal and the interior of $P$ is above $e$.
	We define some points $\{u_1, \ldots, u_{m+1}\}$ of interest on the upper edge of $R$ called \defn{threat points}.
	Make $u_1$ (respectively, $u_{m+1}$) the upper left (respectively, upper right) corner of $R$ and place the remaining $u_i$ so that the distance between consecutive points is the same.
	For each threat point $u_i$, we denote by $u_i'$ its vertical projection on $e$.
	The escaper starts at the upper left corner of $R$ and will move to the right at full speed.
	At each threat point $u_i$, the escaper checks whether they can win by running to $u_i'$ at full speed.
  We show that this will be the case for at least one of the threat points, thus the escaper wins.

	The main idea is that $R$ is chosen so that if a pursuer can guard the vertical projection of a threat point in time to prevent the victory of the escaper at that point, they cannot reach any of the subsequent projections of threat points in time. Then, each pursuer can only prevent the victory of the escaper at a single threat point. Because there are $m+1$ threat points, the escaper wins.
	We proceed with the details.
	At a threat point $u_i$, the escaper can win by running at $u_i'$ if there are no pursuers within $r \delta$ distance from $u_i'$. 
	We make the distance between consecutive threat point $d = \frac{2r\delta+\eps}{1-r}$ for some $\eps>0$ determined later, so that, while the escaper travels $d$, the distance traveled by pursuers is $d - 2 r \delta - \eps$.
	Then, if a pursuer is guarding $u_i'$ when the escaper is at $u_i$, it will be at least $\eps$ away from the disk centered at $u_{i+1}'$ with radius $r \delta$ when the escaper is at $u_{i+1}$.
	Since the escaper runs to the right at full speed until they can win, such pursuer can never catch up.
	By definition, $\Delta = (m+1)\frac{2r\delta+\eps}{1-r}$.
	We can choose $\eps = \frac{\|e\|}{10 m}$, so that we can choose $\delta$ small enough so that $\Delta < \|e\|/2$ and properties (i) and (ii) are satisfied. 
\end{proof}

Next we consider $r=1$ where the pursuers and escaper have equal speeds.
In the case of one pursuer, the escaper can always win by shortcutting across
a convex vertex.
But multiple pursuers can win in some cases:

\begin{proposition}
\label{ConvexRegionsProp}
If $r=1$, and the exterior of the polygon can be divided into $n_\pursuer$ convex regions that cover the boundary of the polygon, then the pursuers can win in the exterior model.
\end{proposition}

\begin{proof}
Each pursuer can stay in one region, staying at the closest point in that region to the current escaper position (satisfying the nonbranching-lookahead constraint). The closest point in a convex region to the escaper cannot move faster than the escaper can, so the pursuers can keep up with this strategy (speed-limit constraint). If the escaper reaches the boundary, there is a pursuer region containing that boundary, and therefore a pursuer at the closest point in that region to the escaper, which is the escaper's location itself. So, the escaper cannot escape.
\end{proof}

\begin{corollary}
If $r=1$, the escaper domain is a polygon $P$ with $n$ vertices, and $n_\pursuer=n$, then pursuers can win.
\end{corollary}

There is no lower bound analogous to Proposition~\ref{ConvexRegionsProp} because 4 pursuers suffice to guard polygons like the one in Figure~\ref{SlowlyGuardableFigure} with arbitrarily many vertices. Two pursuers can stay on the top and two on the bottom; each of those can be assigned to guard every other triangular region of the convex hull outside $P$.

\begin{figure}[h]
	\centering
	\includegraphics[scale=0.3]{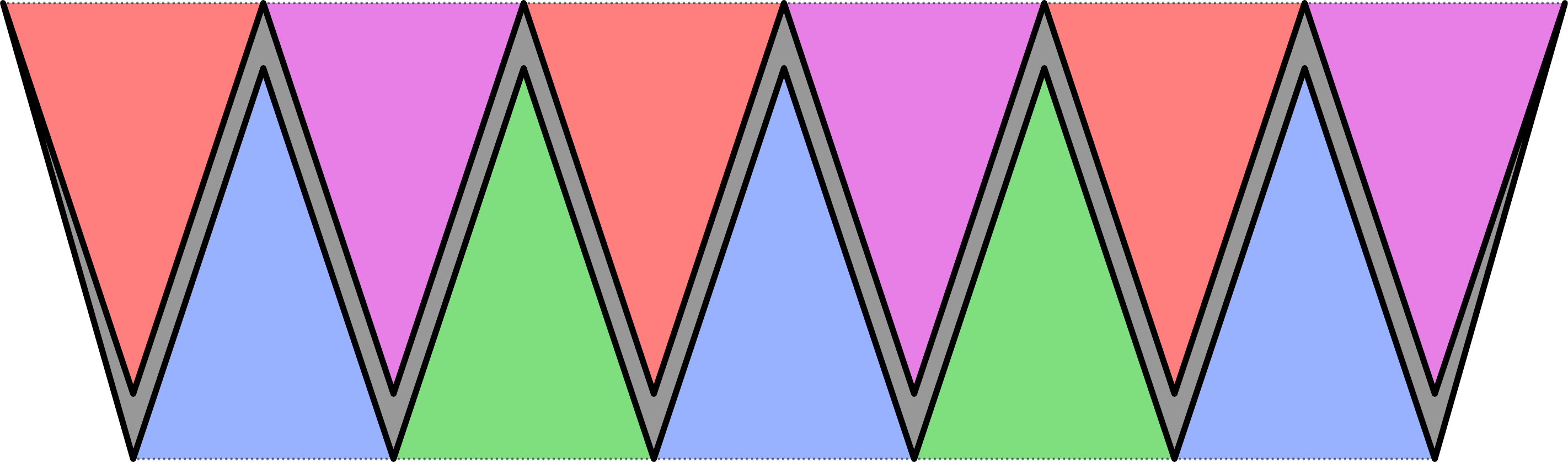}
	\caption{A polygon guardable by 4 pursuers with speed equal to the escaper's in the exterior model. Colored regions denote the (disconnected) region assigned to each pursuer to guard.}
	\label{SlowlyGuardableFigure}
\end{figure}

For convex polygons, we can win with half as many pursuers:

\begin{proposition}
\label{ConvexEscapeProp}
If $r=1$, the escaper domain is a convex $n$-gon $P$, and  $n_\pursuer=\lceil \frac{n}{2} \rceil$, then the escaper can win.
\end{proposition}

\begin{proof}
The escaper should start at any vertex $\escaper$ on the boundary (escaper-start constraint). Let $\escaper'$ be the point opposite $\escaper$ on $\partial P$, that is, the point for which the pursuer distance from $\escaper$ is maximal. The points $\escaper$ and $\escaper'$ split $\partial P$ into two sections, at least one of which must have at least $\lceil \frac{n}{2} \rceil$ vertices (counting $\escaper$ but not $\escaper'$). 
The escaper should run along that section of perimeter except at a small neighborhood of vertices. 
With this strategy, whenever the escaper is running along an edge there should always be a pursuer at the same position in order to prevent an escaper victory.
Let $\theta$ be the maximum internal angle, and $\alpha$ be the length of the shortest edge of $P$.
We first argue that there should be at least two pursuers in the $\alpha \over 4$-neighborhood of $\escaper$ at the start to prevent an escaper win.
If not, the escaper can follow the same strategy as the wedge case (Theorem~\ref{WedgeTheorem}) with a small enough $\eps$ so that the length of the escaper path is at most $\alpha \over 16$ guaranteeing a separation of at least $\alpha \over 8$ from any pursuer not initially close to $\escaper$.
We now describe the escaper strategy at an $\alpha \over 4$-neighborhood of a vertex $v$ (along the chosen section of the perimeter) incident to edges $e_1$ and $e_2$. 
Let $p_1$ and $p_2$ be the points obtained by the intersection of a circle centered at $v$ with radius $\alpha \over 4$ with $e_1$ and $e_2$ respectively.
When the escaper reaches $p_1$, go directly to $p_2$ and continue traversing $e_2$.
At the moment the escaper is at $p_1$, if the only pursuers within $\alpha \over 2$ of $p_2$ (in pursuer metrics) are at $p_1$, the escaper wins by reaching $p_2$ while being at least $2\alpha(1-\sin\frac{\theta}{2})>0$ away from any pursuer.
Otherwise, there is at least one new pursuer (one that was not at $p_1$ with the escaper) that must follow the escaper in its traversal of $e_2$. 
Then the pursuers that were following the escaper in $e_1$ will be behind the escaper and will not be able to be ahead of the escaper again because they do not have time to run around past $\escaper'$ before the escaper gets there.
For each of the $\lceil \frac{n}{2} \rceil-1$ vertices, there must be at least one new pursuer guard to prevent an escaper victory.
With the initial 2 pursuers, $\lceil \frac{n}{2} \rceil+1$ pursuers are necessary to prevent an escaper win.
At all moments the escaper speed is 1 (speed-limit constraint) and, apart from the application of Theorem~\ref{WedgeTheorem}, the escaper path does not depend on pursuer position at all (nonbranching-lookahead constraint).
\end{proof}

Although Proposition~\ref{ConvexEscapeProp} is
true for both the moat and exterior models,
we can make a slightly
stronger statement in the moat model using the same proof. 

\begin{corollary}
\label{MoatConvexEscapeCor}
In the moat model, if $P$ is a polygon with $c$ \emph{convex} vertices, then the escaper can escape from
$\lceil \frac{c}{2} \rceil $ pursuers of the same speed as theirs.

\end{corollary}

\subsection{Hardness Results}
\label{HardnessSection}

In this section, we prove PSPACE-hardness and hardness of approximation results, as specified in Table~\ref{ComplexityTable}, for problems of escaping from pursuers with various combinations of parameters.
All results are for 1-dimensional domains (graph model).
In Table~\ref{ComplexityTable}, the ``Domain'' column describes whether there is an additional constraint to the domains: 
\begin{itemize}
	\item Planar: each domain is a tree, they pairwise intersect only at leaves, and the union of all domains is the embedding of a planar graph;
	\item Connected: there is a single escaper domain and a single pursuer domain.
\end{itemize}

\begin{theorem}
	\label{thm:PSPACE}
	Consider a multi-escaper/pursuer game with $g=1$. It is PSPACE-hard to decide whether pursuers has a winning strategy even if each domain is a tree, they pairwise intersect only at leaves, all leaves are exits, and the union of all domains is the embedding of a planar graph.
\end{theorem}

\begin{proof}
  Our reduction is from \defn{Nondeterministic Constraint Logic (NCL)} \cite{PSPACE-book}.
  An instance of NCL is given by a planar cubic weighted graph $G_\NCL$
  (called a \defn{constraint graph})
  where each edge has either weight 1 (called \defn{red})
  or weight 2 (called \defn{blue}).
  Each vertex is either incident to a single blue and two red edges
  (called an \defn{AND} vertex), or incident to three blue edges
  (called an \defn{OR} vertex).
  A \defn{configuration} of the constraint graph 
  is an orientation (specifying a direction for each edge)
  satisfying that every vertex has incoming edges of total weight
  at least~$2$ (the \defn{inflow constraint}).
  Given a configuration, a \defn{move} flips the orientation of one edge
  in such a way that results in another configuration
  (i.e., satisfying the inflow constraint).
  The reachable configurations remain the same in \defn{asynchronous NCL}
  where we allow partial orientations (some undirected edges),
  where an undirected edge does not count as incoming
  at either endpoint, and allow a move to transform an oriented edge
  into an unoriented one or vice versa
  (while still satisfying the inflow constraint)
  \cite{Viglietta-2013}.
  Given a planar constraint graph, a configuration of that graph,
  and an edge~$e_{\text{out}}$,
  it is PSPACE-complete to decide whether there is a sequence of moves
  that flips $e_{\text{out}}$ at the end \cite{PSPACE-book}.
  The number of moves is less than $2^{|E(G_\NCL)|}$
  because this upper bounds the number of states
  ($3$ possible orientations for each edge).

  The PSPACE-hardness reduction for NCL can be modified
  to have two degree-$1$ vertices
  $v_{\text{in}}$ and $v_{\text{out}}$
  with no constraint on their incoming weights,
  one blue edge $e_{\text{in}}$ initially pointing toward $v_{\text{in}}$
  and another blue edge $e_{\text{out}}$ initially pointing away from
  $v_{\text{out}}$.
  (In fact, a subset of the reduction given in \cite[Section~5.2]{PSPACE-book}
  works exactly this way,
  where $e_{\text{in}}$ is the leftmost \textsf{try in} edge
  and $e_{\text{out}}$ is the leftmost \textsf{try out} edge.
  The reduction then adds a free edge terminator gadget to each of
  these edges, and we can simply not add these gadgets.)
  Furthermore, it is PSPACE-complete to decide whether,
  for some configuration of the constraint graph with $e_{\text{in}}$
  directed toward $v_{\text{in}}$, there is a sequence of moves
  that flips $e_{\text{out}}$ to point toward $v_{\text{out}}$.
  This claim follows from the same reduction, because
  \cite[Lemma~5.8]{PSPACE-book} tells us that edge $e_{\text{in}}$
  initially pointing out from the construction (toward $v_{\text{in}}$)
  forces the entire configuration to reset.
  Furthermore, $v_{\text{in}}$ and $v_{\text{out}}$ are on the same face
  of a planar embedding of $G_\NCL$.

	\begin{figure}[h]
		\centering
		\subcaptionbox{OR gadget}{\includegraphics[page=2,scale=0.7]{figures/PSPACE}}\hfill
		\subcaptionbox{AND gadget}{\includegraphics[page=1,scale=0.7]{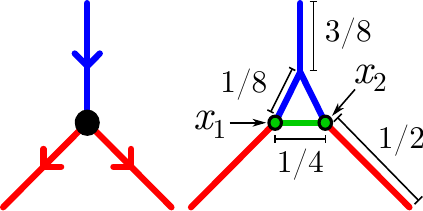}}\hfill
		\subcaptionbox{WIN gadget}{\includegraphics[page=3,scale=0.6]{figures/PSPACE}}
		\caption{Gadgets that simulate a local NCL picture (left)
      with red and blue pursuer domains and green escaper domains (right).
      An edge drawn with only one endpoint represents exactly
      one half of that edge.  (The other half is represented by the
      gadget on the other end of the edge.)}
		\label{fig:pspace}
	\end{figure}

  We build a game with the goal of $g=1$ escaper escaping
  and a speed ratio of $r=1$.
  Refer to Figure~\ref{fig:pspace}.
  Given a planar constraint graph $G_\NCL$ with distinguished edges
  $e_{\text{in}}, e_{\text{out}}$ and vertices
  $v_{\text{in}}, v_{\text{out}}$ as described above,
  we build domains as follows. 
  Every vertex of the constraint graph will be represented by a tree
  escaper domain (colored green in the figures) of capacity~$1$.
  Every edge of the constraint graph will be represented by a tree
  pursuer domain (colored red or blue in the figures to match the $G_\NCL$ edge)
  of capacity~$1$.
  We will describe each edge as the joining of two ``half edges'',
  with one half defined by each endpoint.
  \begin{itemize}
  \item For each OR vertex (Figure~\ref{fig:pspace}(a)),
  the pursuer domain corresponding to each half edge
  is a curve of length $1/2$,
  all incident to a common point~$x$; and
  the corresponding escaper domain is the single point~$x$,
  which is also an exit location.
  This escaper forces some pursuer to block the exit $x$ at all times,
  implementing the OR constraint.
  \item For each AND vertex (Figure~\ref{fig:pspace}(b)),
  the pursuer domain corresponding to each red half edge
  is a curve of length~$1/2$,
  with distinct endpoints $x_1,x_2$ respectively;
  the corresponding escaper domain is a curve of length $1/4$ between
  those endpoints $x_1,x_2$, which are exit locations;
  and the pursuer domain corresponding to the blue half edge is a Y
  with leaf curves of length $1/8$ incident to $x_1,x_2$,
  and a curve of length $3/8$ connecting to the other half of the edge.
  Thus the distance between $x_1$ and $x_2$ is $1/4$ in both the
  escaper domain and the blue pursuer domain,
  so one pursuer in the blue pursuer domain can successfully prevent escape
  (matching the motion of the escaper),
  as can one pursuer in each of the red pursuer domains (staying at $x_1$
  and $x_2$), implementing the AND constraint.
  Also, the pursuer has a distance of $1/2$ from one endpoint
  to the other half edge, as with the curves implementing all other half edges.
  \end{itemize}
  Thus, the escapers can force the pursuers to satisfy the inflow constraint
  at every AND and OR vertex.
  Conversely, the pursuers can make a valid NCL move in unit time
  by moving a pursuer from one end of the edge's pursuer domain
  to the other end.
  \begin{itemize}
  \item For the special vertices $v_{\text{in}}$ and $v_{\text{out}}$
    (Figure~\ref{fig:pspace}(c)),
    the pursuer domain corresponding to each incident half edge
    $e_{\text{in}}$ and $e_{\text{out}}$
    is a curve of length $1/2$,
    with endpoints $x_{\text{in}}$ and $x_{\text{out}}$ respectively,
    both of which are exit locations;
    and we create one escaper domain for both vertices,
    a curve of length $2^{|E(G_\NCL)|}$
    connecting $x_{\text{in}}$ and $x_{\text{out}}$.
    Because $v_{\text{in}}$ and $v_{\text{out}}$ are on a common face of
    $G_\NCL$, this connection preserves planarity.
  \end{itemize}
  To realize this construction in the plane, we scale down the planar embedding
  of $G_\NCL$ to the point where all edges have length at most~$1$,
  and then we wiggle the paths to have the specified lengths.

	Set $n_\escaper=|V(G_\NCL)|-1$ (the number of escaper domains)
  and $n_\pursuer=|E(G_\NCL)|$ (the number of pursuer domains).
	By the Pigeonhole Principle, each domain contains exactly one individual.

  Now suppose that the NCL instance has a solution:
  an initial configuration where $e_{\text{in}}$ points toward $v_{\text{in}}$,
  and a sequence of less than $2^{|E(G_\NCL)|}$ moves that ends with
  flipping edge $e_{\text{out}}$ toward $v_{\text{out}}$.
  Then the pursuer has the following winning strategy,
  parameterized by the location $t$ of the escaper along the
  length-$2^{|E(G_\NCL)|}$ curve from $x_{\text{in}}$ to $x_{\text{out}}$.
  At $t=0$, the pursuers are at the ends of their pursuer domains
  corresponding to the initial configuration.
  Between each integer $t-1$ and $t$,
  one pursuer moves from one end of its pursuer domain
  to the other in unit time, corresponding to the $t$th move in the sequence.
  (Once $t$ is beyond the number of moves in the sequence,
  the pursuer does nothing.)
  Throughout, whenever an AND vertex has an inward-directed blue edge,
  the pursuer assigned to that end tracks the motion of the escaper.
  Because the sequence of configurations satisfies the inflow constraint,
  the escapers cannot win, including at $t=2^{|E(G_\NCL)|}$
  when a pursuer from the pursuer domain corresponding to $e_{\text{out}}$
  has reached $x_{\text{out}}$.

  Conversely, suppose that the NCL instance has no solution.
  Then the escaper has the following winning strategy.
  The escapers at AND and OR gadgets enforce the inflow constraints.
  The escaper along the length-$2^{|E(G_\NCL)|}$ curve
  starts at $x_{\text{in}}$ and runs at full speed to $x_{\text{out}}$.
  This forces the pursuing player to start with a pursuer at $x_{\text{in}}$.
  At all times, we can construct a corresponding configuration of $G_\NCL$,
  where an edge is directed toward a vertex if the corresponding pursuer is
  at the end of the domain corresponding to that vertex,
  and undirected if the pursuer is in the middle.
  Thus we start at a configuration where $e_{\text{in}}$ is directed toward
  $v_{\text{in}}$, and follow moves according to asynchoronous NCL.
  By supposition, we cannot reach a configuration where $e_{\text{out}}$ is
  directed toward $v_{\text{out}}$, so the corresponding pursuer cannot reach
  $x_{\text{out}}$ (being pinned at the other end).
  Thus the escaper reaches exit $x_{\text{out}}$ and wins.
\end{proof}

\begin{theorem}
\label{PlanarHardnessTheorem}
Consider a multi-escaper/pursuer game in the graph model with $g=1$. 
It is NP-hard to distinguish a critical speed ratio of $0$ from $\infty$, even if each domain is a tree, they pairwise intersect only at leaves, and the union of all domains is the embedding of a planar graph.
\end{theorem}

\begin{proof}
We reduce from the Planar Vertex Cover problem of finding a set of at most $k$ vertices in a planar graph such that every edge contains at least one of them, which Lichtenstein \cite{PlanarVertexCoverHard} shows to be NP-hard. 
Given an instance of Planar Vertex Cover consisting of a planar graph $G_{VC}$ and a target number of vertices $k$,
we build a game with $n_\escaper=k$ and $n_\pursuer=|E(G_{VC})|-1$.
Subdivide each edge with a point pursuer domain of capacity 1 marked as an exit.
This splits $G_{VC}$ into $|V(G_{VC})|$ components,
each containing a vertex of $G_{VC}$ and its incident half edges.
Define each such component to be an escaper domain of capacity 1. 

If there is a vertex cover of size at most $k$, then the escapers can start at the corresponding $k$ vertices (escaper-start constraint).
Then the pursuing player places the $|E(G_{VC})|-1$ pursuers,
so there is at least one edge that no pursuer starts on, and an escaper who starts at a vertex incident to that edge can escape by that edge.
The escaper strategy depends only on the pursuer's initial positions (nonbranching-lookahead constraint).

Now consider the pursuer strategy that initially checks whether there is an exit location/pursuer domain incident to escaper domains with no escapers, and if so, places a pursuer at all other locations. 
This pursuer strategy depends only on escaper's initial positions (nonbranching-lookahead constraint).
The escaping player loses if the initial escaper placement do not correspond to a vertex cover.
Because $r$ is irrelevant to the proof, it is NP-hard to distinguish a critical speed ratio of $0$ from $\infty$.
\end{proof}

\begin{theorem}
\label{OneHumanHardnessTheorem}
Consider a multi-pursuer game in the graph model with $n_\escaper=1$. 
It is NP-hard to approximate the critical speed ratio $r$ to within a factor of $2$, even when there is a single escaper domain and a single pursuer domain. 
\end{theorem}

\begin{proof}
We reduce from the Vertex Cover problem of finding a set of at most $k$ vertices in a graph $G$ such that every edge contains at least one of them, which is one of Karp's original 21 NP-hard problems (from \cite{Karp72}).
First we reduce to the special of Vertex Cover where the graph is
guaranteed to be connected; refer to Figure~\ref{fig:connected-vertex-cover}.
Given an instance $(G,k)$ of vertex cover,
where graph $G$ has connected components $C_1, C_2, \dots, C_k$,
we add a new ``apex'' vertex $a$ with incident edges to one arbitrarily chosen
vertex in each $C_i$ as well as a new degree-$1$ vertex $\ell$.
Any vertex cover in the new graph $G'$ includes either $a$ or $\ell$,
and if it includes $\ell$, we can replace it with~$a$,
which covers the incident added edges.
Thus $G'$ has a vertex cover of size $k+1$ if and only if
$G$ has a vertex cover of size~$k$.

\begin{figure}[h]
  \centering
  \includegraphics[scale=0.6]{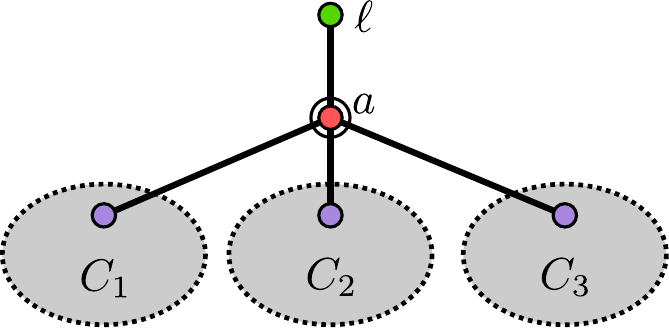}
  \caption{Reduction from Vertex Cover to Vertex Cover on connected graphs.}
  \label{fig:connected-vertex-cover}
\end{figure}

Given an instance of Vertex Cover consisting of a connected graph $G_{VC}$ and a target number of vertices $k$, we make a multi-pursuer game with $n_\escaper=1$, $n_\pursuer=k$, and domains as shown in Figure~\ref{OneHumanHardFigure}.
The pursuer domain realizes the vertex--edge incident graph of $G_{VC}$,
with a node $x_v$ for each vertex $v$ of $G_{VC}$,
a node $x_e$ for each vertex $e$ of $G_{VC}$,
and a length-$1$ curve between two nodes $x_v,x_e$
that correspond to an incident vertex $v$ and edge $e$ of $G_{VC}$.
The escaper domain is a star centered at a point $\escaper$,
with leaves at the nodes $x_e$ corresponding to edges $e$ of $G_{VC}$,
each connected by a curve of length $1$ to~$\escaper$.
The exit points are the leaves of the star, i.e.,
the nodes $x_e$ corresponding to edges $e$ of $G_{VC}$.

\begin{figure}
\centering
\includegraphics[scale=0.25]{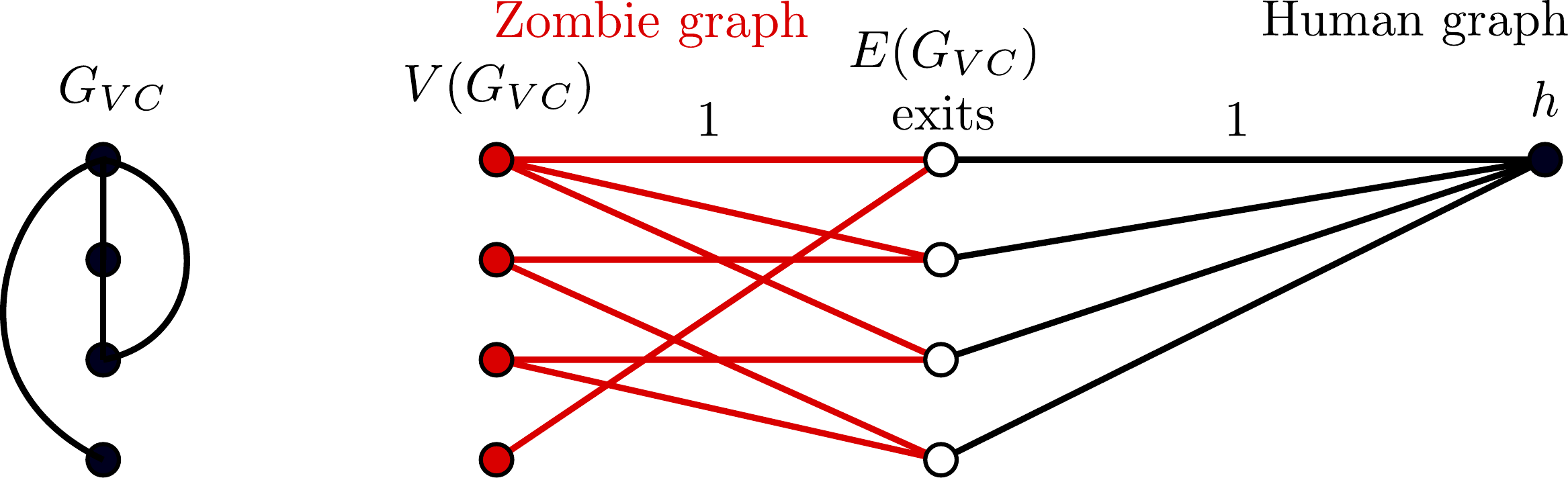}
\caption{A graph with one escaper for which it is NP-hard to determine the critical speed ratio.}
\label{OneHumanHardFigure}
\end{figure}

If there is a vertex cover, then
the following is a winning pursuer strategy for $r\ge 1$.
Assign each pursuer to a vertex in the cover set.
Suppose that the escaper is currently on an edge $x_e \escaper$
of the star escaper domain.
(If the escaper is at the center $\escaper$ of the star,
we consider it to be on the lexically first edge $e_0$.)
Let $t$ be the distance of the escaper from $x_e$.
Let $w$ be the lexically first vertex that covers $e$.
Then we place the pursuer assigned to $w$ on the edge $x_w x_e$,
at distance $t$ away from $x_e$,
while all other pursuers remain at their assigned vertices.
Thus, whenever the escaper reaches an exit~$x_e$ ($t=0$),
a pursuer will be at the same exit.
This strategy depends only on the current escaper position (nonbranching-lookahead constraint) and requires that pursuers run at most at unit speed (speed-limit constraint).

If there is no vertex cover, then
the following is a winning escaper strategy for $r< 2$.
The escaper starts at $\escaper$ (escaper-start constraint).
Wherever the pursuer player initially places the pursuers,
there is an exit that no pursuer is within distance 2 of:
to be within distance 2 of an exit $x_e$, a pursuer must be within distance 1
of a vertex node $x_v$ where $v$ is incident to $e$;
and the regions within distance 1 of each vertex node $x_v$ are disjoint;
so if there were a pursuer within distance 2 of every exit $x_e$,
that would give a vertex cover.
The escaper then runs at full speed to that exit,
and at the moment the exit is reached,
the nearest pursuer is at least $2-r$ away by the speed-limit constraint.
This strategy depends only on the initial pursuer positions
(nonbranching-lookahead constraint).

Therefore it is NP-hard to distinguish a critical speed ratio of at most 1 from one at least 2, as claimed.
\end{proof}

\section{Open Problems}
\label{sec:open}

We conclude with several interesting open problems raised by this research:

\begin{enumerate}
\parindent=1.5em

\item Is the pursuit--escape game (with one pursuer and one evader)
  NP-hard for a 2D polygon?

\item We conjecture that our approximation algorithms of Section~\ref{sec:O(1)}
  and \Section~\ref{sec:pseudoPTAS} generalize to apply in 3D as well,
  with a slightly worse constant in the case of Section~\ref{sec:O(1)}.
  This would nicely complement our 3D NP-hardness result of
  Section~\ref{sec:NPhard}.

\item \Section~\ref{sec:pseudoPTAS} gives a pseudopolynomial-time approximation for the critical speed ratio for a polygon. Is this the best one can do, or is there an approximation scheme whose time depends polynomially only on the length of the description of $P$, or also on $\log \frac1\epsilon$?
Related, we conjecture we can generalize this approximation scheme to apply to
nonpolygonal shapes, such as constant-degree splines (which would include the disk).

\item Can we determine the exact critical speed ratio for regular $n$-gons
  for $n > 4$?
  Our pursuer strategies for equilateral triangle (\Section~\ref{sec:triangle})
  and square (\Section~\ref{sec:square}) generalize naturally,
  but we have been unable to find matching escaper strategies,
  suggesting these may not be tight.

\item Is there an analogue of Theorem~\ref{ConstantApproximationTheorem} describing the critical speed ratio to within a constant factor when there are \emph{two} (or $O(1)$) pursuers?

The most obvious analogue, using a 2nd-order Voronoi diagram, does not work: if $P$ is a long, thin rectangle with one long side subdivided, one pursuer should stay on each side, but a 2nd-order Voronoi diagram might put both pursuers on one side.

The other most obvious analogue would have one pursuer attempts to guard the edge the escaper is closest to, the second pursuer greedily guards whatever point the first pursuer would have the most trouble reaching, and both pursuers delay changing their strategies by the use of fringe regions as in Theorem~\ref{ConstantApproximationTheorem}, but the escaper might exit multiple fringes simultaneously, which seems hard for the pursuers to account for without paying an extra factor equal to the number of pursuers.

\item Can we characterize the exact number of pursuers required to win
  in a polygon, under the exterior or moat model, when the speed ratio $r=1$?
  Section~\ref{MultipleZombiesSection} gives a few sufficient conditions
  and an interesting example.

\item Our PSPACE-hardness result for multiple pursuers
  (Theorem~\ref{thm:PSPACE}) requires one edge of exponential length.
  Is the problem strongly PSPACE-hard, i.e., even when all edge lengths
  are polynomial integers?
  Is the problem in PSPACE?

\item Can we adapt our model to \emph{capturing pursuers}, where an escaper
  loses if it is ever within $\epsilon$ of a pursuer (for arbitrarily
  small $\epsilon > 0$)?  This more natural model should not affect our main
  domains of polygons or Jordan regions, where an escaper can walk near the
  boundary instead of on it.  However, in the general setting considered in
  \Section~\ref{appendix:model}, it becomes more difficult to prove every game has
  a unique winner; in particular, our discrete model needs adaptation to avoid
  accidental captures.  We conjecture that this is possible.

  We believe we can prove many more hardness results in this model.  In particular,
  we believe the 3D one-pursuer one-escaper problem becomes EXPTIME-hard
  by a modification to the proof of \cite{Reif-Tate-1993}, which would
  strengthen our NP-hardness result (Theorem~\ref{thm:NPhard}).

\item What happens if we restrict pursuer and escaper strategies to be
  continuous functions of their opponent's movement?
  Does this change allow us to define escaper winning without needing
  a uniform $\epsilon$ by which they win?
  (See related results in \cite[Lemma~6 and Theorem~7]{Bollobas-Leader-Walters-2009-arXiv}.)
  Is this a reasonable model, or does it forbid natural strategies?

\end{enumerate}

\xxx{Teams that can't communicate is undecidable?
  cite Team Fortress 2 paper}

\section*{Acknowledgments}

We thank Greg Aloupis and Fae Charlton for helpful early discussions on this
topic.
We also thank anonymous referees for helpful comments,
leading us to formulate precise definitions of the model, and
for giving the proofs of Lemmas~\ref{lem:terrible} and~\ref{lem:choice}.
Supported in part by the NSERC and NSF grant CCF-2348067.

\let\realbibitem=\bibitem
\def\bibitem{\par \vspace{-1.2ex}\realbibitem}

\bibliographystyle{alpha}
\bibliography{zombies}

\appendix
\magicappendix

\section{Intrinsic Metrics of Compact Regions with Finitely Rectifiable Boundaries are Compact}

\begin{lemma} \label{lem:intrinsic metric compact}
  If $R$ is a compact subset of $\mathbb R^k$
  and $R$ is finitely rectifiable,
  then the intrinsic (shortest-path) metric space $M$ induced by $R$ is compact.
\end{lemma}

\begin{proof}
  A metric space is compact if and only if it is \emph{sequentially compact},
  i.e., every infinite sequence $p_1, p_2, \ldots$
  has a \emph{limit point} $p^*$, i.e., a point $p^*$ such that,
  for every $\epsilon > 0$,
  there is a $p_i$ within distance $\epsilon$ of~$p^*$.
  We will prove that $M$ is sequentially compact.
  Consider an infinite sequence $p_1, p_2, \ldots \in R$.
  Because $R$ is compact, we can restrict to an infinite subsequence of $p_i$'s
  that converges (in the Euclidean metric) to a limit point $p^* \in R$.
  We will prove that $p^*$ is a limit point with respect to the
  intrinsic metric as well.

  Each $p_i$ lies on an associated Lipschitz patch of~$R$.
  Because there are finitely many Lipschitz patches
  associated with $R$, we can restrict to an infinite subsequence
  $q_1, q_2, \dots$ of $p_1, p_2, \dots$
  for which all $q_i$'s lie on the same Lipschitz patch~$S$.
  Let $r_i$ be a parameter vector for point $p_i$ on~$S$.
  Because $S$'s domain is compact, the points $r_i$ have a limit point $r^*$
  in $S$'s domain, corresponding to a point $q^*$ on~$S$.
  Because $p_1, p_2, \dots$ converges to its limit~$p^*$,
  the subsequence $q_1, q_2, \dots$ converges to the same limit $p^* = q^*$.

  Because $p_i$ and $q_i$ both converge to $p^* = q^*$ in Euclidean metric,
  $d(p_i,q_i) \to 0$; likewise, because $r_i \to r^*$, $|r_i - r^*| \to 0$.
  Therefore, $d_R(p_i,p^*) \leq d_S(p_i,p^*) \to 0$,
  so $p^*$ is a limit point of the $p_i$'s in the intrinsic metric.
\end{proof}

\end{document}